\documentclass[11pt]{article}
\usepackage{graphicx,amsmath,amsbsy,amssymb,amsthm,citesort,epsfig,epsf,url,comment,algorithm,algorithmic}

\usepackage[usenames]{color}

\newtheorem{thm}{Theorem}[section] 
\newtheorem{lemma}{Lemma}[section] 
\newtheorem{cor}{Corollary}[section]

\newtheorem{definition}{Definition}[section]

\newcommand{\hard}{\mathrm{hard}}

\newcommand{\supp}{\mathrm{supp}}

\newcommand{\argmin}{\mathop{\mathrm{arg \ min}}}

\newcommand{\real}{\mathbb{R}}
\newcommand{\complex}{\mathbb{C}}
\newcommand{\integers}{\mathbb{Z}}
\newcommand{\cF}{\mathcal{F}}
\newcommand{\cM}{\mathcal{M}}
\newcommand{\cB}{\mathcal{B}}
\newcommand{\cT}{\mathcal{T}}
\newcommand{\cE}{\mathcal{E}}
\newcommand{\cI}{\mathcal{I}}
\newcommand{\cX}{\mathcal{X}}

\newcommand{\cS}{\mathcal{S}}
\newcommand{\cP}{\mathcal{P}}
\newcommand{\cR}{\mathcal{R}}
\newcommand{\cN}{\mathcal{N}}

\newcommand{\cPhi}{\Phi}

\newcommand{\bB}{{\boldsymbol{B}}}
\newcommand{\bQ}{{\boldsymbol{Q}}}

\newcommand{\bE}{\boldsymbol{E}}
\newcommand{\bU}{\boldsymbol{U}}
\newcommand{\bP}{\boldsymbol{P}}
\newcommand{\br}{\boldsymbol{r}}
\newcommand{\bx}{\boldsymbol{x}}
\newcommand{\bPsi}{\boldsymbol{\Psi}}
\newcommand{\balpha}{\boldsymbol{\alpha}}
\newcommand{\bA}{\boldsymbol{A}}
\newcommand{\by}{\boldsymbol{y}}
\newcommand{\bR}{\boldsymbol{R}}
\newcommand{\bS}{\boldsymbol{S}}

\newcommand{\be}{\boldsymbol{e}}
\newcommand{\bi}{\boldsymbol{i}}
\newcommand{\bh}{\boldsymbol{h}}
\newcommand{\bq}{\boldsymbol{q}}
\newcommand{\bs}{\boldsymbol{s}}
\newcommand{\bz}{\boldsymbol{z}}
\newcommand{\bw}{\boldsymbol{w}}

\newcommand{\bLambda}{\boldsymbol{\Lambda}}
\newcommand{\blambda}{\boldsymbol{\lambda}}
\newcommand{\bI}{\boldsymbol{I}}
\newcommand{\bzero}{\boldsymbol{0}}

\newcommand{\bnyq}{B_{\mathrm{nyq}}} 
\newcommand{\bband}{B_{\mathrm{band}}} 
\newcommand{\tsamp}{T_{\mathrm{s}}} 
\newcommand{\twind}{T_{\mathrm{w}}} 
\newcommand{\slepw}{W} 

\newcommand{\norm}[2][2]{\left\| #2 \right\|_{#1}}
\newcommand{\iprod}[2]{\left\langle {#1},{#2} \right\rangle}

\newcommand{\expval}[1]{\mathbb{E} \left[ #1 \right]}
\newcommand{\sinc}[1]{\mathop{\mathrm{sinc}}\left( #1 \right)}
\newcommand{\prob}[1]{\mathbb{P} \left[ #1 \right]}

\newcommand{\sampe}[1]{\be_{#1}}

\newlength{\imgwidth}
\title{{\bf Compressive Sensing of Analog Signals Using \\ Discrete Prolate Spheroidal Sequences}}
\author{Mark A.\ Davenport$^s$ and Michael B.\ Wakin$^c$\thanks{
Corresponding author. Phone: (303) 273-3607. Fax: (303) 273-3602.\newline
\indent ~ This work was partially supported by NSF grants DMS-1004718 and CCF-0830320, DARPA grant FA8650-08-C-7853, and AFOSR grant FA9550-09-1-0465.}\\[2mm]
\small $^s$Department of Statistics, Stanford University, \\[-1mm]
\small 390 Serra Mall, Sequoia Hall, Stanford, CA 94305, USA. Email: markad@stanford.edu\\[1mm]
\small $^c$Department of Electrical Engineering and Computer Science, Colorado School of Mines, \\[-1mm]
\small 1500 Illinois St., Golden, CO 80401, USA. Email: mwakin@mines.edu
}

\date{September 2011; revised March 2012}

\begin{document}
\maketitle
\vspace{-0.2in}

\begin{abstract}
Compressive sensing (CS) has recently emerged as a framework for efficiently capturing signals that are sparse or compressible in an appropriate basis.  %
While often motivated as an alternative to Nyquist-rate sampling, there remains a gap between the discrete, finite-dimensional CS framework and the problem of acquiring a continuous-time signal.
In this paper, we attempt to bridge this gap by exploiting the Discrete Prolate Spheroidal Sequences (DPSS's), a collection of functions that trace back to the seminal work by Slepian, Landau, and Pollack on the effects of time-limiting and bandlimiting operations.
DPSS's form a highly efficient basis for sampled bandlimited functions; by modulating and merging DPSS bases, we obtain a dictionary that offers high-quality sparse approximations for most sampled multiband signals.
This multiband modulated DPSS dictionary can be readily incorporated into the CS framework.
We provide theoretical guarantees and practical insight into the use of this dictionary for recovery of sampled multiband signals from compressive measurements.
\end{abstract}

\begin{center}
\end{center}

\section{Introduction}

\subsection{Compressive sensing of analog signals}

In recent decades the digital signal processing community has enjoyed enormous success in developing hardware and algorithms for capturing and extracting information from signals. Capitalizing on the early work of Whitaker, Nyquist, Kotelnikov, and Shannon on the sampling and representation of continuous signals, signal processing has moved from the analog to the digital domain and ridden the wave of Moore's law.  Digitization has enabled the creation of sensing and processing systems that are more robust, flexible, cheaper and, therefore, more ubiquitous than their analog counterparts.

The foundation of this progress has been the Nyquist sampling theorem, which states that in order to perfectly capture the information in an arbitrary continuous-time signal $x(t)$ with bandlimit $\frac{\bnyq}{2}$ Hz, we must sample the signal at its Nyquist rate of $\bnyq$ samples/sec. This requirement has placed a growing burden on analog-to-digital converters as applications that require processing signals of ever-higher bandwidth lead to ever-higher sampling rates. This pushes these devices toward a physical barrier, beyond which their design becomes increasingly difficult and costly~\cite{Walde_Analog}.

In recent years, {\em compressive sensing} (CS) has emerged as a framework that can significantly reduce the acquisition cost at a sensor~\cite{Cande_Compressive,Donoh_Compressed,Baran_Compressive}. CS builds on the work of Cand\`{e}s, Romberg, and Tao~\cite{CandeRT_Robust} and Donoho~\cite{Donoh_Compressed}, who showed that a signal that can be compressed using classical methods such as transform coding can also be efficiently acquired via a small set of nonadaptive, linear, and usually randomized measurements.

There remains, however, a prominent gap between the theoretical framework of CS, which deals with acquiring finite-length, discrete signals that are sparse or compressible in a basis or dictionary, and the problem of acquiring a continuous-time signal.
Previous work has attempted to bridge this gap by employing two very different strategies.
First, works such as~\cite{TroppLDRB_Beyond} have employed a simple {\em multitone} analog signal model that maps naturally into a finite-dimensional sparse model.
Although this assumption allows the reconstruction problem to be formulated directly within the CS framework, the multitone model can be unrealistic for many analog signals of practical interest.
Alternatively, other authors have considered a more plausible {\em multiband} analog signal model that is also amenable to sub-Nyquist sampling~\cite{FengB_Spectrum,BreslF_Spectrum,Feng_Universal,VenkaB_Further,MishaE_From,MishaE_Blind}.
These works, however, have involved customized sampling protocols and reconstruction formulations that fall largely outside of the standard CS framework.
Indeed, some of this body of literature and many of its underlying ideas actually pre-date the very existence of CS.

\subsection{Contributions and paper organization}

In this paper, we bridge this gap in a different manner.
Namely, we show that when dealing with finite-length windows of samples, it is possible to map the multiband analog signal model---in a very natural way---into a finite-dimensional sparse model.
One can then apply many of the standard theoretical tools of CS to develop algorithms for both recovery as well as compressive domain processing of multiband signals.

Our work actually rests on ideas that trace back to the classical signal processing literature and the study of time-frequency localization.
The Weyl-Heisenberg uncertainty principle states that a signal cannot be simultaneously localized on a finite interval in both time and frequency.
A natural question is to what extent it is possible to concentrate a signal $x(t)$ and its continuous-time Fourier transform (CTFT) $X(F)$ {\em near} finite intervals.
In an extraordinary series of papers from the 1960s and 1970s, Slepian, Landau, and Pollack provide an in-depth investigation into this question~\cite{SlepiP_ProlateI,LandaP_ProlateII,LandaP_ProlateIII,Slepi_ProlateIV,Slepi_ProlateV}.
The implications of this body of work have had a tremendous impact across a number of disciplines within mathematics and engineering, particularly in the field of spectral estimation and harmonic analysis (e.g.,~\cite{Thoms_Spectrum}).
Very few of these ideas have appeared in the CS literature, however, and so one goal of this paper is to carefully explain---from a CS perspective---the natural role that these ideas can indeed play in CS.

We begin this paper in Sections~\ref{sec:csthy} and~\ref{sec:background} with a description of our problem setup and a survey of the necessary CS background material.
In Section~\ref{sec:MBDPSS}, we introduce the {\em multitone} and {\em multiband} analog signal models.  We then discuss how sparse representations for multiband signals can be incorporated into the CS framework through the use of {\em Discrete Prolate Spheroidal Sequences} (DPSS's)~\cite{Slepi_ProlateV}.
First described by Slepian in 1978, DPSS's form a highly efficient basis for sampled bandlimited functions.
For the sake of clarity and completeness, we provide a self-contained review of the key results from Slepian's work that are most relevant to the problem of modeling sampled multiband signals.
We then explain how, by modulating and merging DPSS bases, one obtains a dictionary that---to a very high degree of approximation---provides a sparse representation for most finite-length, Nyquist-rate sample vectors arising from multiband analog signals.
We also explain why the qualifiers ``approximation'' and ``most'' in the preceding sentence are necessary; however, we characterize them formally and justify the use of the multiband modulated DPSS dictionary in practical settings.

In Section~\ref{sec:recovery}, we discuss the use of the multiband modulated DPSS dictionary for recovery of sampled multiband signals from compressive measurements.
We discuss the implications (in terms of formulating reconstruction procedures and guaranteeing their performance) of the fact that our dictionary is not quite orthogonal; in fact, it may be undercomplete or overcomplete, depending on the setting of a user-defined parameter.
We also provide theoretical guarantees for recovery algorithms that exploit the {\em block-sparse} nature of signal expansions in our dictionary.
Ultimately, this allows us to guarantee that most finite-length sample vectors arising from multiband analog signals can---to a high degree of approximation---be recovered from a number of measurements that is proportional to the underlying information level (also known as the Landau rate~\cite{Landa_Necessary}).

In Section~\ref{sec:sims}, we present the results of a detailed suite of simulations for signal recovery from compressive measurements, illustrating the effectiveness of our proposed approaches on realistic signals.
We show that the reconstruction quality achieved using the multiband modulated DPSS dictionary is far better than what is achieved using the discrete Fourier transform (DFT) as a sparsifying basis.
These results confirm that a DPSS-based dictionary can provide a very attractive alternative to the DFT for sparse recovery.
We conclude in Section~\ref{sec:conc} with a final discussion and directions for future work.

\subsection{Relation to existing work}

Although customized measurement and reconstruction schemes~\cite{FengB_Spectrum,BreslF_Spectrum,Feng_Universal,VenkaB_Further,MishaE_From,MishaE_Blind} have previously been proposed for efficiently sampling multiband signals, we believe that our paper is of independent interest from these works, specifically because we restrict ourselves to operating within the finite-dimensional CS framework.
There are a variety of plausible CS (and even non-CS) scenarios where a sparse representation of a finite-length Nyquist-rate sample vector would be useful, and it is this problem to which we devote our attention.
This work may be of interest, for example, to any practitioner who has struggled with the lack of sparsity that the DFT dictionary provides even for pure sampled tones at ``off-grid'' frequencies.
Moreover, as we discuss more fully in Section~\ref{sec:csthy}, several analog CS hardware architectures can be viewed as providing random projections of finite-length, Nyquist-rate sample vectors.
Our formulation is compatible with these architectures and does not require a customized sampling protocol.

It is important to mention that we are not the first authors to recognize the potential role that DPSS's (or their continuous-time counterparts, the Prolate Spheroidal Wave Functions, or PSWF's~\cite{SlepiP_ProlateI,LandaP_ProlateII,LandaP_ProlateIII,Slepi_ProlateIV}) can play in CS.
Izu and Lakey~\cite{izu2009time} have drawn an analogy between sampling bounds for multiband signals and classical results in CS, but not specifically for the purpose of using the finite-dimensional CS framework for sparse recovery of sample vectors from multiband analog signals.
Gosse~\cite{gosse2010compressed} has considered the recovery of smooth functions from random samples; however, this work focuses on a very different setting, employing a PSWF (not DPSS) dictionary, considering only baseband signals, and exploiting sparsity in a different way than our work.
Senay et al.~\cite{senay2008compressive,senay2009reconstruction} have also considered a PSWF dictionary for reconstruction of signals from nonuniform samples; however, this work also focuses on baseband signals and lacks formal approximation and CS recovery guarantees.
Oh et al.~\cite{oh2010signal} have employed a modulated DPSS dictionary for sampled bandpass signals; however, this work falls largely outside the standard CS framework and again lacks formal approximation and CS recovery guarantees of the type we provide.
Finally, Sejdi\'{c} et al.~\cite{sejdic2008channel} have proposed a multiband modulated DPSS dictionary very similar to our own and a greedy algorithm for signal decomposition in that dictionary.
However, this work is again not devoted to developing sparse approximation guarantees for sampled multiband signals.
It focuses not on signal recovery but on identification of a communications channel, and the proposed reconstruction algorithm is not intended to exploit block-sparse structure in the signal coefficients.
We hope that our paper will be a valuable addition to this nascent literature and help to encourage much further exploration of the connections between DPSS's, PSWF's, and CS.

\subsection{Preliminaries}

Before proceeding, we first briefly introduce some mathematical notation that we will use throughout the paper. We use bold characters to indicate finite-dimensional vectors and matrices.  All such vectors and matrices are indexed beginning at $0$, so that the first element of a length-$N$ vector $\bx$ is given by $\bx[0]$ and the last by $\bx[N-1]$.
We denote the Hermetian transpose of a matrix $\bA$ by $\bA^H$.  We use $\norm[p]{\cdot}$ to denote the standard $\ell_p$ norm. We also use $\norm[0]{\bx} := |\supp(\bx)|$ to denote the number of nonzeros of $\bx$, and we say that $\bx$ is $S$-sparse if $\norm[0]{\bx} = S$.  We use $\expval{x}$ to denote expected value of a random variable $x$ and $\prob{E}$ to denote the probability of an event $E$. Finally, we adopt the convention $j = \sqrt{-1}$ throughout the paper.

\section{Mapping Analog Sensing to the Digital Domain}
\label{sec:csthy}

In the standard CS setting, one is concerned with recovering a finite-dimensional vector $\bx \in \complex^N$ from a limited number of measurements.  A typical first order assumption is that the vector $\bx$ is {\em sparse}, meaning that there exists some basis or dictionary $\bPsi$ such that $\bx = \bPsi \balpha$ and $\balpha$ has a small number of nonzeros, i.e., $\norm[0]{\balpha} \le S$ for some $S \ll N$.  One then acquires the measurements
\begin{equation} \label{eq:phi}
\by = \bA \bx + \be
\end{equation}
where $\bA \in \complex^{M \times N}$ maps $\bx$ to a length-$M$ vector of complex-valued measurements, and where $\be$ is a length-$M$ vector that represents {\em measurement noise} generated by the acquisition hardware. In the context of CS, one seeks to design $\bA$ so that $M$ is on the order of $S$ (the number of degrees of freedom of the signal) and potentially much smaller than $N$.

In the present work, however, we are concerned with the acquisition of a finite-length window of a complex-valued, continuous-time signal, which we denote by $x(t)$.
Specifically, we suppose that we are interested in a time window of length $\twind$ seconds and that we acquire the measurements
\begin{equation} \label{eq:phi1}
\by = \cPhi(x(t)) + \be,
\end{equation}
where $\cPhi$ is a linear {\em measurement operator} that maps functions defined on $[0,\twind)$ to a length-$M$ vector of measurements and $\be$ again represents measurement noise.
We assume throughout this paper that $x(t)$ is bandlimited with bandlimit $\frac{\bnyq}{2}$ Hz, i.e., that $x(t)$ has a continuous-time Fourier transform (CTFT)
$$
X(F) = \int_{-\infty}^{\infty} x(t)e^{-j2\pi F t} dt
$$
such that $X(F) = 0$ for $|F| > \frac{\bnyq}{2}$. Additional assumptions on $x(t)$ will be specified in Section~\ref{ssec:multiband}.

Because we assume that $x(t)$ is bandlimited and that the measurement process~\eqref{eq:phi1} takes place over a finite window of time, we restrict our attention to the problem of recovering the Nyquist-rate samples of $x(t)$ over this time interval.
Specifically, we let $\tsamp \le \frac{1}{\bnyq}$ denote a sampling interval (in seconds) chosen to meet the minimum Nyquist sampling rate, and we let $x[n]$ denote the infinite-length sequence that would be obtained by uniformly sampling $x(t)$ with sampling period $\tsamp$, i.e., $x[n] = x(n\tsamp)$.
We are interested in a time window of length $\twind$ seconds, during which there are $N = \lceil \frac{\twind}{\tsamp} \rceil$ samples.
We let $\bx = \left[x[0] ~ x[1] ~ \cdots ~ x[N-1]\right]^T$ denote $x[n]$ truncated to the $N$ samples from $0$ to $N-1$.
This paper is specifically devoted to the problem of recovering $\bx$, the vector of Nyquist-rate samples of $x(t)$ on $[0,\twind)$,\footnote{Note that our goal is to recover $\bx$, which of course carries useful information about $x(t)$, but recovering $\bx$ may not be sufficient for exactly recovering $x(t)$ on the entire window $[0,\twind)$. (This depends on the exact sampling rate and the decay of the analog interpolation kernel.) In practice, the methods we describe in this paper for digital single-window reconstruction could be implemented in a streaming multi-window setting, and this would allow for a more accurate reconstruction of $x(t)$ on the entire window.} from compressive measurements $\by$ of the form~\eqref{eq:phi1}.

To facilitate this, we first note that the sensing model in~\eqref{eq:phi1} is clearly very similar to the standard CS model in~\eqref{eq:phi}.
We briefly describe conditions under which these models are equivalent.
Recall from the Shannon-Nyquist sampling theorem that $x(t)$ can be perfectly reconstructed from $x[n]$ since $\frac{1}{\tsamp} \ge \bnyq$. Specifically, we have the formula
\begin{equation} \label{eq:sincinterp}
x(t) = \sum_{n=-\infty}^{\infty} x[n] \sinc{t/\tsamp - n},
\end{equation}
where
$$
\sinc{t} = \begin{cases} \sin(\pi t)/(\pi t), & t \neq 0 \\ 1, & t = 0. \end{cases}
$$
Observe that since $\cPhi$ is linear, we can express each measurement $\by[m]$ in~\eqref{eq:phi1} simply as the inner product between $x(t)$ and some {\em sensing functional} $\phi_m(t)$, i.e.,\footnote{In our setup, since $\cPhi$ maps functions defined on $[0,\twind)$ to vectors in $\complex^M$, we are inherently assuming that $\phi_m(t) = 0$ outside of $[0,\twind)$, so that the sensing functionals are time-limited. Although certain acquisition systems (such as the {\em modulated wideband converter} of~\cite{MishaE_From}) do not satisfy this condition, we believe that it is often a reasonable assumption in practice and that many acquisition systems can at least be well-approximated as time-limited.}
\begin{equation} \label{eq:y1}
\by[m] = \iprod{\phi_m(t)}{x(t)} + \be[m].
\end{equation}
In this case we can use (\ref{eq:sincinterp}) to reduce (\ref{eq:y1}) to
\begin{equation} \label{eq:y2}
\by[m] = \iprod{\phi_m(t)}{\sum_{n=-\infty}^{\infty} x[n] \sinc{t/\tsamp - n}} + \be[m] = \sum_{n=-\infty}^{\infty} x[n] \iprod{\phi_m(t)}{\sinc{t/\tsamp - n}} + \be[m].
\end{equation}
If we let $\bA$ denote the $M \times N$ matrix with entries given by
$$
\bA[m,n] = \iprod{\phi_m(t)}{\sinc{t/\tsamp - n}}
$$
and let $\bw$ denote the length-$M$ vector with entries given by
\begin{equation} \label{eq:wdef}
\bw[m] = \sum_{\overset{n \le -1}{n \ge N}} x[n] \iprod{\phi_m(t)}{\sinc{t/\tsamp - n}},
\end{equation}
then (\ref{eq:phi1}) reduces to
\begin{equation} \label{eq:phi2}
\by = \bA \bx + \bw + \be.
\end{equation}
If the vector $\bw = \bzero$, then~\eqref{eq:phi2} is exactly equivalent to the standard CS sensing model in~\eqref{eq:phi}.  Moreover, if $\bw$ is not zero but is small compared to $\be$, then we can simply absorb $\bw$ into $\be$ and again reduce~\eqref{eq:phi2} to~\eqref{eq:phi}.

A precise statement concerning the size of $\bw$ would depend greatly on the choice of the $\phi_m(t)$.  While a detailed analysis of $\bw$ for various practical choices of $\phi_m(t)$ is beyond the scope of this paper, we briefly mention some possible strategies for controlling $\bw$.  First, one can easily show that if each $\phi_m(t)$ consists of any weighted combination of Dirac delta functions positioned at times $0, \tsamp, \ldots, (N-1)\tsamp$, then by construction $\bw = \bzero$. This should not be surprising, as in this case it is clear that the measurements are simply a linear combination of the Nyquist-rate samples from the finite window.  Importantly, it is possible to collect measurements of this type {\em without first acquiring the Nyquist-rate samples} (see, for example, the architecture proposed in~\cite{SlaviLDB_Compressive}), although there are also plenty of situations in which one might explicitly apply a matrix multiplication to compress data {\em after} acquiring a length-$N$ vector of Nyquist-rate samples.

For many architectures used in practice, it will not be the case that $\bw = \bzero$ exactly.  However, it may still be possible to ensure that $\bw$ remains very small.  There are a number of possible routes to such a guarantee.  For example, the $\phi_m(t)$ could be designed to incorporate a smooth window $g(t)$ so that we effectively sample $x(t) g(t)$ instead of $x(t)$, where $g(t)$ is designed to ensure that $x[n] g[n] \approx 0$ for $n\le -1$ or $n\ge N$.  The reconstruction algorithm could then compensate for the effect of $g[n]$ on $0, 1, \ldots, N-1$.
Alternatively, by considering a slightly oversampled version of $x(t)$ (so that $\frac{1}{\tsamp}$ exceeds $\bnyq$ by some nontrivial amount) it is also possible to replace the $\mathrm{sinc}$ interpolation kernel with one that decays significantly faster, ensuring that the inner products in~\eqref{eq:wdef} decay to zero extremely quickly~\cite{DaubeD_Approximating}.  Finally, as we will see below, many constructions of $\phi_m(t)$ often involve a degree of randomness that could also be leveraged to show that with high probability, the inner products in~\eqref{eq:wdef} decay even faster.  However, since the details will depend greatly on the particular architecture used, we leave such an investigation for future work.

Having argued that the measurement model in~\eqref{eq:phi1} can often be expressed in the form~\eqref{eq:phi}, we now turn to the central theoretical question of this paper:
\begin{quote}
Supposing that $x(t)$ obeys the multiband model described in Section~\ref{ssec:analogmodels}, how can we recover $\bx$, i.e., the Nyquist-rate samples of $x(t)$ on $[0,\twind)$, from compressive measurements of the form $\by = \bA \bx + \be$?
\end{quote}
In order to answer this question, of course, we will need a dictionary $\bPsi$ that provides a suitably sparse representation for $\bx$.
We devote Section~\ref{sec:MBDPSS} to constructing such a dictionary.
In addition to a dictionary $\bPsi$, however, we will also need a reconstruction algorithm that can efficiently recover $\bx$ from the compressive measurements $\by$.
While it is certainly possible to apply out-of-the-box CS recovery algorithms to this problem, there are certain properties of our dictionary that make the recovery problem worthy of further consideration.
(In particular, the columns of our dictionary $\bPsi$ will typically not be orthogonal, and the sparse coefficient vectors $\balpha$ that arise will tend to have structured (block-sparse) sparsity patterns.)
In light of these nuances, Section~\ref{sec:background} now provides additional background on CS that will allow us to formulate a principled recovery technique.

\section{Compressive Sensing Background}
\label{sec:background}

\subsection{Sensing matrix design}
\label{ssec:matdesign}

Setting aside the question of how to design the sparsity-inducing dictionary $\bPsi$, we first address the problem of designing $\bA$.
Although many favorable properties for sensing matrices have been studied in the context of CS, the most common is the {\em restricted isometry property} (RIP)~\cite{CandeT_Decoding}.  We say that the matrix $\bA \bPsi$ satisfies the RIP of order $S$ if there exists a constant $\delta_S \in (0,1)$ such that
\begin{equation}
\label{RIP} \sqrt{1-\delta_S} \le \frac{\norm{\bA \bPsi \balpha}}{\norm{\balpha}} \le \sqrt{1+\delta_S}
\end{equation}
holds for all $\balpha$ such that $\norm[0]{\balpha} \le S$.  In words, $\bA \bPsi$ preserves the norm of $S$-sparse vectors.  Note that for any pair of vectors $\balpha$ and $\balpha'$ such that $\norm[0]{\balpha} = \norm[0]{\balpha'} = S$, we have that $\norm[0]{\balpha-\balpha'} \le 2S$.  This gives us an alternative interpretation of~\eqref{RIP}---namely that the RIP of order $2S$ ensures that $\bA \bPsi$ preserves Euclidean distances between $S$-sparse vectors $\balpha$.

A related concept is what we call the $\bPsi$-RIP (following the notation in~\cite{CandeENR_Compressed}). Specifically, we say that the matrix $\bA$ satisfies the $\bPsi$-RIP of order $S$ if there exists a constant $\delta_S \in (0,1)$ such that
\begin{equation}
\label{Psi-RIP} \sqrt{1-\delta_S} \le \frac{\norm{\bA \bPsi \balpha}}{\norm{\bPsi \balpha}} \le \sqrt{1+\delta_S}
\end{equation}
holds for all $\balpha$ such that $\norm[0]{\balpha} \le S$. When $\bPsi$ is an orthonormal basis,~\eqref{RIP} and~\eqref{Psi-RIP} are equivalent. However, we will be concerned in this paper with non-orthogonal (and even non-square) dictionaries $\bPsi$, in which case the RIP and the $\bPsi$-RIP are slightly different concepts: the former ensures norm preservation of all sparse coefficient vectors $\balpha$, while the latter ensures norm preservation of all signals having a sparse representation $\bx = \bPsi \balpha$.
In many problems (such as when $\bPsi$ is an overcomplete dictionary), the RIP is considered to be a stronger requirement.

There are a variety of approaches to constructing matrices that satisfy the RIP or $\bPsi$-RIP, some of which are better suited to practical architectures than others. From a theoretical standpoint, however, the most fruitful approaches involve the use of random matrices.  Specifically, we consider matrices constructed as follows: given $M$ and $N$, we generate a random $M \times N$ matrix $\bA$ by choosing the entries $\bA[m,n]$ as independent and identically distributed (i.i.d.)\ random variables. While it is not strictly necessary, for the sake of simplicity we will consider only real-valued random variables, so that $\bA \in \real^{M \times N}$.

We impose two conditions on the random distribution. First, we require that the distribution is centered and normalized such that $\mathbb{E}(\bA[m,n]) = 0$ and $\mathbb{E}(\bA[m,n]^2) = \frac{1}{M}$. Second, we require that the distribution is {\em subgaussian}~\cite{BuldyK_Metric,vershynin2010introduction}, meaning that there exists a constant $c_0 > 0$ such that
\begin{equation} \label{eq:sgdef}
\mathbb{E} \left( e^{\bA[m,n] t} \right) \le e^{ c_0^2t^2}
\end{equation}
for all $t \in \real$.
%
%
Examples of subgaussian distributions include the Gaussian distribution, the Rademacher distribution, and the uniform distribution.  In general, any distribution with bounded support is subgaussian.

The key property of subgaussian random variables that will be of use in this paper is that for any $\bx \in \complex^N$, the random variable $\norm[2]{\bA \bx}^2$ is highly concentrated about $\norm[2]{\bx}^2$. In particular, there exists a constant $c_1(\eta) > 0$ that depends only on $\eta$ and the constant $c_0$ in (\ref{eq:sgdef}) such that
\begin{equation} \label{star}
\prob{ \left| \norm[2]{\bA \bx}^2- \norm[2]{\bx}^2\right| \ge \eta \norm[2]{\bx}^2} \le  4 e^{-c_1(\eta) M},
\end{equation}
where the probability is taken over all draws of the matrix $\bA$ (see Lemma 6.1 of~\cite{DeVorPW_Instance} or~\cite{Daven_Concentration}).\footnote{The concentration result in~\eqref{star} is typically stated for $\bx \in \real^N$ instead of $\complex^N$.  The complex case follows from the real case by handling the real and imaginary parts separately and then applying the union bound, which results in a factor of $4$ instead of  $2$ in front of the exponent.}  We leave the constant $c_1(\eta)$ undefined since it will depend both on the particular subgaussian distribution under consideration and on the range of $\eta$ considered.  Importantly, however, for any subgaussian distribution and any $\eta_{\mathrm{max}}$, we can write $c_1(\eta) = \kappa \eta^2$ for $\eta \le \eta_{\mathrm{max}}$ with $\kappa$ being a constant that depends on certain properties of the distribution~\cite{Daven_Concentration}.  This concentration bound has a number of important consequences.  Perhaps most important for our purposes is the following lemma (an adaptation of Lemma 5.1 in~\cite{BaranDDW_Simple}).\footnote{The constants in~\cite{BaranDDW_Simple} differ from those in Lemma~\ref{lem:subspace}, but the proof is substantially the same (see~\cite{DavenBWB_Signal}). Note that in~\cite{DavenBWB_Signal} $\cX$ is a subspace of $\real^N$ rather than $\complex^N$.  In our case we incur an additional factor of 2 in the constant which arises as a consequence of the increase in the covering number for a sphere in $\complex^S$ (which can easily be derived from the fact that there is an isometry between $\complex^S$ and $\real^{2S}$).}
\begin{lemma} \label{lem:subspace}
Let $\cX$ denote any $S$-dimensional subspace of $\complex^N$.  Fix $\delta, \beta \in (0,1)$. Let $\bA$ be an $M \times N$ random matrix with i.i.d.\ entries chosen from a distribution satisfying (\ref{star}). If
\begin{equation} \label{eq:subspace-M}
M \ge \frac{ 2S \log (42/\delta) + \log(4/\beta) }{c_1(\delta/\sqrt{2})}
\end{equation}
then with probability exceeding $1 - \beta$,
\begin{equation} \label{eq:subspace-embed}
\sqrt{1-\delta} \norm{\bx} \le \norm{\bA \bx} \le \sqrt{1+\delta} \norm{\bx}
\end{equation}
for all $\bx \in \cX$.
\end{lemma}
When $\bPsi$ is an orthonormal basis, one can use this lemma to go beyond a single $S$-dimensional subspace to instead consider all possible subspaces spanned by $S$ columns of $\bPsi$, thereby establishing the RIP for $\bA \bPsi$.
The proof follows that of Theorem 5.2 of~\cite{BaranDDW_Simple}.
\begin{lemma} \label{lem:RIP}
Let $\bPsi$ be an orthonormal basis for $\complex^N$ and fix $\delta, \beta \in (0,1)$. Let $\bA$ be an $M \times N$ random matrix with i.i.d.\ entries chosen from a distribution satisfying (\ref{star}). If
\begin{equation} \label{eq:RIP-M}
M \ge \frac{ 2 S \log (42 e N/ \delta S) + \log(4/\beta) }{c_1 (\delta/\sqrt{2})}
\end{equation}
with $e$ denoting the base of the natural logarithm, then with probability exceeding $1 - \beta$, $\bA \bPsi$ will satisfy the RIP of order $S$ with constant $\delta$.
\end{lemma}
\begin{proof}
This is a simple generalization of Lemma~\ref{lem:subspace}, which follows from the observation that~\eqref{RIP} is equivalent to~\eqref{eq:subspace-embed} holding for all $S$-dimensional subspaces.  There are $\binom{N}{S} \le (eN/S)^{S}$ subspaces of dimension $S$ aligned with the coordinate axes of $\bPsi$, and so applying a union bound to Lemma~\ref{lem:subspace} we obtain the desired result.
\end{proof}
From essentially the same argument, we can also prove for more general dictionaries $\bPsi$ that $\bA$ will satisfy the $\bPsi$-RIP.
\begin{cor} \label{cor:PsiRIP}
Let $\bPsi$ be an arbitrary $N \times D$ matrix and fix $\delta, \beta \in (0,1)$. Let $\bA$ be an $M \times N$ random matrix with i.i.d.\ entries chosen from a distribution satisfying (\ref{star}). If
\begin{equation} \label{eq:Psi-RIP-M}
M \ge \frac{ 2S \log (42 e D/ \delta S) + \log(4/\beta) }{c_1( \delta/\sqrt{2})}
\end{equation}
with $e$ denoting the base of the natural logarithm, then with probability exceeding $1 - \beta$, $\bA$ will satisfy the $\bPsi$-RIP of order $S$ with constant $\delta$.
\end{cor}

As noted above, the random matrix approach is somewhat impractical to build in hardware.  However, several hardware architectures have been implemented and/or proposed that enable compressive samples to be acquired in practical settings. Examples include the random demodulator~\cite{TroppLDRB_Beyond}, random filtering~\cite{TroppWDBB_Random}, the modulated wideband converter~\cite{MishaE_From}, random convolution~\cite{BajwaHRWN_Toeplitz,Rombe_Compressive}, and the compressive multiplexer~\cite{SlaviLDB_Compressive}.  In this paper we will rely on random matrices in the development of our theory, but we will see via simulations that the techniques we propose are also applicable to systems that use some of these more practical architectures.

\subsection{CS recovery algorithms}
\label{subsec:rec}

\subsubsection{Greedy and iterative algorithms}
\label{sec:greedyiterative}

Before we return to the problem of designing $\bPsi$, we first discuss the question of how to recover the vector $\bx$ from measurements of the form $\by= \bA \bx + \be = \bA \bPsi \balpha + \be$. The original CS theory proposed $\ell_1$-minimization as a recovery technique~\cite{Cande_Compressive,Donoh_Compressed}.  Convex optimization techniques are powerful methods for CS signal recovery, but there also exist a variety of alternative greedy or iterative algorithms that are commonly used in practice and that satisfy similar performance guarantees, including iterative hard thresholding (IHT)~\cite{BlumeD_Iterative}, orthogonal matching pursuit (OMP)~\cite{PatiRK_Orthogonal,DavisMZ_Adaptive,Tropp_Greed,DavenW_Analysis}, and several more recent variations on OMP~\cite{DonohDTS_Sparse,NeedeV_Uniform,NeedeV_Signal,NeedeT_CoSaMP,DaiM_Subspace,CohenDD_Instance}.

In this paper we will restrict our attention to two of the most commonly used algorithms in practice---IHT and CoSaMP~\cite{BlumeD_Iterative,NeedeT_CoSaMP}. We begin with IHT, which is probably the simplest of all CS recovery algorithms.  As is the case for most iterative recovery algorithms, a core component of IHT is {\em hard thresholding}.  Specifically, we define the operator
\begin{equation}
\label{eq:hard}
\hard(\balpha,S)[n] = \begin{cases} \balpha[n], & \textrm{$|\balpha[n]|$ is among the $S$ largest elements of $|\balpha|$;} \\
                                        0, & \textrm{otherwise.}
                          \end{cases}
\end{equation}
In words, the hard thresholding operator sets all but the $S$ largest elements of a vector to zero (with ties broken according to any arbitrary rule).

To the best of our knowledge, there are no existing papers that specifically discuss how to implement IHT when $\bPsi$ is not an orthonormal basis or a tight frame (see~\cite{cevher2011alps} for a discussion of the latter case).
Nonetheless, we can envision two natural (and reasonable) ways that the canonical IHT algorithm~\cite{BlumeD_Iterative} can be extended to handle a general dictionary.
In the first of these variations, the algorithm would consist of iteratively applying the update rule
\begin{equation}
\label{eq:IHTalpha}
\balpha^{\ell+1} = \hard(\balpha^{\ell} + \mu (\bA \bPsi)^H(\by-\bA \bPsi \balpha^\ell),S)
\end{equation}
where $\mu$ is a parameter set by the user.
In the second of these variations, the algorithm would consist of iteratively applying the update rule
\begin{equation} \label{eq:IHT}
\bx^{\ell+1} = \bPsi \cdot \hard \left( \bPsi^H \left( \bx^{\ell} + \mu \bA^H \left(\by - \bA \bx^{\ell} \right) \right), S \right).
\end{equation}
When $\bPsi$ is an orthonormal basis these algorithms are equivalent, but in general they are not.
On the whole, IHT is a remarkably simple algorithm, but in practice its performance is greatly dependent on careful selection and adaptation of the parameter $\mu$.
We refer the reader to~\cite{BlumeD_Iterative} for further details.

CoSaMP is a somewhat more complicated algorithm, but can be easily understood as breaking the recovery problem into two separate sub-problems: identifying the $S$ columns of $\bPsi$ that best represent $\bx$ and then projecting onto that subspace.
The former problem is clearly somewhat challenging, but once solved, the latter is relatively straightforward.
In particular, if we have identified the optimal columns of $\bPsi$, indexed by the set $\Lambda$, then we can recover $\bx$ via least-squares.  In this case, an optimal recovery strategy is to solve the problem:
\begin{equation}
\label{eq:oracle}
\widehat{\bx} = \mathop{\argmin}_{\bz} \ \norm{ \by - \bA \bz} \quad \mathrm{s.t.} \quad \bz \in \cR(\bPsi_{\Lambda}),
\end{equation}
where $\bPsi_{\Lambda}$ denotes the submatrix of $\bPsi$ that contains only the columns of $\bPsi$ corresponding to the index set $\Lambda$ and $\cR(\bPsi_{\Lambda})$ denotes the range of $\bPsi_{\Lambda}$.  If we let $\widetilde{\bA} = \bA \bPsi$, then one way to obtain the solution to (\ref{eq:oracle}) is via the pseudoinverse of $\widetilde{\bA}_\Lambda$, denoted $\widetilde{\bA}_\Lambda^\dagger$.  Specifically, we can compute
\begin{equation}
\widehat{\balpha}|_{\Lambda} = \widetilde{\bA}_{\Lambda}^\dag \by = \left(\widetilde{\bA}_{\Lambda}^H \widetilde{\bA}_{\Lambda} \right)^{-1} \widetilde{\bA}_{\Lambda}^H  \by \quad \quad \mathrm{and} \quad \quad \widehat{\balpha}|_{\Lambda^c} = \bzero
\label{eq:ls2}
\end{equation}
and then set $\widehat{\bx} = \bPsi \widehat{\balpha}$.  While this is certainly not the only approach to solving~\eqref{eq:oracle} (as we will see in Section~\ref{ssec:regu}), it allows us to easily observe that in the noise-free setting, if the support estimate $\Lambda$ is correct, then $\by = \bA \bx = \widetilde{\bA}_\Lambda \balpha|_\Lambda$, and so plugging this into (\ref{eq:ls2}) yields $\widehat{\balpha} = \balpha$ (and hence $\widehat{\bx} = \bx$) provided that $\widetilde{\bA}_\Lambda$ has full column rank.  Thus, the central challenge in recovery is to correctly identify the set $\Lambda$.  CoSaMP and related algorithms solve this problem by iteratively identifying likely columns, performing a projection, and then improving the estimate of which columns to use.

Unfortunately, we are again not aware of any papers that specifically discuss how to implement CoSaMP when $\bPsi$ is not an orthonormal basis.
Nonetheless, we can envision two natural extensions of the canonical CoSaMP algorithm~\cite{NeedeT_CoSaMP}.
One of these is shown in Algorithm~\ref{alg:cosamp};\footnote{ We note that the choice of $2S$ in the ``identify'' step is primarily driven by the proof technique, and is not intended to be interpreted as an optimal or necessary choice.  For example, in~\cite{DaiM_Subspace} it is shown that the choice of $S$ is sufficient to establish performance guarantees similar to those for CoSaMP.  It is also important to note that when the number of measurements $M$ is very small (less than $3S$) it is necessary to make suitable modifications as the assumptions of the algorithm are clearly violated in this case.  Moreover, a simple extension of CoSaMP as presented here involves including an additional orthogonalization step after pruning $\widetilde{\bx}$ down to an $S$-dimensional estimate, as is also done in~\cite{DaiM_Subspace}.  This can often result in modest performance gains and is a technique that we exploit in our simulations.} in a sense, this formulation is more analogous to~\eqref{eq:IHT} than to~\eqref{eq:IHTalpha} because it is focused on recovery of $\bx$ rather than $\balpha$.
However, Algorithm~\ref{alg:cosamp} is actually quite flexible and can be invoked in multiple ways.
To help distinguish among the different possibilities, it will be helpful to introduce the notation
$$
\widehat{\bx} = \mathrm{CoSaMP}(\bA,\bPsi,\by,S)
$$
to denote the output produced by Algorithm~\ref{alg:cosamp} when the arguments $(\bA,\bPsi,\by,S)$ are provided as input.
Having set this notation, it is also reasonable to consider invoking Algorithm~\ref{alg:cosamp} with the input arguments $(\bA \bPsi, \bI, \by, S)$.
This formulation is more analogous to~\eqref{eq:IHTalpha}. In this case we will denote the output by $\widehat{\balpha} = \mathrm{CoSaMP}(\bA \bPsi, \bI, \by, S)$ since the algorithm will construct and output an estimate of $\balpha$ (rather than $\bx$).

\begin{algorithm}[t]
\caption{Compressive Sampling Matching Pursuit (CoSaMP)} \label{alg:cosamp}
\begin{algorithmic}
\STATE \textbf{input:} $\bA$, $\bPsi$, $\by$, $S$, stopping criterion
\STATE \textbf{initialize:} $\br^0 = \by$,  $\bx^0 = 0$, $\ell = 0$
\WHILE{not converged}
\STATE
\begin{tabular}{ll}
\textbf{proxy:} & $\bh^{\ell} = \bA^H \br^{\ell}$ \\
\textbf{identify:} & $\Omega^{\ell+1} = \supp(\hard( \bPsi^H \bh^{\ell},2S))$ \\
\textbf{merge:} & $\Lambda^{\ell+1} = \supp( \hard( \bPsi^H \bx^{\ell}, S ) ) \cup \Omega^{\ell+1}$ \\
\textbf{update:} & $\widetilde{\bx} = \argmin_{\bz}  \norm{\by - \bA \bz} \quad \mathrm{s.t.} \quad \bz \in \cR(\bPsi_{\Lambda^{\ell+1}})$ \\
 & $\bx^{\ell+1} = \bPsi \cdot \hard(\bPsi^H \widetilde{\bx},S)$  \\
 & $\br^{\ell+1} = \by - \bA \bx^{\ell+1}$   \\
 & $\ell = \ell+1$
\end{tabular}
\ENDWHILE
\STATE \textbf{output:} $\widehat{\bx} = \bx^{\ell}$
\end{algorithmic}
\end{algorithm}

\subsubsection{``Model-based'' recovery algorithms}

Traditional approaches to CS signal recovery, like those described above, place no prior assumptions on $\supp(\balpha)$.  Sparsity on its own implies nothing about the locations of the nonzeros, and hence most approaches to CS signal recovery treat every possible support as equally likely.  However, in many practical applications the nonzeros are not distributed completely at random, but rather exhibit a degree of structure. In the case of signals exhibiting such {\em structured sparsity}, it is possible to both reduce the required number of measurements and develop specialized ``model-based'' algorithms for recovery that exploit this structure~\cite{BaranCDH_Model,Blume_Sampling}.

In this paper, we are interested in the model of {\em block-sparsity}~\cite{BaranCDH_Model,EldarKB_Block-Sparse}.  In a block-sparse vector $\balpha$, the nonzero coefficients cluster in a small number of blocks.  Specifically, suppose that $\bx  = \bPsi \balpha$ with $\bPsi$ being an $N \times D$ matrix and that we decompose $\bPsi$ into $J$ submatrices of size $N \times \frac{D}{J}$, i.e.,
$$
\bPsi = \left[ \bPsi_0 ~ \bPsi_1 ~ \cdots ~ \bPsi_{J-1} \right].
$$
Then we can write $\bx = \sum_{i=0}^{J-1} \bPsi_i \balpha_i$, where each $\balpha_i \in \complex^{D/J}$.  We say that $\balpha$ is $K$-block-sparse if there exists a set $\cI \subseteq \{0, 1, \ldots, J-1 \}$ such that $|\cI| = K$ and $\balpha_i = 0$ for all $i \notin \cI$.
With some abuse of notation, we now let let $\bPsi_{\cI}$ denote the submatrix of $\bPsi$ that contains only the columns of $\bPsi$ corresponding to the {\em blocks} indexed by $\cI$.

We first illustrate how we can exploit block-sparsity algorithmically. Our goal is to generalize IHT and CoSaMP to the block-sparse setting.  To do this, we observe that the hard thresholding function plays a key role in both algorithms.  One way to interpret this role is that $\hard(\bPsi^H \bx,S)$ is actually computing a projection of $\bPsi^H \bx$ onto the set of $S$-sparse vectors.  In the case where $\bPsi$ is an orthonormal basis we can also interpret $\bPsi \cdot \hard(\bPsi^H \bx,S)$ as projecting $\bx$ onto the set of signals that are $S$-sparse with respect to the basis $\bPsi$.

In the block-sparse case we must replace hard thresholding with an appropriate operator that takes a candidate signal and finds the closest $K$-block-sparse approximation.  Towards this end, we define
\begin{equation} \label{eq:blocksupp}
\cS_{\bPsi}(\bx,K) := \argmin_{|\cI| \le K} \min_{\bz \in \cR(\bPsi_{\cI})} \norm{\bx - \bz}.
\end{equation}
$\cS_{\bPsi}(\bx,K)$ is analogous to the support of $\balpha$ in the traditional sparse setting: it tells us which set of $K$ blocks of $\bPsi$ can best approximate $\bx$.  Along with $\cS_{\bPsi}(\bx,K)$, we also define
\begin{equation} \label{eq:blockproj}
\cP_{\bPsi}(\bx,K) := \argmin_{\bz} \norm{\bx - \bz} \quad \mathrm{s.t.} \quad \bz \in \cR \left(\bPsi_{\cS_{\bPsi}(\bx,K)}\right).
\end{equation}
$\cP_{\bPsi}$ is simply the projection of the vector $\bx \in \complex^N$ onto the set of $K$-block-sparse signals. To simplify our notation, we will often write $\cS(\balpha,K) = \cS_{\bPsi}(\balpha,K)$ and $\cP(\bx,K) = \cP_{ \bPsi}(\bx,K)$ when $\bPsi$ is clear from the context.

For each of the IHT and CoSaMP algorithms proposed in Section~\ref{sec:greedyiterative}, it is possible to propose a variation of the algorithm designed to exploit block-sparsity simply by replacing the hard thresholding operator with an appropriate block-sparse projection.
For example, one block-based version of IHT (which is also a special case of the iterative projection algorithm in~\cite{Blume_Sampling}), would consist of replacing the core iteration of IHT in~\eqref{eq:IHT} with
\begin{equation} \label{eq:IHTblock}
\bx^{\ell+1} = \cP(\bx^{\ell} + \mu \bA^H(\by-\bA \bx^\ell),K).
\end{equation}
Note that the only difference from~\eqref{eq:IHT} is that we have replaced hard thresholding with the projection onto the set of block-sparse signals.
Similarly, a block-based version of CoSaMP is shown in Algorithm~\ref{alg:bbcosamp}.

\begin{algorithm}[t]
\caption{Block-Based CoSaMP (BBCoSaMP)} \label{alg:bbcosamp}
\begin{algorithmic}
\STATE \textbf{input:} $\bA$, $\bPsi$, $\by$, $K$, stopping criterion
\STATE \textbf{initialize:} $\br^0 = \by$, $\bx^0 = 0$, $\ell = 0$
\WHILE{not converged}
\STATE
\begin{tabular}{ll}
\textbf{proxy:} & $\bh^{\ell} = \bA^H \br^{\ell}$ \\
\textbf{identify:} & $\widetilde{\cI}^{\ell+1} = \cS( \cP( \bh^{\ell},2K),2K)$ \\
\textbf{merge:} & $\cI^{\ell+1} = \cS(\bx^{\ell},K) \cup \widetilde{\cI}^{\ell+1}$ \\
\textbf{update:} & $\widetilde{\bx} = \argmin_{\bz}  \norm{\by - \bA \bz} \quad \mathrm{s.t.} \quad \bz \in \cR(\bPsi_{\cI^{\ell+1}})$ \\
 & $\bx^{\ell+1} = \cP(\widetilde{\bx},K)$ \\
 & $\br^{\ell+1} = \by - \bA \bx^{\ell+1}$   \\
 & $\ell = \ell+1$
\end{tabular}
\ENDWHILE
\STATE \textbf{output:} $\widehat{\bx} = \bx^{\ell}$
\end{algorithmic}
\end{algorithm}

Algorithm~\ref{alg:bbcosamp} is in differing senses both a special case and a generalization of the model-based CoSaMP algorithm proposed in~\cite{BaranCDH_Model}.
Specifically, in~\cite{BaranCDH_Model} an algorithm for block-sparse signal recovery is proposed that is equivalent to Algorithm~\ref{alg:bbcosamp} when $\bPsi = \bI$.
The more general case of arbitrary $\bPsi$ is not discussed.
However, there are alternative options for handling $\bPsi \neq \bI$ besides the one specified in Algorithm~\ref{alg:bbcosamp}.
Like Algorithm~\ref{alg:cosamp}, Algorithm~\ref{alg:bbcosamp} is quite flexible and can be invoked in multiple ways.
Following our convention in Section~\ref{sec:greedyiterative}, we use the notation $\widehat{\bx} = \mathrm{BBCoSaMP}(\bA,\bPsi,\by,K)$ to denote the output produced by Algorithm~\ref{alg:bbcosamp} when the arguments $(\bA,\bPsi,\by,K)$ are provided as input, and we use the notation $\widehat{\balpha} = \mathrm{BBCoSaMP}(\bA \bPsi, \bI, \by, K)$ to denote the output produced by Algorithm~\ref{alg:bbcosamp} when the arguments $(\bA \bPsi, \bI, \by, K)$ are provided as input.

\subsubsection{``Model-based'' recovery guarantees}
\label{sec:mbguarantees}

Theoretical guarantees for standard CS recovery algorithms typically rely on the RIP, and since in the standard case any $S$-sparse signal is possible, there is little room for improvement.  However, the block-sparse model actually rules out a large number of possible signal supports, and so we no longer require the full RIP or $\bPsi$-RIP, i.e., we no longer need~\eqref{RIP} or~\eqref{Psi-RIP} to hold for all possible $S$-sparse signals.  Instead we only require that~\eqref{RIP} or~\eqref{Psi-RIP} hold for all $\balpha$ which are $K$-block-sparse.  We will refer to these relaxed properties as the block-RIP and $\bPsi$-block-RIP respectively.

The relaxation to block-sparse signals allows us to potentially dramatically reduce the required number of measurements. Specifically, note that a $K$-block-sparse vector $\balpha$ satisfies $\norm[0]{\balpha} \le K \frac{D}{J}$. In the standard sparsity model we would have that the number of possible subspaces is $\binom{D}{S}$ with $S = K \frac{D}{J}$, whereas now the number of possible subspaces is given by $\binom{J}{K}$, which can be potentially much smaller.  Establishing~\eqref{RIP} or~\eqref{Psi-RIP} for a more general union of subspaces is a problem that has received some attention in the CS literature~\cite{BlumeD_Sampling,LuD_Theory,EldarM_Robust}.  In our context it should be clear that we can simply apply Lemma~\ref{lem:subspace} just as in the proofs of Lemma~\ref{lem:RIP} and Corollary~\ref{cor:PsiRIP} to obtain the following improved bounds.

\begin{cor} \label{cor:blockRIP}
Let $\bPsi$ be an orthonormal basis for $\complex^N$ and fix $\delta, \beta \in (0,1)$. Let $\bA$ be an $M \times N$ random matrix with i.i.d.\ entries chosen from a distribution satisfying (\ref{star}). If
\begin{equation} \label{eq:blockRIP-M}
M \ge \frac{ 2K \left(\frac{N}{J} \log (42/ \delta) + \log (e J/K)\right) + \log(4/\beta) }{c_1(\delta/\sqrt{2})}
\end{equation}
with $e$ denoting the base of the natural logarithm, then with probability exceeding $1 - \beta$, $\bA \bPsi$ will satisfy the block-RIP of order $K$ with constant $\delta$.
\end{cor}

\begin{cor} \label{cor:PsiblockRIP}
Let $\bPsi$ be an arbitrary $N \times D$ matrix and fix $\delta, \beta \in (0,1)$. Let $\bA$ be an $M \times N$ random matrix with i.i.d.\ entries chosen from a distribution satisfying (\ref{star}). If
\begin{equation} \label{eq:Psi-blockRIP-M}
M \ge \frac{ 2 K \left(\frac{D}{J} \log (42/ \delta) + \log (e J/K)\right) + \log(4/\beta) }{c_1(\delta/\sqrt{2})}
\end{equation}
with $e$ denoting the base of the natural logarithm, then with probability exceeding $1 - \beta$, $\bA$ will satisfy the $\bPsi$-block-RIP of order $K$ with constant $\delta$.
\end{cor}

The measurement requirements in~\eqref{eq:blockRIP-M} and~\eqref{eq:Psi-blockRIP-M} represent improvements over~\eqref{eq:RIP-M} or~\eqref{eq:Psi-RIP-M} in that a straightforward application of~\eqref{eq:RIP-M} or~\eqref{eq:Psi-RIP-M} would lead to replacing the $\log(eJ/K)$ term above with $\frac{N}{J} \log(eJ/K)$ or $\frac{D}{J} \log(eJ/K)$ respectively.

We can combine these corollaries with the following theorems to show that we can stably recover block-sparse signals using fewer measurements. Note that the following theorems are simplified guarantees for the case of only approximately block-sparse signals. Both theorems can be generalized to block-compressible signals,\footnote{By block-compressible, we mean signals that are well-approximated by a block-sparse signal.  The guarantee on the recovery error for block-compressible signals is similar to those in Theorems~\ref{thm:blockiht} and~\ref{thm:blockcosamp} but includes an additional additive component that quantifies the error incurred by approximating $\balpha$ with a block-sparse signal.  If a signal is close to being block-sparse, then this error is negligible, but if a signal is not block-sparse at all, then this error can potentially be large.} but we restrict our attention to the guarantees for sparse signals to simplify discussion.

\begin{thm}[Theorem 2 from~\cite{Blume_Sampling}] \label{thm:blockiht}
Suppose that $\bx$ can be written as $\bx = \bPsi \balpha$ where $\balpha$ is $K$-block-sparse and that we observe $\by = \bA \bPsi \balpha + \be = \bA \bx + \be$.  If $\bA$ satisfies the $\bPsi$-block-RIP of order $2K$ with constant $\delta_{2K} \le 0.2$ and $\frac{1}{\mu} \in [1+\delta_{2K},1.5(1-\delta_{2K}))$, then the estimate obtained from block-based IHT~\eqref{eq:IHTblock} satisfies
\begin{equation} \label{eq:bbihtthm}
\norm{\bx - \widehat{\bx}} \le \kappa_1 \norm{\be}
\end{equation}
where $\kappa_1>1$ is a constant determined by $\delta_{2K}$  and the stopping criterion.
\end{thm}

\begin{thm}[Theorem 6 from~\cite{BaranCDH_Model}] \label{thm:blockcosamp}
Suppose that $\balpha$ is $K$-block-sparse and that we observe $\by = \bA \bPsi \balpha + \be$.  If $\bA \bPsi$ satisfies the block-RIP of order $4K$ with constant $\delta_{4K} \le 0.1$, then the output of block-based CoSaMP (Algorithm~\ref{alg:bbcosamp}) with $\widehat{\balpha} = \mathrm{BBCoSaMP}(\bA \bPsi, \bI, \by, K)$ satisfies
\begin{equation} \label{eq:bbcosampthm}
\norm{\balpha - \widehat{\balpha}} \le \kappa_2 \norm{\be}
\end{equation}
where $\kappa_2>1$ is a constant determined by $\delta_{4K}$ and the stopping criterion.
\end{thm}

Theorems~\ref{thm:blockiht} and~\ref{thm:blockcosamp} show that measurement noise has a controlled impact on the amount of noise in the reconstruction.
However, note that Theorems~\ref{thm:blockiht} and~\ref{thm:blockcosamp} are fundamentally different from one another when the matrix $\bPsi$ is not an orthonormal basis.
Theorem~\ref{thm:blockcosamp} requires the assumption that $\bA \bPsi$ satisfies the block-RIP while Theorem~\ref{thm:blockiht} requires that $\bA$ satisfies the $\bPsi$-block-RIP.
Theorem~\ref{thm:blockcosamp} also provides a different guarantee (recovery of $\balpha$) compared to Theorem~\ref{thm:blockiht} (which guarantees recovery of $\bx$).
In the case where $\bPsi$ is an orthonormal basis, we could immediately set $\widehat{\bx} = \bPsi \widehat{\balpha}$ so that the recovery guarantee in~\eqref{eq:bbcosampthm} applies to the error in $\bx$ as well.
However, for arbitrary dictionaries $\bPsi$ this equivalence no longer holds.\footnote{It is also worth noting that in the context of traditional (as opposed to model-based) CS, there do exist guarantees on $\norm{\widehat{\bx}-\bx}$ for general $\bPsi$ when using an alternative optimization-based approach combined with the $\bPsi$-RIP~\cite{CandeENR_Compressed}.
}
We conjecture that were we to instead consider $\widehat{\bx} = \mathrm{BBCoSaMP}(\bA, \bPsi, \by, K)$ we should be able to establish a theorem for block-based CoSaMP analogous to Theorem~\ref{thm:blockiht}, but we leave such an analysis for future work.

In Section~\ref{sec:MBDPSS} we develop a multiband modulated DPSS dictionary $\bPsi$ designed to offer high-quality block-sparse representations for sampled multiband signals.
Using this dictionary we establish block-RIP and $\bPsi$-block-RIP guarantees in Section~\ref{sec:recovery}, which allows us to translate Theorems~\ref{thm:blockiht} and~\ref{thm:blockcosamp} into guarantees for recovery of sampled multiband signals from compressive measurements.
When we then turn to implement these algorithms, however, it is important to note that, although we can implement the $\mathrm{BBCoSaMP}(\bA \bPsi, \bI, \by, K)$ version of block-based CoSaMP with no trouble, there is an important caveat to the results for block-based IHT and the $\mathrm{BBCoSaMP}(\bA, \bPsi, \by, K)$ version of block-based CoSaMP.
Specifically, both of those latter algorithms require that we be able to compute $\cP(\bx,K)$ as defined in~\eqref{eq:blockproj}.
Unfortunately, because our dictionary $\bPsi$ is not orthonormal we are aware of no algorithms that are guaranteed to solve this problem.
However, we will see empirically in Section~\ref{sec:sims} that we can attempt to solve this problem by applying many of the same algorithms commonly used for CS recovery.
In other words, we can perfectly implement an algorithm ($\mathrm{BBCoSaMP}(\bA \bPsi, \bI, \by, K)$) that has a provable guarantee on the recovery of $\balpha$, and we can approximately implement algorithms (block-based IHT and $\mathrm{BBCoSaMP}(\bA, \bPsi, \by, K)$) that are intended to accurately recover $\bx$.
We have implemented both variations of block-based CoSaMP, and while both perform well in practice, the empirical performance of $\mathrm{BBCoSaMP}(\bA, \bPsi, \by, K)$ turns out to be superior.  In Section~\ref{sec:sims} we present a suite of simulations demonstrating the remarkable effectiveness of $\mathrm{BBCoSaMP}(\bA, \bPsi, \by, K)$ (when combined with the multiband modulated DPSS dictionary) in recovering sampled multiband signals from compressive measurements. 

\section{A Sparse Dictionary for Sampled Multiband Signals}
\label{sec:MBDPSS}

\subsection{Analog signal models}
\label{ssec:analogmodels}

We now confront the problem of designing a dictionary $\bPsi$ in which $\bx$, a length-$N$ vector of Nyquist-rate samples of $x(t)$, will be sparse or compressible.
This is typically a significant challenge to anyone applying CS techniques to analog signals, since many natural analog signal models cannot be obviously represented via a simple basis $\bPsi$.
We now describe two of the most appealing analog signal models and discuss the degree to which these models can fit within the CS framework.

\subsubsection{Multitone signals}

There are a variety of possibilities for quantifying the notion of sparsity in a continuous-time signal $x(t)$.  Perhaps the simplest, dating back at least to the work of Prony~\cite{Prony_Essai}, is the {\em multitone} model, which assumes that $x(t)$ can be expressed as
\begin{equation} \label{eq:mtone1}
x(t) = \sum_{k=0}^{S-1} \alpha_k e^{j 2 \pi F_k t},
\end{equation}
i.e., $x(t)$ is simply a weighted combination of $S$ complex exponentials of arbitrary frequencies.  A related model is given by
\begin{equation} \label{eq:mtone2}
x(t) = \sum_{n=0}^{N-1} \balpha[n] e^{j 2 \pi (n-N/2+1) t/N \tsamp},
\end{equation}
where $\norm[0]{\balpha} = S$.  This model is considered in~\cite{TroppLDRB_Beyond}, which provided one of the first extensions of the CS framework to the case of continuous-time signals.  The advantage of this model is that (\ref{eq:mtone2}) implies that $\bx = \bPsi \balpha$, where $\bx$ is the vector of discrete-time samples of $x(t)$ on $[0,\twind)$ and $\bPsi$ is the non-normalized $N \times N$ DFT matrix with suitably ordered columns. Thus, the model (\ref{eq:mtone2}) immediately fits into the standard CS framework when the vector of coefficients $\balpha$ is sparse.

In practice, however, this approach is inadequate for two main reasons.  First, the model in (\ref{eq:mtone2}) assumes that any tones in the signal have a frequency that lies exactly on the so-called {\em Nyquist-grid}, i.e., the tones are {\em bounded harmonics}.  When dealing with tones of arbitrary frequencies, the corresponding ``DFT leakage'' results in $\balpha$ that are not sparse and are not even well-approximated as sparse. In this case it can be useful to either incorporate a smooth windowing function into the measurement system, as in~\cite{TroppLDRB_Beyond}, or to consider the less restrictive model in (\ref{eq:mtone1}), as in~\cite{DuartB_Spectral}.  However, these approaches do not address the second main objection, which is that many (if not most) real-world signals are not mere combinations of a few pure tones.  For a variety of reasons, it is typically more realistic to assume that each of the $S$ signal components has some non-negligible bandwidth, which leads us to instead consider the following extension of the multitone model.

\subsubsection{Multiband signals}
\label{ssec:multiband}

For the remainder of this paper, we will focus on {\em multiband} signals, or signals whose Fourier transform is concentrated on a small number of continuous intervals or bands.  Towards this end, for a general continuous-time signal $x(t)$, we define $\cF(x)$ as the support of $X(F)$, i.e., $\cF(x) = \{F \in \real: X(F) \neq 0 \}$.  We will be interested in signals for which we can decompose $\cF = \cF(x)$ into a union of $K$ continuous intervals, so that we can write
$$
\cF \subseteq \bigcup_{i=0}^{K-1} [a_i,b_i].
$$
In the most general setting, we would allow the endpoints of each interval to be arbitrary but subject to a restriction on the total Lebesgue measure of their union.  See, for example,~\cite{FengB_Spectrum,BreslF_Spectrum,Feng_Universal,VenkaB_Further}.  In this paper we restrict ourselves to a simpler model.  Specifically, we divide the range of possible frequencies from $-\frac{\bnyq}{2}$ to $\frac{\bnyq}{2}$ into $J = \frac{\bnyq}{\bband}$ equal bands of width $\bband$ and require $X(F)$ to be supported on at most $K$ of these bands. An example is illustrated in Figure~\ref{fig:mbspectrum}.  More formally, we define the $i^{\mathrm{th}}$ band as
$$
\Delta_i = \left[-\frac{\bnyq}{2} + i \bband, -\frac{\bnyq}{2}+ (i+1)\bband \right].
$$
We then define
$$
\Gamma(K,\bband) = \left\{ \cF: \cF \subseteq \bigcup_{i \in \cI} \Delta_i \ \mathrm{for~some} \ \cI \subseteq \{0,1,\ldots,J-1\} \ \mathrm{with} \ |\cI| \le K \right\}
$$
as the set of all possible supports. Using this, we define
\begin{equation} \label{eq:mband}
\cM(K,\bband) = \left\{ x(t) : \cF(x) \subseteq \cF' \ \mathrm{for~some} \ \cF' \in \Gamma(K,\bband) \right\}
\end{equation}
as the set of multiband signals.  Note that the total occupied bandwidth is at most $K\bband$.
Our interest is in the setting where $K\bband \ll \bnyq$, so that if we knew {\em a priori} where the $K$ bands were located, we could acquire $x(t)$ in a streaming setting with only $K\bband$ samples per second.
(This is the so-called {\em Landau rate}~\cite{Landa_Necessary}.)
Our goal is to acquire finite windows of multiband signals without such {\em a priori} knowledge while keeping the measurement rate as close as possible to the Landau rate.

\begin{figure}
   \centering
   \hspace*{-.15in} \includegraphics[width=\imgwidth]{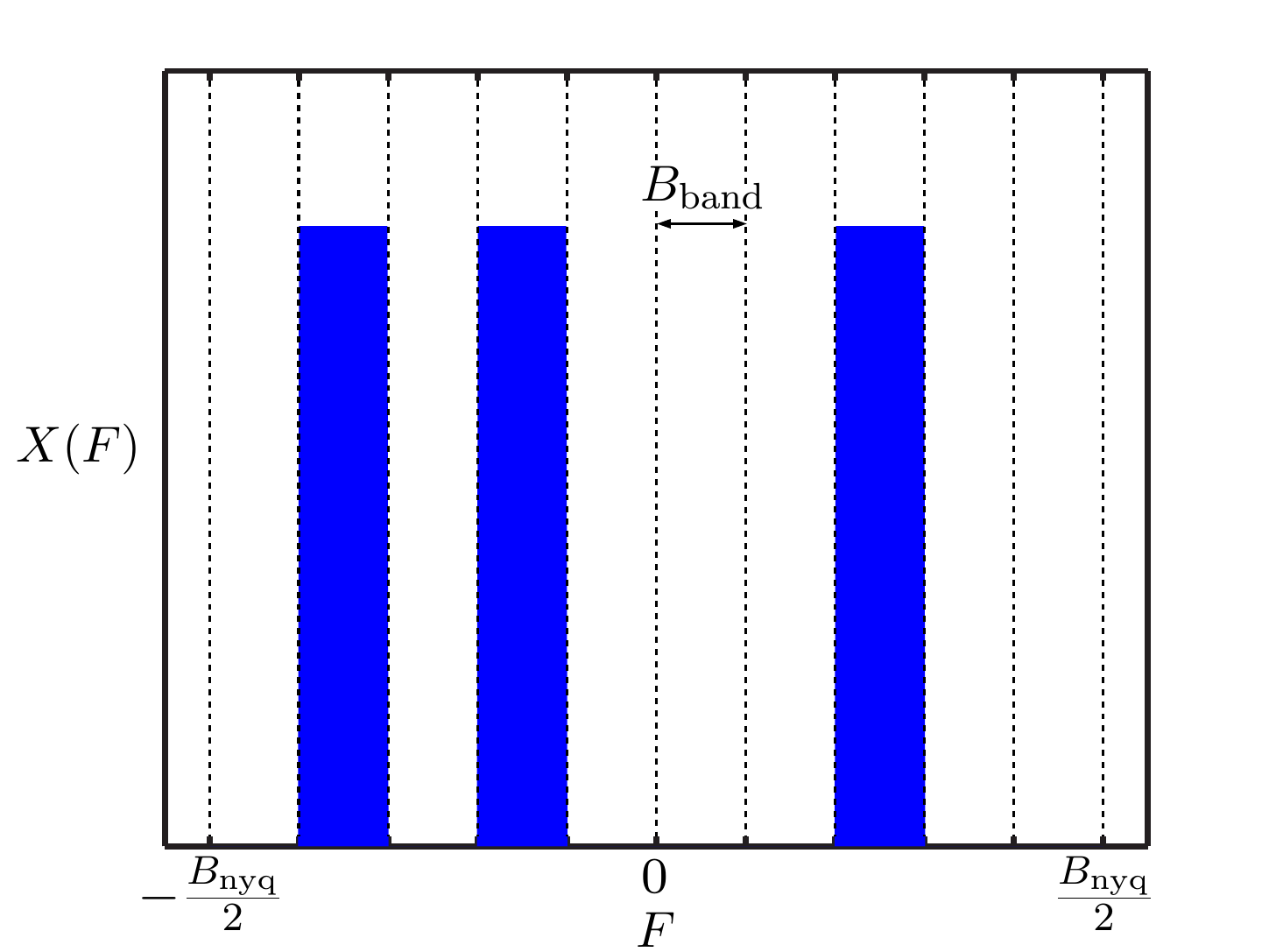}
   \caption{ \small \sl Illustration showing the CTFT of a multiband signal where $J = \frac{\bnyq}{\bband} = 10$ and $K=3$.
   \label{fig:mbspectrum}}
\end{figure}

Note, however, that the set $\cM(K,\bband)$ is defined for infinite-length signals $x(t)$.  Indeed, any signal with a Fourier transform supported on a finite range of frequencies cannot also be supported on a finite range of time.  This would seem to be somewhat at odds with the finite-dimensional CS framework described above.  As a result, previous efforts aimed at sampling multiband signals have developed largely outside the framework of CS~\cite{FengB_Spectrum,BreslF_Spectrum,Feng_Universal,VenkaB_Further,MishaE_From,MishaE_Blind}. It is our goal in this paper to show that it is possible to recover a finite-length window of samples of a multiband signal using many of the standard tools from CS.  To do this, we will need to construct an appropriate dictionary $\bPsi$ for capturing the structure of this set, which we do in the following section.

Finally, we also note that while our signal model breaks up the spectrum into $J = \frac{\bnyq}{\bband}$ bands with fixed boundaries and bandwidth, it actually encompasses the broader class of signals where the bandwidth and center frequency of each band are arbitrary.  For example, a signal with $K$ bands of width $\bband$ but with arbitrary center frequencies will automatically lie within $\cM(2K,\bband)$.  Since we are primarily interested in the case where $K\bband \ll \bnyq$, this factor of 2 will not be significant in the development of our theoretical results. However, we do note that in practice it may be possible to achieve a significant gain by exploiting a more flexible signal model.

In the remainder of this section we demonstrate that it is possible to construct discrete-time bases using {\em Discrete Prolate Spheroidal Sequences} (DPSS's) that efficiently capture the structure of sampled multiband signals. We first review DPSS's and their key properties as first developed in~\cite{Slepi_ProlateV}, and we then discuss some of the consequences of these properties in terms of their utility in representing sampled continuous-time signals.  Ultimately, we demonstrate how to use DPSS's to construct a dictionary $\bPsi$ which sparsely represents windows of sampled multiband signals.

\subsection{Discrete Prolate Spheroidal Sequences (DPSS's)}
\label{sec:dpssprops}

Our goal in this subsection is to provide a self-contained review of the concepts from Slepian's work on DPSS's~\cite{Slepi_ProlateV} that are most relevant to the problem of modeling sampled multiband signals. We refer the reader to~\cite{SlepiP_ProlateI,LandaP_ProlateII,LandaP_ProlateIII,Slepi_ProlateIV,Slepi_ProlateV,Slepi_On,Slepi_Some} for more complete overviews of DPSS's, PSWF's, and their implications in time-frequency localization.

\subsubsection{DPSS's}

Let $N$ be an integer, and let $0 < \slepw < \frac12$.
Given $N$ and $\slepw$, the DPSS's are a collection of $N$ discrete-time sequences that are strictly bandlimited to the digital frequency range $|f| \le \slepw$ yet highly concentrated in time to the index range $n = 0,1,\dots,N-1$. The DPSS's are defined to be the eigenvectors of a two-step procedure in which one first time-limits the sequence and then bandlimits the sequence. Before we can state a more formal definition, let us note that for a given discrete-time signal $x[n]$, we let
$$
X(f) = \sum_{n=-\infty}^\infty x[n] e^{-j 2 \pi f n}
$$
denote the discrete-time Fourier transform (DTFT) of $x[n]$.\footnote{Note that we use lower-case $f$ to indicate the digital frequency (so that $X(f)$ represents the DTFT of a discrete-time sequence $x[n]$, while $X(F)$ represents the CTFT of a continuous-time signal $x(t)$).} Next, we let $\cB_\slepw$ denote an operator that takes a discrete-time signal, bandlimits its DTFT to the frequency range $|f| \le \slepw$, and returns the corresponding signal in the time domain.
Additionally, we let $\cT_N$ denote an operator that takes an infinite-length discrete-time signal and zeros out all entries outside the index range $\{0,1,\dots,N-1\}$ (but still returns an infinite-length signal).
With these definitions, the DPSS's are defined as follows.
\begin{definition}~{\em\cite{Slepi_ProlateV}}
Given $N$ and $\slepw$, the {\em Discrete Prolate Spheroidal Sequences} (DPSS's) are a collection of $N$ real-valued discrete-time sequences
$s_{N,\slepw}^{(0)}, s_{N,\slepw}^{(1)}, \dots, s_{N,\slepw}^{(N-1)}$
that, along with the corresponding scalar eigenvalues
$$
1 > \lambda_{N,\slepw}^{(0)} > \lambda_{N,\slepw}^{(1)} > \cdots > \lambda_{N,\slepw}^{(N-1)} > 0,
$$
satisfy
\begin{equation}
\cB_\slepw(\cT_N (s_{N,\slepw}^{(\ell)})) = \lambda_{N,\slepw}^{(\ell)} s_{N,\slepw}^{(\ell)}
\label{eq:dpsseigmain}
\end{equation}
for all $\ell \in \{0,1,\dots,N-1\}$.
The DPSS's are normalized so that
\begin{equation}
\label{eq:dpssnorm}
\| \cT_N(s_{N,\slepw}^{(\ell)}) \|_2 = 1
\end{equation}
for all $\ell \in \{0,1,\dots,N-1\}$.
\end{definition}

As we discuss in more detail below in Section~\ref{sec:eigconc}, the eigenvalues $\lambda_{N,\slepw}^{(0)}, \lambda_{N,\slepw}^{(1)}, \dots, \lambda_{N,\slepw}^{(N-1)}$ have a very distinctive behavior: the first $2N\slepw$ eigenvalues tend to cluster extremely close to $1$, while the remaining eigenvalues tend to cluster similarly close to $0$.

Before proceeding, let us briefly mention several key properties of the DPSS's that will be useful in our subsequent analysis. First, it is clear from (\ref{eq:dpsseigmain}) that the DPSS's are, by definition, strictly bandlimited to the digital frequency range $|f| \le \slepw$. Second, the DPSS's are also approximately time-limited to the index range $n = 0,1,\dots,N-1$.  Specifically, it can be shown that~\cite{Slepi_ProlateV}
\begin{equation}
\label{eq:snorm}
\| s_{N,\slepw}^{(\ell)} \|_2 = \frac{1}{\sqrt{\lambda_{N,\slepw}^{(\ell)}}}.
\end{equation}
Comparing~\eqref{eq:dpssnorm} with~\eqref{eq:snorm}, we see that for values of $\ell$ where $\lambda_{N,\slepw}^{(\ell)} \approx 1$, nearly all of the energy in $s_{N,\slepw}^{(\ell)}$ is contained in the interval $\{0,1,\dots,N-1\}$. Third, the DPSS's are orthogonal over $\integers$~\cite{Slepi_ProlateV}, i.e., for any $\ell,\ell' \in \{0,1,\dots,N-1\}$ with $\ell \neq \ell'$, $\langle s_{N,\slepw}^{(\ell)}, s_{N,\slepw}^{(\ell')} \rangle = 0$.

\subsubsection{Time-limited DPSS's}
\label{sec:tldpss}

While each DPSS actually has infinite support in time, several very useful properties hold for the collection of signals one obtains by time-limiting the DPSS's to the index range $n = 0,1,\dots,N-1$. First, the time-limited DPSS's $\cT_N(s_{N,\slepw}^{(0)}), \cT_N(s_{N,\slepw}^{(1)}), \dots, \cT_N(s_{N,\slepw}^{(N-1)})$ are approximately bandlimited to the digital frequency range $|f| \le \slepw$.  Specifically, from (\ref{eq:dpsseigmain}) and (\ref{eq:snorm}), one can deduce that
\begin{equation}
\| \cB_\slepw( \cT_N (s_{N,\slepw}^{(\ell)})) \|_2 = \sqrt{\lambda_{N,\slepw}^{(\ell)}}.
\label{eq:apbl}
\end{equation}
Comparing~\eqref{eq:dpssnorm} with~\eqref{eq:apbl}, we see that for values of $\ell$ where $\lambda_{N,\slepw}^{(\ell)} \approx 1$, nearly all of the energy in $\cT_N (s_{N,\slepw}^{(\ell)})$ is contained in the frequencies $|f| \le \slepw$. An illustration of four representative time-limited DPSS's and their DTFT's is provided in Figure~\ref{fig:dpssplots}.
While by construction the DTFT of any DPSS is perfectly bandlimited, the DTFT of the corresponding time-limited DPSS will only be concentrated in the bandwidth of interest for the first $2N\slepw$ DPSS's.  As a result, we will frequently be primarily interested in roughly the first $2N\slepw$ DPSS's.

\begin{figure*}
   \centering
   \begin{tabular*}{\linewidth}{@{\extracolsep{\fill}} cccc}
   \hspace*{-.18in} \includegraphics[width=.28\linewidth]{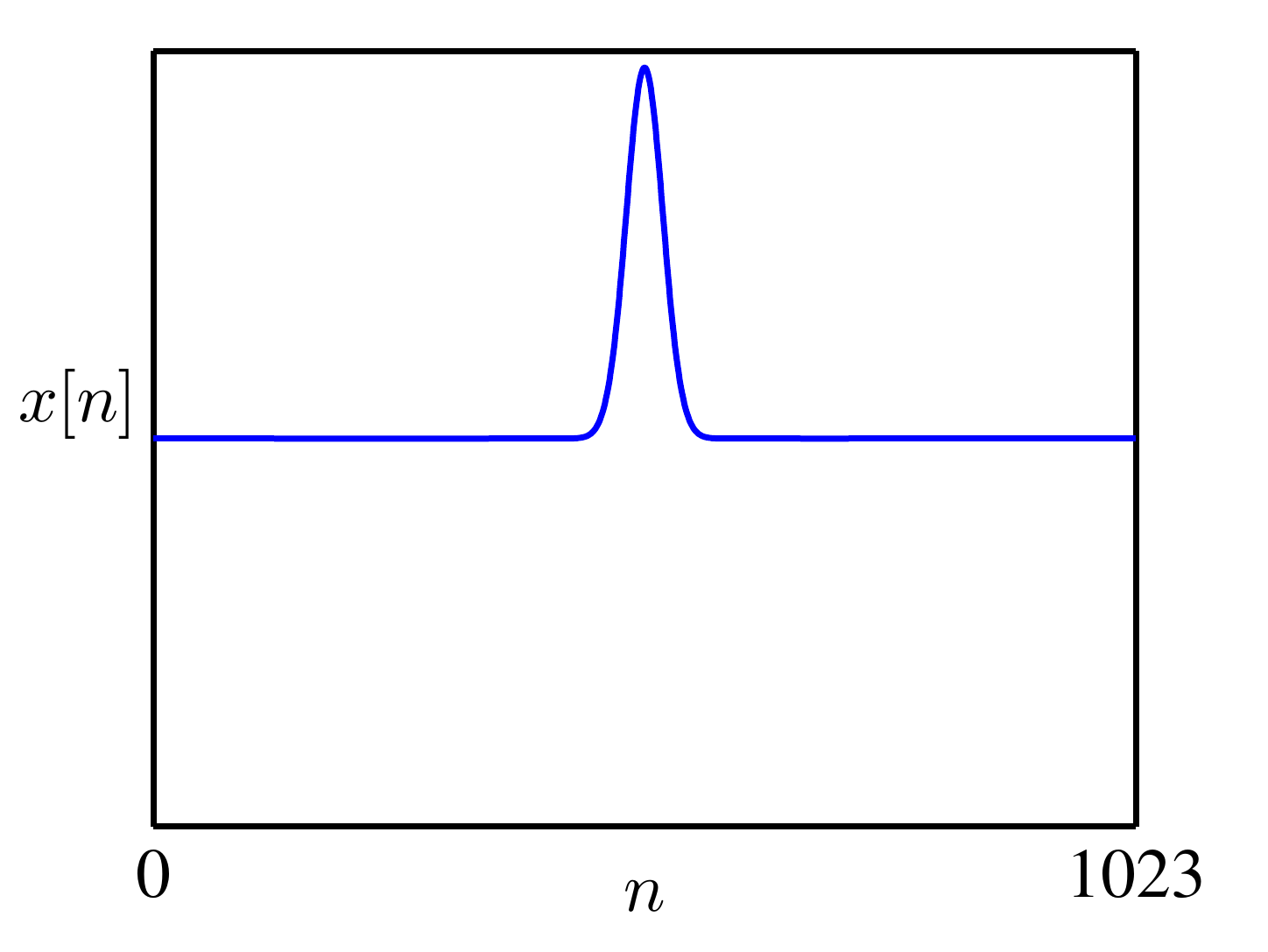} & \hspace*{-.37in}\includegraphics[width=.28\linewidth]{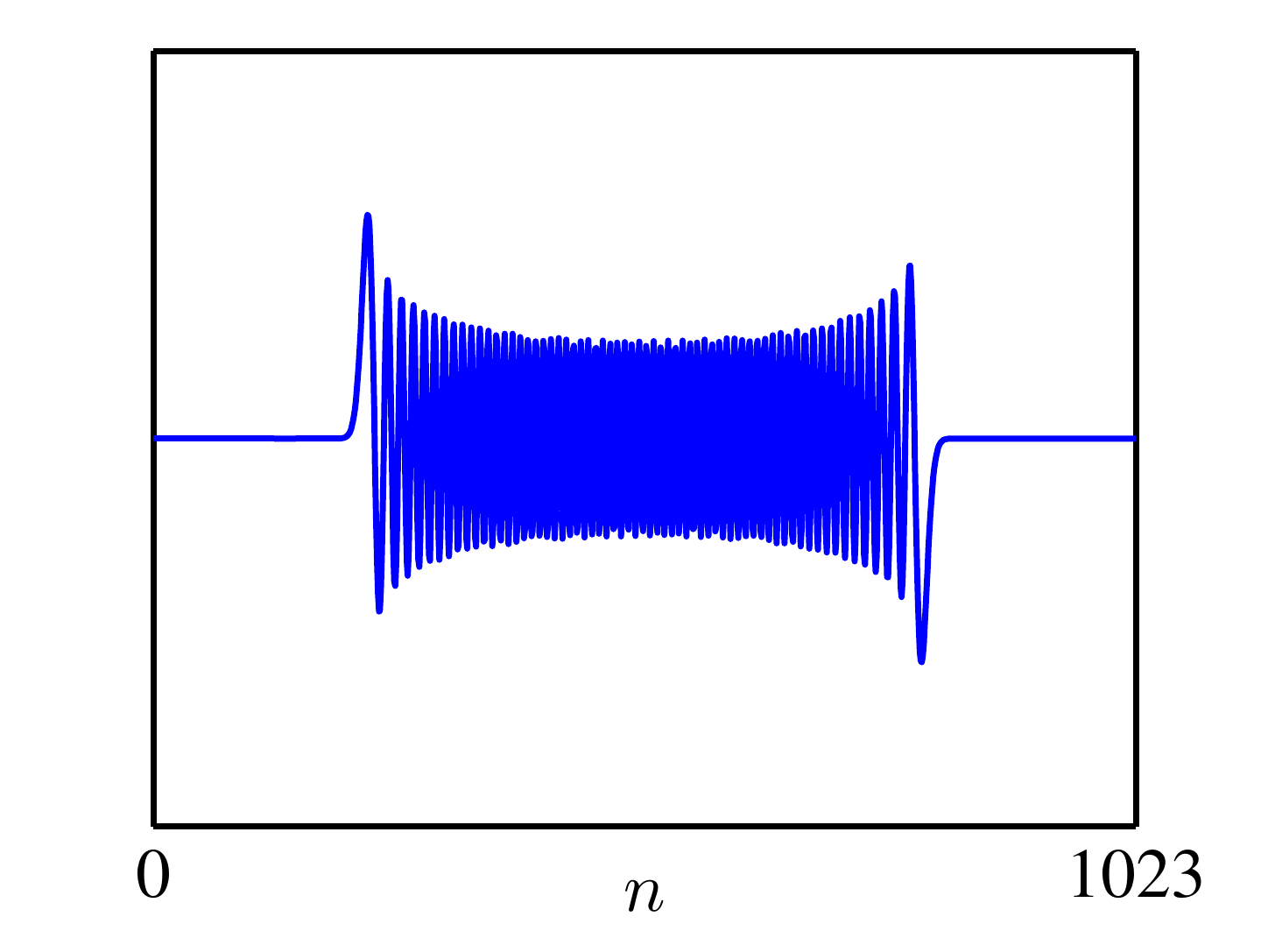} & \hspace*{-.37in} \includegraphics[width=.28\linewidth]{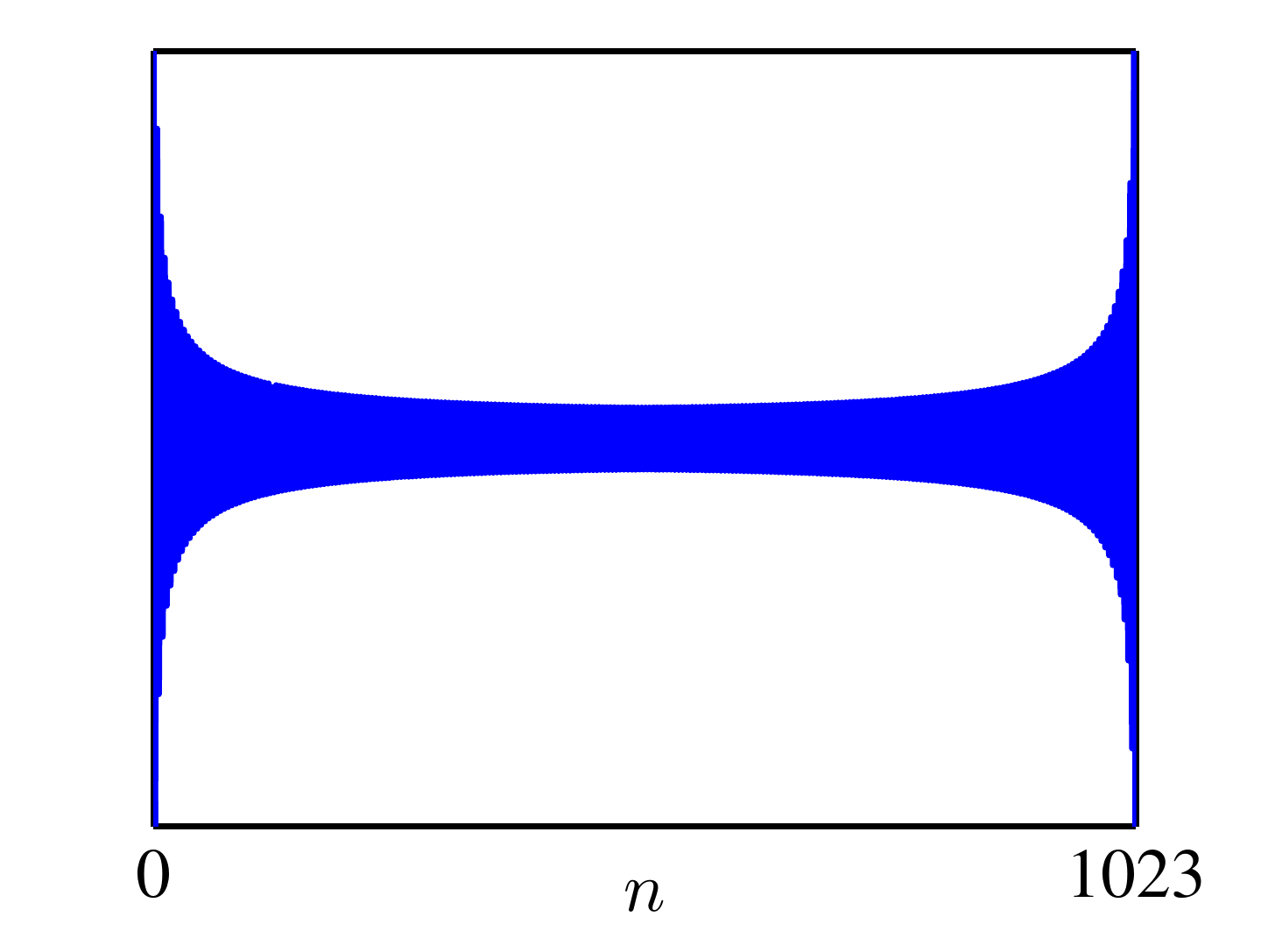} & \hspace*{-.37in}\includegraphics[width=.28\linewidth]{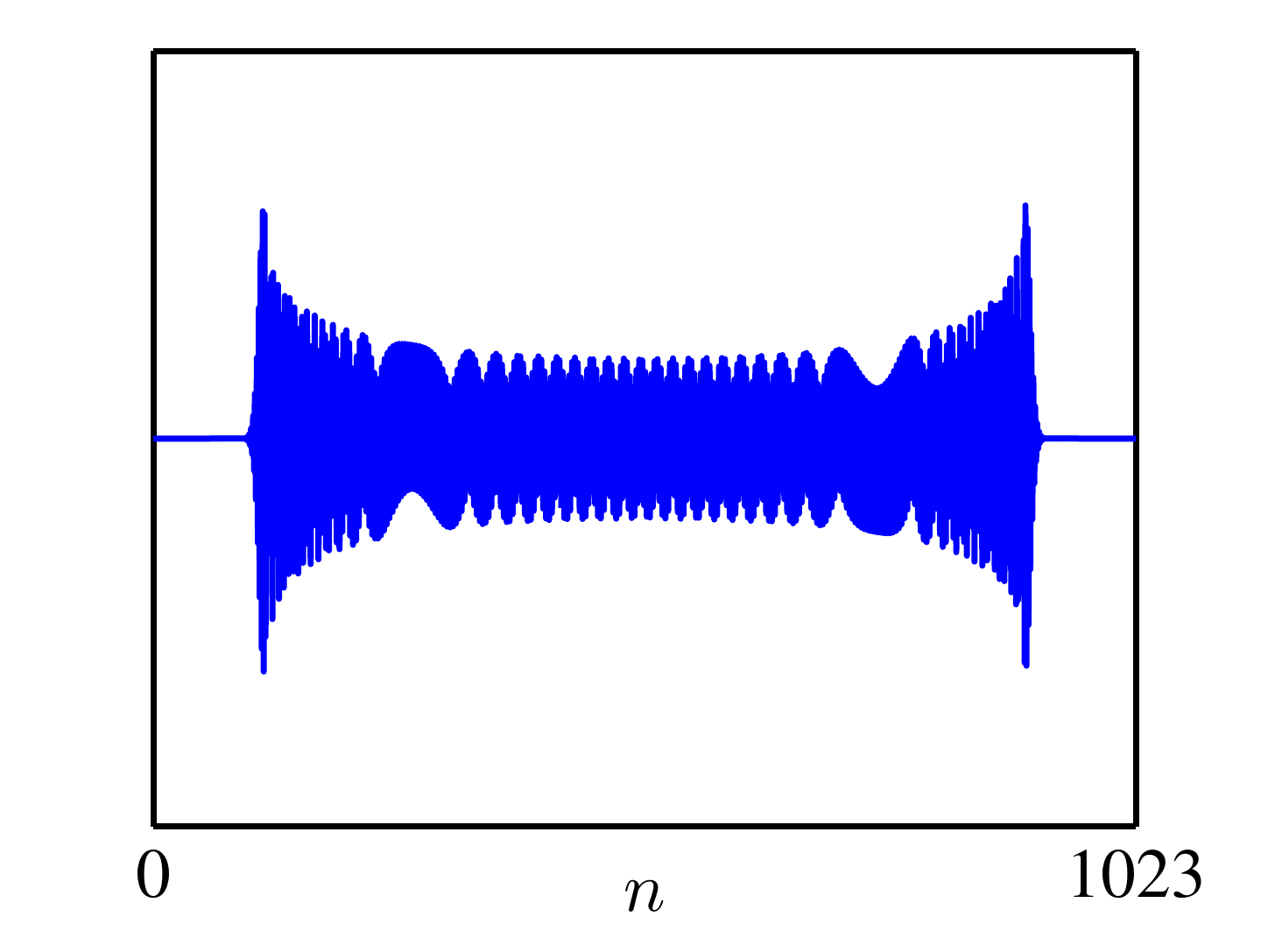} \\
   \hspace*{-.47in} \includegraphics[width=.28\linewidth]{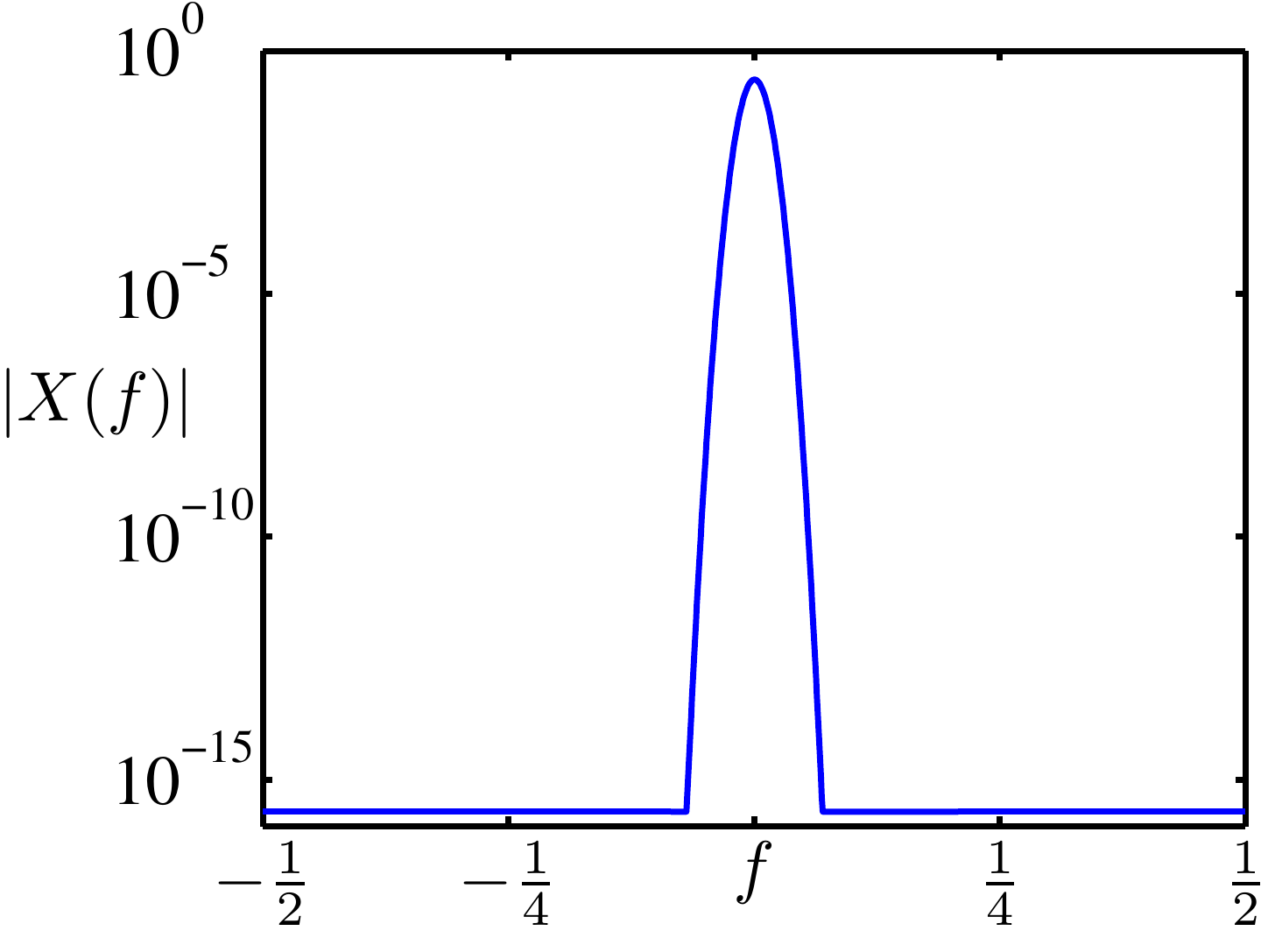} & \hspace*{-.38in} \includegraphics[width=.28\linewidth]{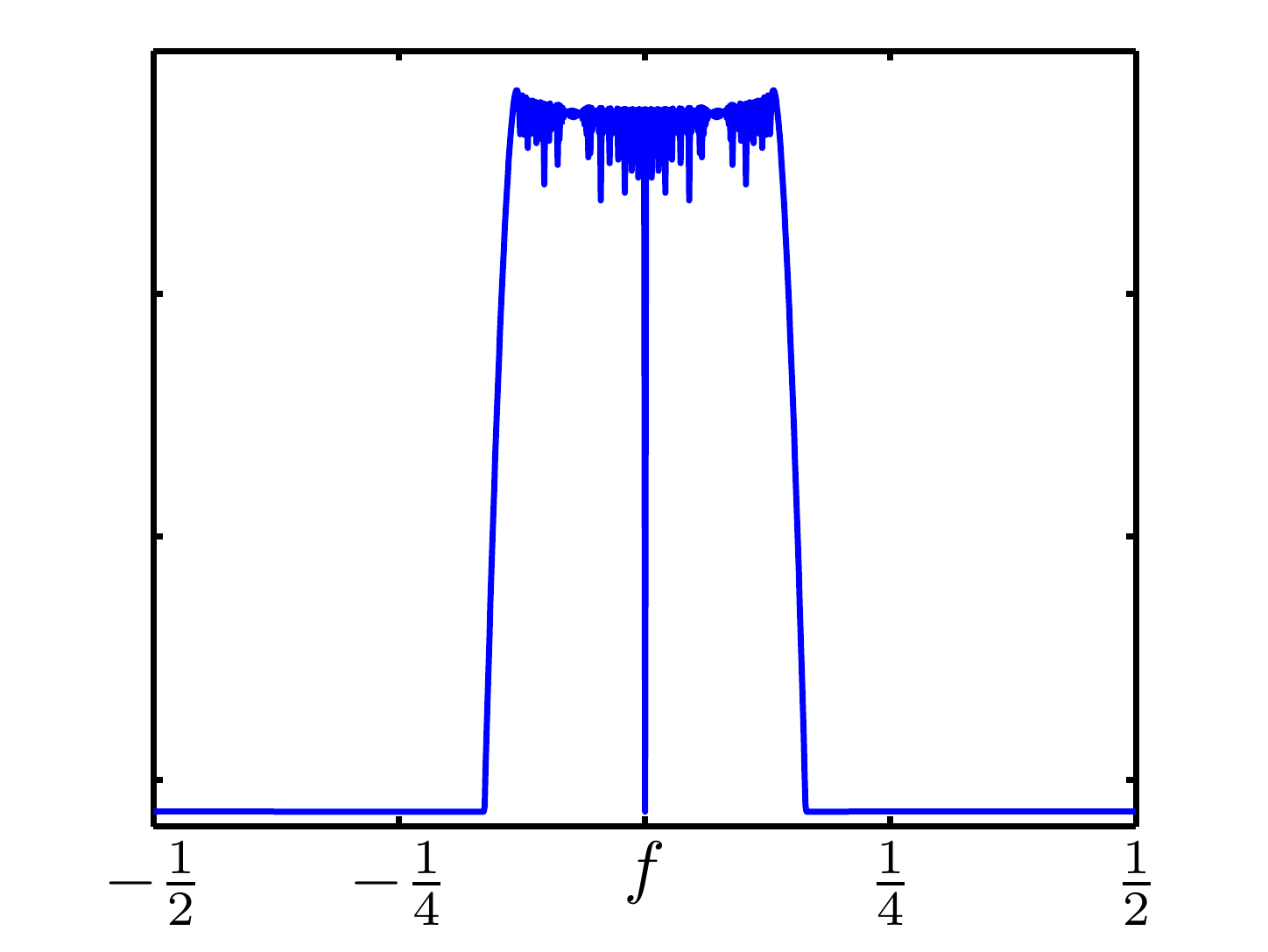} & \hspace*{-.37in} \includegraphics[width=.28\linewidth]{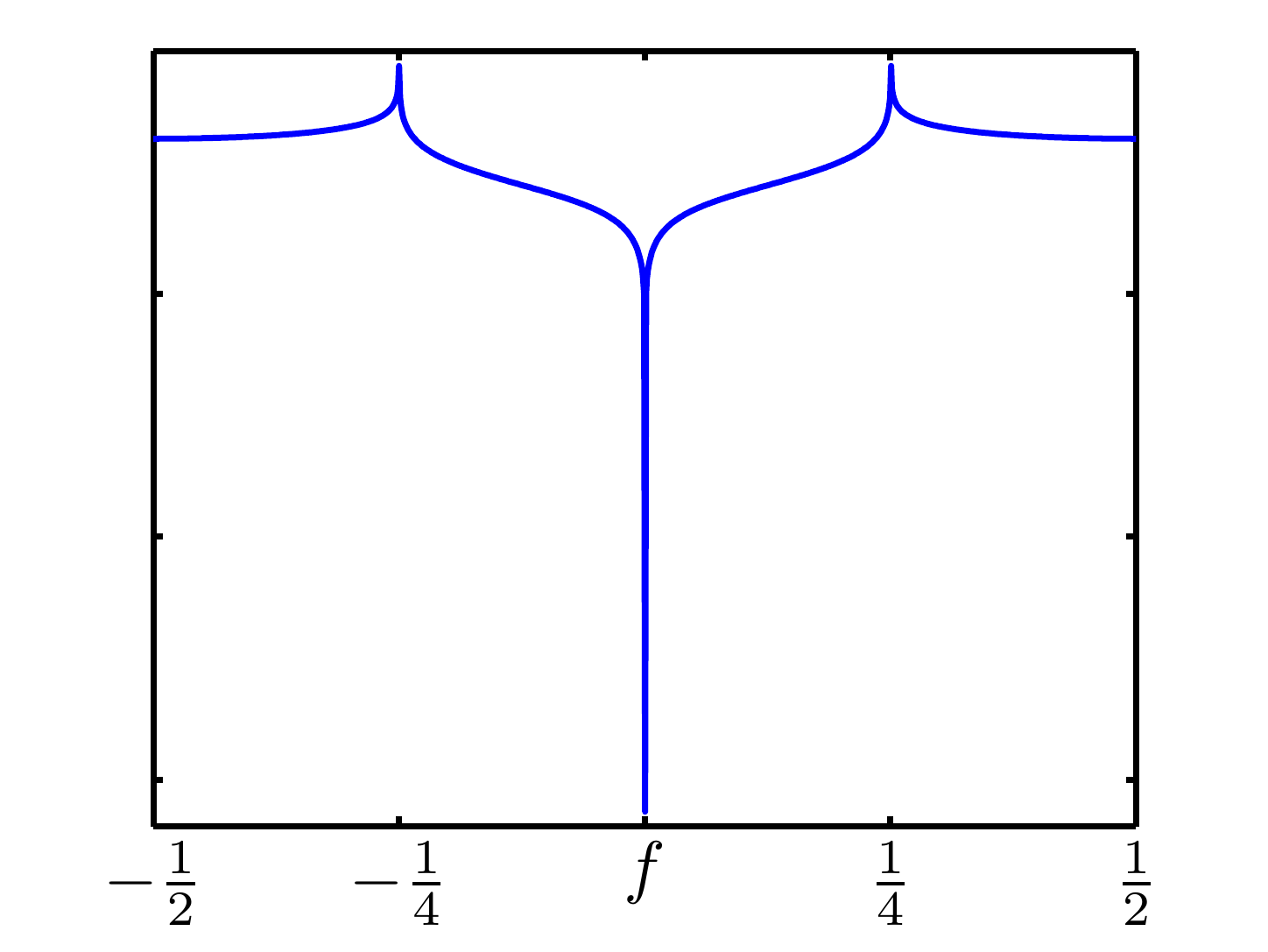} & \hspace*{-.37in}\includegraphics[width=.28\linewidth]{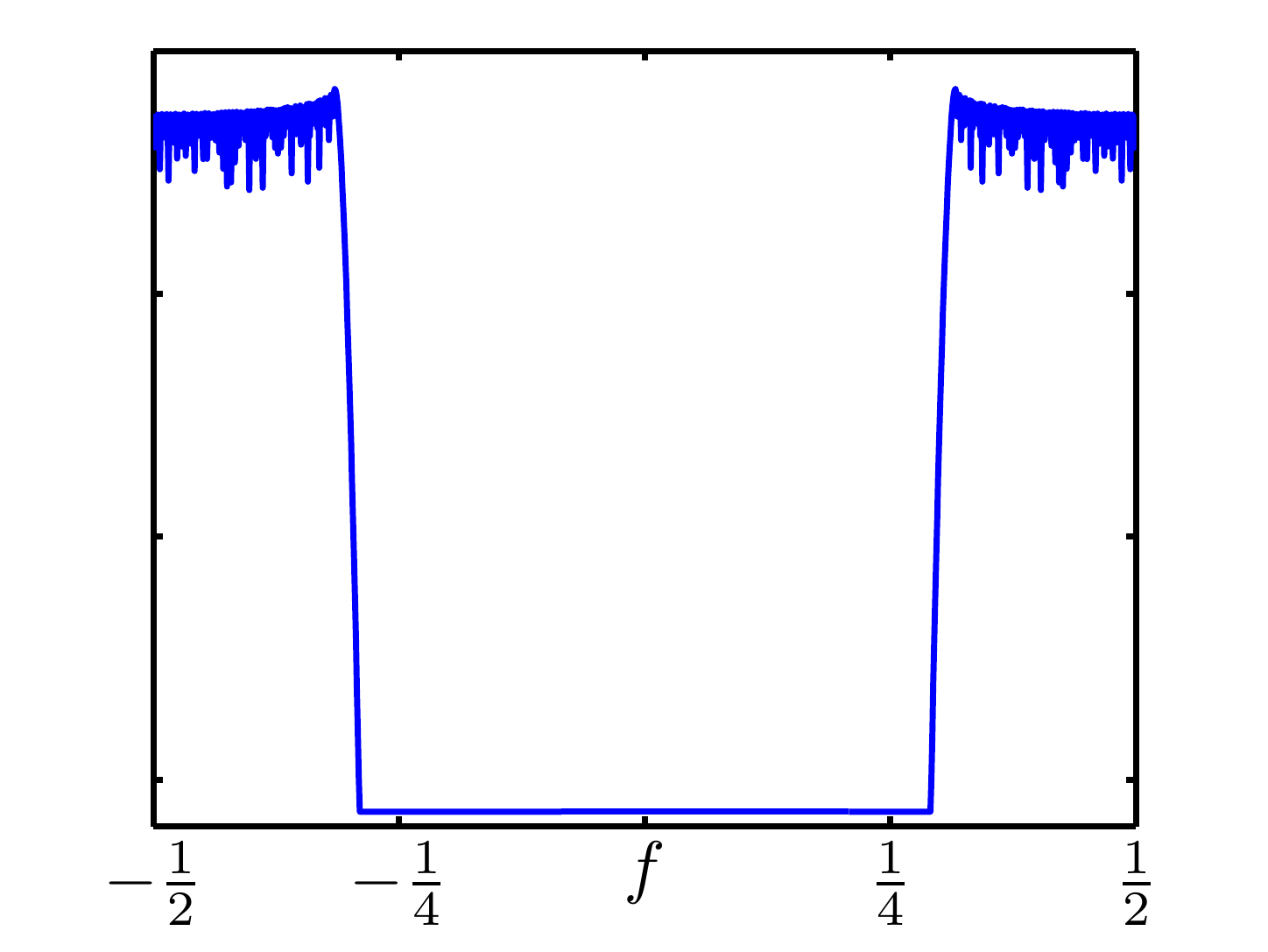} \\
   \hspace*{-.1in} {\small $\ell=0$} & \hspace*{-.15in}{\small $\ell=127$} & \hspace*{-.14in}{\small  $\ell=511$} & \hspace*{-.17in}{\small $\ell=767$}
   \end{tabular*}
   \caption{\small \sl Illustration of four example DPSS's time-limited to the interval $[0,1,\dots,N-1]$ and the magnitude of their DTFT's.  In this example, $N=1024$ and $\slepw = \frac{1}{4}$. Note that for $\ell$ up to approximately $2N\slepw-1$ the energy of the spectrum is highly concentrated on the interval $[-\slepw,\slepw]$, and when $\ell$ is sufficiently larger than $2N\slepw-1$ the energy of the spectrum is concentrated almost entirely outside the interval $[-\slepw,\slepw]$. \label{fig:dpssplots}}
\end{figure*}

Second, the time-limited DPSS's are also orthogonal~\cite{Slepi_ProlateV} so that for any $\ell,\ell' \in \{0,1,\dots,N-1\}$ with $\ell \neq \ell'$,
\begin{equation}
\label{eq:dpssorth}
 \langle \cT_N(s_{N,\slepw}^{(\ell)}), \cT_N(s_{N,\slepw}^{(\ell')}) \rangle  = 0.
\end{equation}
Finally, like the DPSS's, the time-limited DPSS's have a special eigenvalue relationship with the time-limiting and bandlimiting operators.
In particular, if we apply the operator $\cT_N$ to both sides of (\ref{eq:dpsseigmain}), we see that the sequences $\cT_N (s_{N,\slepw}^{(\ell)})$ are actually eigenfunctions of the two-step procedure in which one first bandlimits a sequence and then time-limits the sequence.

\subsubsection{DPSS vectors}

Because our focus in this paper is primarily on representing and reconstructing finite-length vectors, we will find the following restriction of the time-limited DPSS's to be useful, where we restrict the domain exclusively to the index range $n = 0,1,\dots,N-1$ (discarding the zeros outside this range).
\begin{definition}
Given $N$ and $\slepw$, the {\em DPSS vectors} $\bs_{N,\slepw}^{(0)}, \bs_{N,\slepw}^{(1)}, \dots, \bs_{N,\slepw}^{(N-1)} \in \real^N$ are defined by restricting the time-limited DPSS's to the index range $n = 0,1,\dots,N-1$:
$$
\bs_{N,\slepw}^{(\ell)}[n] := \cT_N(s_{N,\slepw}^{(\ell)})[n] = s_{N,\slepw}^{(\ell)}[n]
$$
for all $\ell,n \in \{0,1,\dots,N-1\}$.
\end{definition}

Following from our discussion in Section~\ref{sec:tldpss}, the DPSS vectors obey several favorable properties.
First, combining (\ref{eq:dpssnorm}) and (\ref{eq:dpssorth}), it follows that the DPSS vectors form an orthonormal basis for $\complex^N$ (or for $\real^N$).
However, as we discuss in subsequent sections, bases constructed using just the first $\approx 2N\slepw $ DPSS vectors can be remarkably effective for capturing the energy in our signals of interest.
Second, if we define $\bB_{N,\slepw}$ to be the $N \times N$ matrix with entries given by
\begin{equation}
\bB_{N,\slepw}[m,n] := 2\slepw \sinc{ 2\slepw(m-n)},
\label{eq:sdef}
\end{equation}
we see that the eigenvectors of $\bB_{N,\slepw}$ are given by the DPSS vectors $\bs_{N,\slepw}^{(0)}, \bs_{N,\slepw}^{(1)}, \dots, \bs_{N,\slepw}^{(N-1)}$, and the corresponding eigenvalues are $\lambda_{N,\slepw}^{(0)}, \lambda_{N,\slepw}^{(1)}, \dots, \lambda_{N,\slepw}^{(N-1)}$~\cite{Slepi_ProlateV}. Thus, if we concatenate the DPSS vectors into an $N \times N$ matrix
\begin{equation}
\bS_{N,\slepw} := \left[\bs_{N,\slepw}^{(0)} ~ \bs_{N,\slepw}^{(1)} ~  \cdots ~ \bs_{N,\slepw}^{(N-1)}\right] \in \real^{N \times N}
\label{eq:vdef}
\end{equation}
and let $\bLambda_{N,\slepw}$ denote an $N \times N$ diagonal matrix with the DPSS eigenvalues $\lambda_{N,\slepw}^{(0)}, \lambda_{N,\slepw}^{(1)}, \dots, \lambda_{N,\slepw}^{(N-1)}$ along the main diagonal, then we can write the eigendecomposition of $\bB_{N,\slepw}$ as
\begin{equation} \label{eq:Seigdecomp}
\bB_{N,\slepw} = \bS_{N,\slepw} \bLambda_{N,\slepw} \bS^H_{N,\slepw}.
\end{equation}
This decomposition will prove useful in our analysis below.

\subsubsection{Eigenvalue concentration}
\label{sec:eigconc}

As mentioned above, the eigenvalues $\lambda_{N,\slepw}^{(0)}, \lambda_{N,\slepw}^{(1)}, \dots, \lambda_{N,\slepw}^{(N-1)}$ have a very distinctive and important behavior: the first $2N\slepw$ eigenvalues tend to cluster extremely close to $1$, while the remaining eigenvalues tend to cluster similarly close to $0$. This behavior---which will allow us to construct efficient bases using small numbers of DPSS vectors---is illustrated in Figure~\ref{fig:dpsseigs} and captured more formally in the following results.

\begin{figure}
   \centering
   \hspace*{-.4in} \includegraphics[width=\imgwidth]{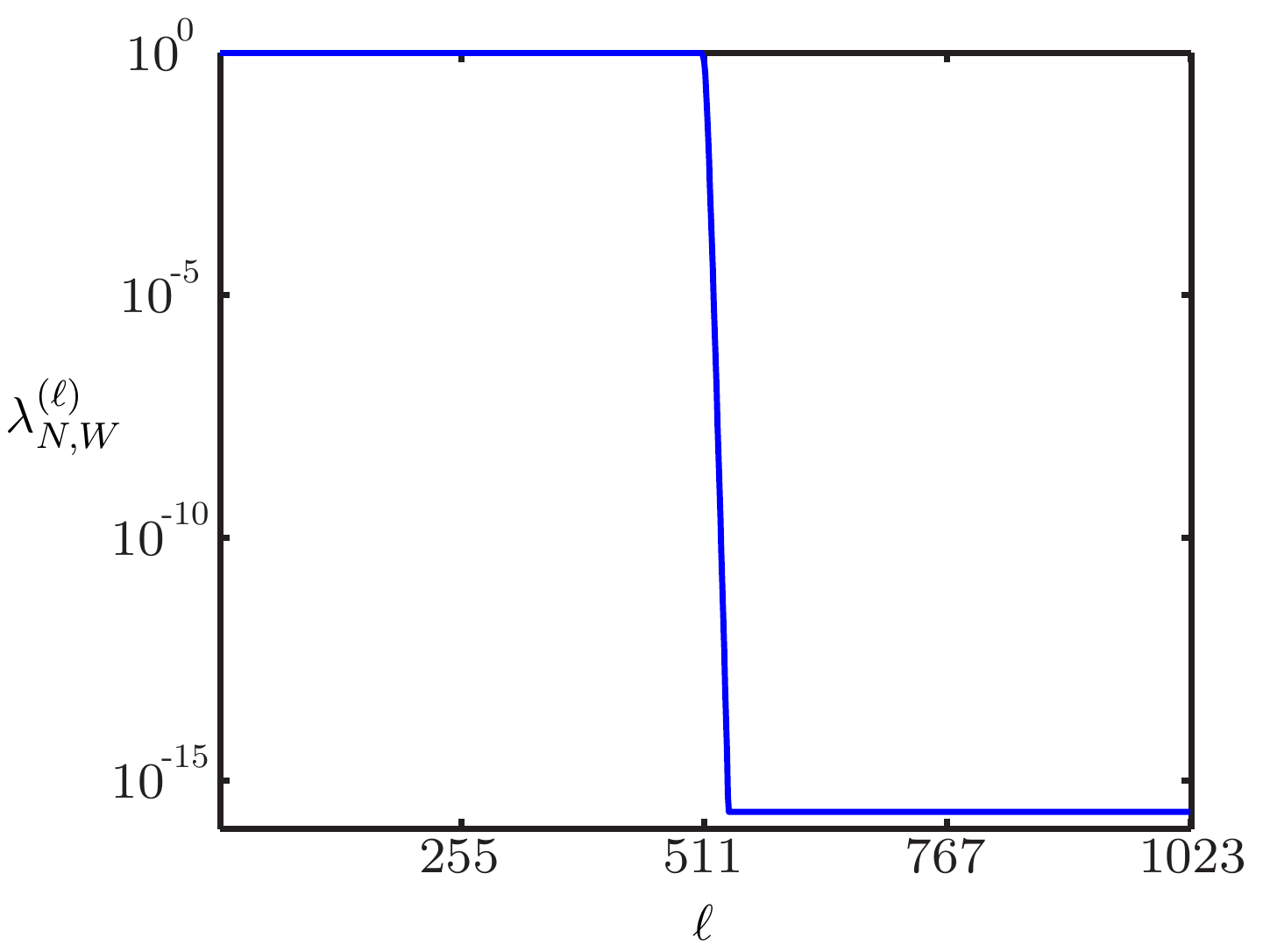}
   \caption{\small \sl Eigenvalue concentration for $N=1024$ and $\slepw = \frac{1}{4}$. Note that the first $2N\slepw = 512$ eigenvalues cluster near $1$ and the remaining eigenvalues rapidly approach the level of machine precision.
   \label{fig:dpsseigs}}
\end{figure}

\begin{lemma}[Eigenvalues that cluster near one \cite{Slepi_ProlateV}] Suppose that $\slepw$ is fixed, and let $\epsilon \in (0,1)$ be fixed. Then there exist constants $C_1, C_2$ (where $C_2$ may depend on $\slepw$ and $\epsilon$) and an integer $N_0$ (which may also depend on $\slepw$ and $\epsilon$) such that
\begin{equation}
\lambda_{N,\slepw}^{(\ell)} \ge 1 - C_1 e^{- C_2 N} ~~ \mathrm{for~all}~ \ell \le 2N\slepw(1-\epsilon) ~\mathrm{and~all}~N \ge N_0.
\label{eq:eig2nbminus}
\end{equation}
\label{lem:smallk}
\end{lemma}

\begin{lemma}[Eigenvalues that cluster near zero \cite{Slepi_ProlateV}]
Suppose that $\slepw$ is fixed, and let $\epsilon \in (0,\frac{1}{2\slepw}-1)$ be fixed. Then there exist constants $C_3, C_4$ (where $C_4$ may depend on $\slepw$ and $\epsilon$) and an integer $N_1$ (which may also depend on $\slepw$ and $\epsilon$) such that
\begin{equation}
\lambda_{N,\slepw}^{(\ell)} \le C_3 e^{- C_4 N} ~~ \mathrm{for~all}~ \ell \ge 2N\slepw(1+\epsilon) ~\mathrm{and~all}~N \ge N_1.
\label{eq:eig2nbplus}
\end{equation}
Alternatively, suppose that $\slepw$ is fixed, and let $\alpha > 0$ be fixed. Then there exist constants $C_5, C_6$ and an integer $N_2$ (where $N_2$ may depend on $\slepw$ and $\alpha$) such that
$$
\lambda_{N,\slepw}^{(\ell)} \le C_5 e^{- C_6 \alpha} ~~ \mathrm{for~all}~ \ell \ge 2N\slepw + \alpha \log(N) ~\mathrm{and~all}~N \ge N_2.
$$
\label{lem:bigk}
\end{lemma}

On occasion, we will have a need to compute bounds on sums of the eigenvalues. First, we note the following.
\begin{lemma}
\begin{equation}
\sum_{\ell = 0}^{N-1} \lambda_{N,\slepw}^{(\ell)} = \mathrm{trace}(\bB_{N,\slepw}) = 2N\slepw.
\label{eq:trace}
\end{equation}
\end{lemma}
The following results will also prove useful.

\begin{cor}
Suppose that $\slepw$ is fixed, let $\epsilon \in (0,1)$, and let $k = 2N\slepw(1-\epsilon)$. Then for $N \ge N_0$,
$$
\sum_{\ell = k}^{N-1} \lambda_{N,\slepw}^{(\ell)} \le 2 N \slepw (\epsilon + C_1 e^{-C_2 N}),
$$
where $C_1$, $C_2$, and $N_0$ are as specified in Lemma~\ref{lem:smallk}.
\label{cor:smallk}
\end{cor}
\begin{proof} It follows from~\eqref{eq:eig2nbminus} and~\eqref{eq:trace} that
\begin{align}
 \sum_{\ell = k}^{N-1} \lambda_{N,\slepw}^{(\ell)} &= 2N\slepw - \sum_{\ell = 0}^{k-1} \lambda_{N,\slepw}^{(\ell)} \notag \\
 &\le 2N\slepw - 2N\slepw(1-\epsilon)\left(1 - C_1 e^{-C_2 N}\right)  \notag \\
 &= 2 N \slepw \left(1 - (1-\epsilon)\left(1 - C_1 e^{-C_2 N}\right) \right) \notag \\
 &= 2 N\slepw \left(\epsilon + C_1 e^{-C_2 N} -  \epsilon C_1 e^{-C_2 N} \right) \notag \\
 &\le 2N\slepw \left(\epsilon + C_1 e^{-C_2 N} \right) \label{eq:DPSSbound1}
 \end{align}
for $N \ge N_0$.
\end{proof}

\begin{cor}
Suppose that $\slepw$ is fixed, let $\epsilon \in (0,\frac{1}{2\slepw}-1)$, and let $k = 2N\slepw(1+\epsilon)$. Then for for $N \ge \max(N_0,N_1)$,
\begin{equation}
\label{eq:DPSSbound2}
\sum_{\ell = k}^{N-1} \lambda_{N,\slepw}^{(\ell)} \le 2 N \slepw \min\left( \epsilon + C_1 e^{-C_2 N} , \;  \frac{C_3}{2\slepw} e^{-C_4 N} \right), \end{equation}
where $C_1$, $C_2$, $C_3$, $C_4$, $N_0$, and $N_1$ are as specified in Lemma~\ref{lem:smallk} and Lemma~\ref{lem:bigk}.
\label{cor:bigk}
\end{cor}
\begin{proof}
This result follows simply from Corollary~\ref{cor:smallk} and from~\eqref{eq:eig2nbplus}.
\end{proof}

\subsection{DPSS bases for sampled bandpass signals}

Using the DPSS vectors, it is possible to construct remarkably efficient bases for representing the discrete-time vectors that arise when collecting finite numbers of samples from most multiband signals (e.g., signals in $\cM(K,\bband)$).
Before presenting our full construction, however, we illustrate the basic concepts using sampled bandpass signals (e.g., signals in $\cM(1,\bband)$).

\subsubsection{A bandpass modulated DPSS basis}
\label{sec:mod}

For the moment, we restrict our attention to vectors of samples acquired from a continuous-time bandpass signal $x(t)$.
We assume that $\cF(x)$, the support of $X(F)$, is restricted to some interval $[F_c-\frac{\bband}{2}, F_c+\frac{\bband}{2}]$, where the center frequency $F_c$ and width $\bband$ are known.
We define $\bx = [x(0) ~ x(\tsamp) ~ \cdots ~ x((N-1)\tsamp)]^T$ as in Section~\ref{sec:background} to be a finite-length vector of $N$ samples of $x(t)$ collected uniformly with a sampling interval $\tsamp \le 1/(2 \max\{|F_c\pm\frac{\bband}{2}|\})$.

As a basis for efficiently representing many such vectors $\bx$, we propose the following.
First, let $\slepw= \frac{\bband\tsamp}{2}$, and as in~\eqref{eq:vdef}, let $\bS_{N,\slepw}$ denote the $N \times N$ matrix containing the $N$ DPSS vectors (constructed with parameters $N$ and $W$) as columns.
Next, define $f_c = F_c \tsamp$ and let $\bE_{f_c}$ denote an $N \times N$ diagonal matrix with entries
\begin{equation} \label{eq:modulatordef}
\bE_{f_c}[m,n] := \left\{ \begin{array}{ll} e^{j 2\pi f_c m}, & m = n, \\ 0, & m \neq n. \end{array} \right.
\end{equation}
Multiplying a vector by $\bE_{f_c}$ simply modulates that vector by a frequency $f_c$.
Finally, consider the $N \times N$ matrix $\bE_{f_c} \bS_{N,\slepw}$, whose columns are given by the DPSS vectors, each modulated by $f_c$.
One can easily check that $\bE_{f_c} \bS_{N,\slepw}$ forms an orthonormal basis for $\complex^N$, since $(\bE_{f_c} \bS_{N,\slepw})^H \bE_{f_c} \bS_{N,\slepw} = (\bS_{N,\slepw})^H (\bE_{f_c})^H \bE_{f_c} \bS_{N,\slepw} = (\bS_{N,\slepw})^H  \bS_{N,\slepw} = \bI$.
For a given integer $k \in \{1,2,\dots,N\}$, we let
$$
\left[ \bE_{f_c} \bS_{N,\slepw} \right]_k
$$
denote the $N \times k$ matrix formed by taking the first $k$ columns of $\bE_{f_c} \bS_{N,\slepw}$.
We will see that taking $\left[ \bE_{f_c} \bS_{N,\slepw} \right]_k$ with  $k \approx 2N\slepw$ forms an efficient basis that, to a high degree of accuracy, captures most sample vectors $\bx$ that can arise from sampling bandpass signals.

\subsubsection{Approximation quality for sampled complex exponentials}
\label{ssec:DPSSsinusoid}

To best illustrate one of our key points regarding the approximation of bandpass signals, let us first restrict our focus to the simplest possible bandpass signals: pure complex exponentials.
Specifically, consider a continuous-time signal of the form $x(t) = e^{j 2\pi F t}$ where the frequency $F$ belongs to the interval $[F_c-\frac{\bband}{2}, F_c+\frac{\bband}{2}]$.
We define $\bx = [x(0) ~ x(\tsamp) ~ \cdots ~ x((N-1)\tsamp)]^T$, $f_c = F_c \tsamp$, and $\slepw= \frac{\bband\tsamp}{2}$ as in Section~\ref{sec:mod}.
Also, defining
$$
\sampe{f} := \left[ \begin{array}{c} e^{j 2\pi f 0} \\ e^{j 2\pi f} \\ \vdots \\ e^{j 2\pi f (N-1)}  \end{array} \right]
$$
for all $f \in [f_c-\slepw,f_c+\slepw]$, we note that the sample vector $\bx$ will equal $\sampe{f}$ for $f = F\tsamp$.
Without knowing the exact carrier frequency $F$ in advance, we ask whether it is possible to define an efficient low-dimensional basis for capturing the energy in any sample vector $\bx$ that could arise in this model.

At first glance, this problem may appear difficult or impossible.
The infinite set of possible sample vectors $\{\sampe{f} \}_{f \in [f_c-\slepw,f_c+\slepw]}$ traverses a $1$-dimensional manifold (parameterized by $f$) within $\complex^N$.
Technically speaking, these vectors collectively span all of $\complex^N$.
What is remarkable, however, is that to a high degree of accuracy, these vectors are approximately concentrated along a very low-dimensional subspace of $\complex^N$.
Moreover, for any $k \in \{1,2,\dots,N\}$, it is possible to analytically find the best $k$-dimensional subspace that minimizes the average squared distance from the vectors $\sampe{f}$ to the subspace, and this subspace is spanned by the first $k$ modulated DPSS vectors.

To see this, we let $\bQ$ denote a subspace of $\complex^N$ and $\bP_\bQ$ denote the orthogonal projection operator onto $\bQ$. We would like to minimize \begin{equation}
\frac{1}{\bband}  \cdot \int_{F_c-\frac{\bband}{2}}^{F_c+\frac{\bband}{2}} \| \sampe{F\tsamp} - \bP_\bQ \sampe{F\tsamp} \|_2^2 \; dF =
\frac{1}{2\slepw} \cdot \int_{f_c-\slepw}^{f_c+\slepw} \| \sampe{f} - \bP_\bQ \sampe{f} \|_2^2 \; df
\label{eq:obj}
\end{equation}
over all subspaces $\bQ$ of some prescribed dimension $k$. As we show below in Theorem~\ref{thm:sampsinusoid}, this minimization problem can be solved by relating it to the Karhunen-Loeve (KL) transform.\footnote{The observation that a connection exists between DPSS's and the KL transform is not a new one (see, for example,~\cite{scharf1991svd,hein1994theoretical,mugler99linear,fancourt2000relationship}).}
For the benefit of the reader, we briefly review the relevant concepts from the KL transform in Appendix~\ref{sec:klreview}.

\begin{thm} \label{thm:sampsinusoid}
For any $k$ with $k \in \{1,2,\dots,N\}$, the $k$-dimensional subspace which minimizes~\eqref{eq:obj} is spanned by the columns of $\bQ = [\bE_{f_c} \bS_{N,\slepw}]_k$, i.e., by the first $k$ DPSS vectors modulated to center frequency $f_c$. Furthermore, with this choice of $\bQ$, we will have
$$
\frac{1}{2\slepw} \cdot \int_{f_c-\slepw}^{f_c+\slepw} \| \sampe{f} - \bP_\bQ \sampe{f} \|_2^2 \; df = \frac{1}{2\slepw} \sum_{\ell = k}^{N-1} \lambda_{N,\slepw}^{(\ell)}.
$$
For a point of comparison, each sampled sinusoid has energy $\| \sampe{f} \|_2^2 = N$.
\end{thm}
\begin{proof}See Appendix~\ref{pf:sampsinusoid}.\end{proof}

Recall that the first $\approx 2N\slepw$ DPSS eigenvalues are very close to $1$ and the rest are small, so we are guaranteed a high degree of approximation to the sampled sinusoids if we choose $k \approx 2N\slepw$. In particular, if we choose $k = 2N\slepw(1-\epsilon)$, it follows from Corollary~\ref{cor:smallk} that
$$
\frac{1}{2\slepw} \sum_{\ell = k}^{N-1} \lambda_{N,\slepw}^{(\ell)} \le N (\epsilon + C_1 e^{-C_2 N})
$$
for $N \ge N_0$. Alternatively, if we choose $k = 2N\slepw(1+\epsilon)$, it follows from Corollary~\ref{cor:bigk} that
$$
\frac{1}{2\slepw} \sum_{\ell = k}^{N-1} \lambda_{N,\slepw}^{(\ell)} \le N \min\left(  \epsilon + C_1 e^{-C_2 N}, \; \frac{C_3}{2\slepw} e^{-C_4 N} \right)
$$
for $N \ge \max(N_0,N_1)$. Since each sampled sinusoid has energy $\| \sampe{f} \|_2^2 = N$, for the subspace we have chosen (say with $k = 2N\slepw(1-\epsilon)$), the average relative approximation error across the sampled sinusoids is bounded by a small fraction $\epsilon + C_1 e^{-C_2 N}$ of the energy of a given sinusoid.

It is important to note that, while we are guaranteed a very high degree of approximation accuracy in an MSE sense, we are not guaranteed such accuracy uniformly over {\em all} sampled sinusoids in our band of interest. A relatively small number of sinusoids may have higher values of $\| \sampe{f} - \bP_\bQ \sampe{f} \|_2^2$, and in practice this diminished approximation performance tends to occur for those sinusoids with frequencies near the edge of the band (i.e., for $f$ near $f_c \pm \slepw$).

\begin{figure}
   \centering
   \begin{tabular}{cc}
   \hspace*{-.25in} \includegraphics[width=\imgwidth]{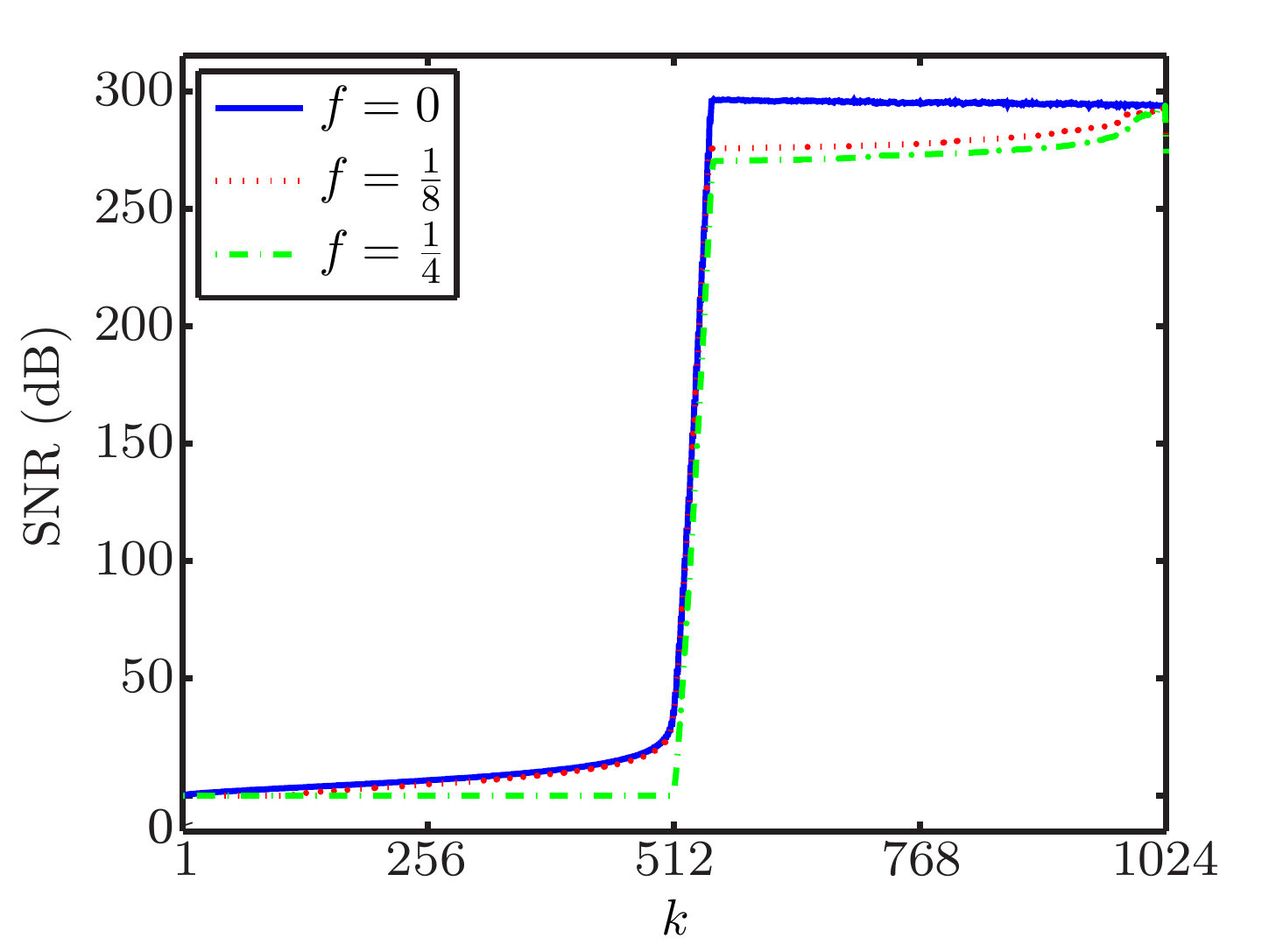} &    \includegraphics[width=\imgwidth]{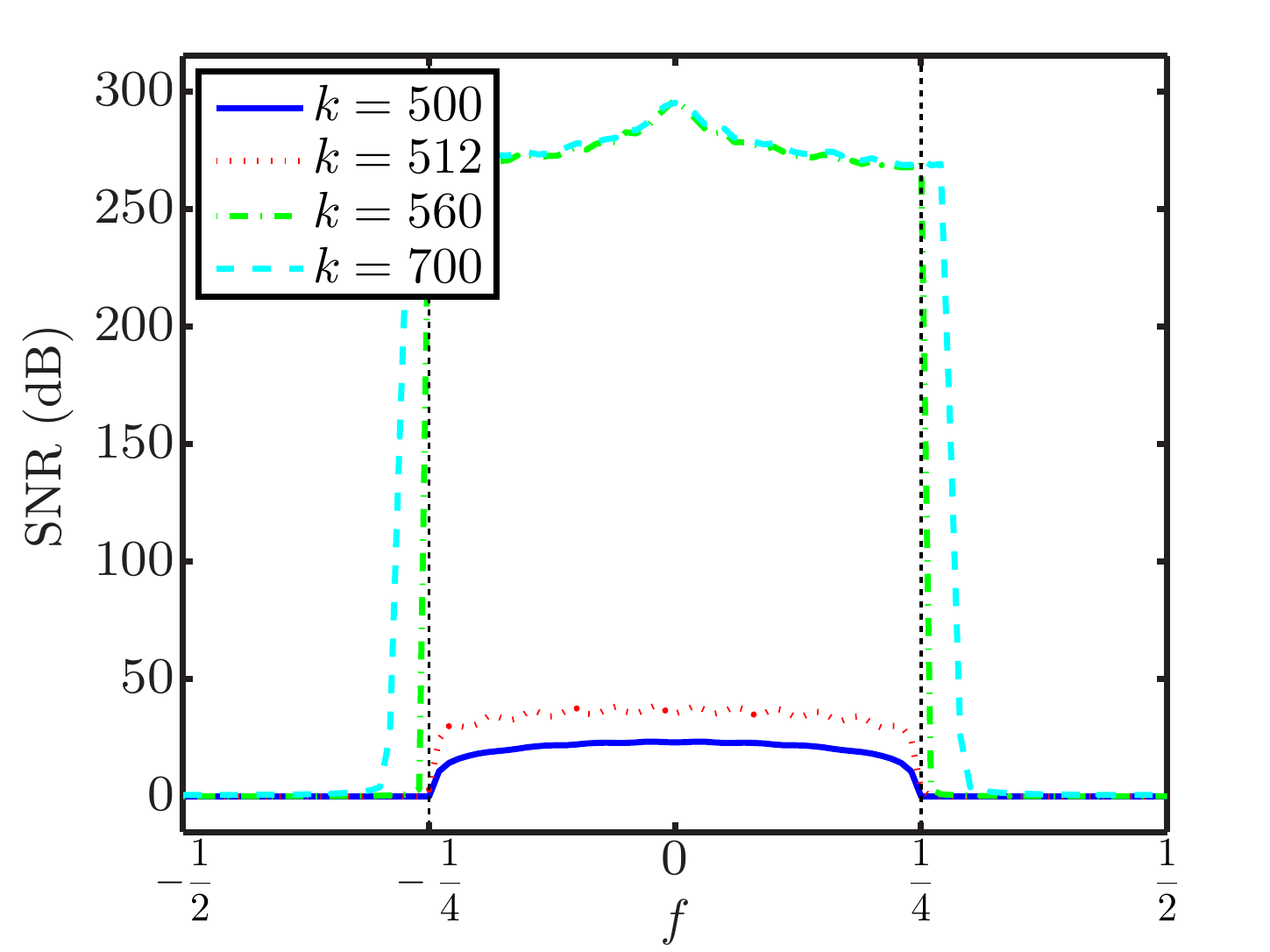} \\
   {\small (a)} & \hspace*{.15in} {\small (b)}
   \end{tabular}
   \caption{\small \sl DPSS bases efficiently capture the energy in sampled complex exponentials. (a) SNR captured from three sampled complex exponentials (with differing frequencies $f$) as a function of the number $k$ of DPSS basis elements.  (b) SNR captured as a function of $f$ for four fixed values of $k$. In both plots, $N=1024$ and $\slepw = \frac{1}{4}$.
   \label{fig:dpss_approx}}
\end{figure}

This behavior is illustrated in Figure~\ref{fig:dpss_approx}.  In Figure~\ref{fig:dpss_approx}(a) we consider three possible frequencies ($f=0$, $\frac{1}{8}$, or $\frac{1}{4}$) and show the ability of the baseband DPSS basis (with $f_c = 0$ and $W = \frac{1}{4}$) to capture the energy in these sinusoids as a function of how many DPSS vectors are added to the basis.  The ability of a basis to capture a given signal is quantified through
$$
\mathrm{SNR} = 20\log_{10} \left( \frac{\norm{\sampe{f}}}{\norm{\sampe{f} - \bP_\bQ \sampe{f}}} \right) \mathrm{dB}.
$$
Overall, we observe broadly similar behavior for each frequency in that by adding slightly more than $2N\slepw$ DPSS vectors to our basis, we capture essentially all of the energy in the original signal.  However, we do observe slightly different behavior for the sinusoid with a frequency exactly at $\slepw$.  In this case we capture very little of the energy in the signal until we have added more than $2N\slepw$ vectors, while for lower frequencies we begin to do well before this point.

To illustrate this phenomenon, Figure~\ref{fig:dpss_approx}(b) considers four different sizes for the DPSS basis and shows the SNR as a function of the frequency of the sinusoid.  In the cases where $k=500 < 2N\slepw$ and $k=2N\slepw$ we see somewhat similar behavior---we are capturing a good fraction, but not all, of any sinusoid whose frequency is not too close to $\slepw$.  We see a dramatic difference when we increase $k$ slightly to $560$, at which point we are capturing virtually all of the energy in any sinusoid within the band of interest.  However, this eventually has a potentially problematic side-effect, as we can see by further increasing $k$ to $700$.  Specifically, as we continue to increase the size of the DPSS basis we begin to capture energy of sinusoids that lie outside the targeted band as well.  This tradeoff will play an important role in the selection of the appropriate $k$ in the algorithms we propose below.

\subsubsection{Approximation quality for sampled bandpass signals}
\label{ssec:DPSSbandpass}

The above analysis shows that sampled sinusoids, on average, are well-approximated by modulated DPSS bases. This is a strong indication that such bases might also be useful for approximating more general sampled bandpass signals, since the vectors $\{\sampe{f} \}_{f\in[f_c-\slepw,f_c+\slepw]} \subset \complex^N$ themselves act as ``building blocks'' for representing sampled, bandpass signals in $\complex^N$.  Formally, for any continuous-time bandpass signal $x(t)$ with frequency content restricted to the interval $[F_c-\frac{\bband}{2}, F_c+\frac{\bband}{2}]$, one can show that for each $n \in \{0,1,\dots,N-1\}$,
$$
x[n] := x(n\tsamp) = \int_{F_c-\frac{\bband}{2}}^{F_c+\frac{\bband}{2}} X(F) e^{j 2\pi F n \tsamp} \; dF = \int_{f_c-\slepw}^{f_c+\slepw} X(f) \sampe{f}[n] \; df,
$$
where we recall that $X(F)$ denotes the CTFT of $x(t)$ and $X(f)$ denotes the DTFT of $x[n]$.
So, informally, $\bx$ can be expressed as a linear combination of (infinitely many) sampled complex exponentials $\sampe{f}$, where $f$ ranges from $f_c-\slepw$ to $f_c+\slepw$.

Our analysis from Section~\ref{ssec:DPSSsinusoid} allows us to show that in a certain sense, most continuous-time bandpass signals, when sampled and time-limited, are well-approximated by the modulated DPSS basis.
In particular, the following result establishes that bandpass random processes, when sampled and time-limited, are in expectation well-approximated.

\begin{thm}
Let $x(t)$ denote a continuous-time, zero-mean, wide sense stationary random process with power spectrum
\begin{equation*}
P_x(F) = \left\{ \begin{array}{ll} \frac{1}{\bband} ,&  F \in [F_c-\frac{\bband}{2}, F_c+\frac{\bband}{2}], \\ 0, & \mathrm{otherwise}, \end{array} \right.
\end{equation*}
and let $\bx = [x(0) ~ x(\tsamp) ~ \cdots ~ x((N-1)\tsamp)]^T \in \complex^N$ denote a finite vector of samples acquired from $x(t)$ with a sampling interval $\tsamp \le 1/(2 \max\{|F_c\pm\frac{\bband}{2}|\})$.
Then over all $k$-dimensional subspaces of $\complex^N$, $\bx$ is best approximated (in an MSE sense) by the subspace spanned by the columns of $\bQ = [\bE_{f_c} \bS_{N,\slepw}]_k$, where $f_c = F_c \tsamp$ and $\slepw= \frac{\bband\tsamp}{2}$.
The corresponding MSE is given by
\begin{equation}
\expval{\| \bx - \bP_\bQ \bx \|_2^2} = \frac{1}{2\slepw} \sum_{\ell = k}^{N-1} \lambda_{N,\slepw}^{(\ell)},
\label{eq:rprocmse}
\end{equation}
while $\expval{\| \bx \|_2^2} = N$.
\label{thm:rproc}
\end{thm}
\begin{proof}See Appendix~\ref{pf:rproc}.\end{proof}

As in our discussion following Theorem~\ref{thm:sampsinusoid}, we can ensure that the MSE in (\ref{eq:rprocmse}) is small compared to  $\expval{\| \bx \|_2^2}$ by choosing $k \approx 2N\slepw$. This suggests that in a probabilistic sense, most bandpass signals, when sampled and time-limited, will be well-approximated by a small number of modulated DPSS vectors. Again, however, we are not guaranteed such accuracy uniformly over {\em all} sampled bandpass signals in our band of interest.  A relatively small number of bandpass signals $x(t)$ could lead to sample vectors $\bx$ with higher values of $\| \bx - \bP_\bQ \bx \|_2^2$. In particular, recalling that the baseband DPSS's are themselves strictly bandlimited, it follows that there exist strictly bandpass signals that when sampled and time-limited yield the modulated DPSS vectors. If we restrict $\bQ$ to the first $k$ columns of $\bE_{f_c} \bS_{N,\slepw}$, then any bandpass signal producing a sample vector $\bx = \bE_{f_c} \bs_{N,\slepw}^{(\ell)}$ with $\ell \ge k$ will have $\bP_\bQ \bx = \bzero$ and $\| \bx - \bP_\bQ \bx \|_2^2 = \| \bx\|_2^2 $. Fortunately, Theorem~\ref{thm:rproc} confirms that such bandpass signals are relatively uncommon: at the risk of belaboring this important point, most bandpass signals, when sampled and time-limited, produce sample vectors approximately in the span of just the first $k \approx 2N\slepw$ modulated DPSS vectors.

On a related note, signal processing engineers often have a sense for how much ``spectral leakage'' to anticipate when collecting a finite window of samples of a continuous-time signal.
(Frequently, this leakage is reduced via a smooth windowing function~\cite{Lyons_Understanding}.)
Such practitioners can rest assured that, in every case where the spectral leakage is small outside a bandpass range of frequencies, the resulting sample vector can be well-approximated by a small number of modulated DPSS vectors.
\begin{thm} \label{thm:sampbandpass}
Let $x[n] = \cT_N(x[n])$ be a time-limited sequence, and suppose that $x[n]$ is approximately bandlimited to the frequency range $f \in [f_c-\slepw,f_c+\slepw]$ such that for some $\delta$,
$$
\| \cB_{f_c,\slepw} x \|^2_2 \ge (1-\delta) \|x\|_2^2,
$$
where $\cB_{f_c,\slepw}$ denotes an orthogonal projection operator that takes a discrete-time signal, bandlimits its DTFT to the frequency range $f \in [f_c-\slepw,f_c+\slepw]$, and returns the corresponding signal in the time domain.
Let $\bx \in \complex^N$ denote the vector formed by restricting $x[n]$ to the indices $n = 0,1,\dots,N-1$.
Set $k = 2N\slepw(1+\epsilon)$ and let $\bQ = \left[ \bE_{f_c} \bS_{N,\slepw} \right]_k$.
Then for $N \ge N_1$,
\begin{equation}
\| \bx - \bP_\bQ \bx \|_2^2 \le (\delta + N C_3 e^{- C_4 N})\|\bx\|_2^2,
\label{eq:sampbandpassdet}
\end{equation}
where $C_3$, $C_4$, and $N_1$ are as specified in Lemma~\ref{lem:bigk}.
\end{thm}
\begin{proof} See Appendix~\ref{pf:sampbandpass}. \end{proof}

Note that, in contrast to Theorem~\ref{thm:rproc}, Theorem~\ref{thm:sampbandpass} requires $k$ to be slightly larger than $2N\slepw$.

\subsection{DPSS dictionaries for sampled multiband signals}
\label{ssec:DPSSmultiband}

\subsubsection{A multiband modulated DPSS dictionary}

In order to construct an efficient dictionary for sampled multiband signals, we propose simply to concatenate an ensemble of modulated DPSS bases, each one modulated to the center of a band in our model. In particular, in light of the multiband signal model discussed in Section~\ref{ssec:multiband}, where the bandwidth $[-\frac{\bnyq}{2},\frac{\bnyq}{2}]$ is partitioned into bands $\Delta_i$ of size $\bband$, let us define the midpoint of $\Delta_i$ as
$$
F_i = -\frac{\bnyq}{2}+ \left(i+\frac{1}{2}\right)\bband, ~ i \in \{0,1,\dots,J-1\},
$$
where $J = \frac{\bnyq}{\bband}$. Let $\slepw= \frac{\bband\tsamp}{2}$ (we assume a sampling interval $\tsamp \le \frac{1}{\bnyq}$), and for each $i$, let $f_i = F_i \tsamp$ and define
\begin{equation}
\bPsi_i = \left[ \bE_{f_i} \bS_{N,\slepw} \right]_k
\label{eq:qidef}
\end{equation}
for some value of $k$ that we can choose as desired.
We construct the multiband modulated DPSS dictionary $\bPsi$ by concatenating all of the $\bPsi_i$:
\begin{equation} \label{eq:Qconcat}
\bPsi = \left[ \bPsi_0 ~ \bPsi_1 ~ \cdots ~ \bPsi_{J-1} \right].
\end{equation}
The matrix $\bPsi$ will have size $N \times k J$ (note that if $k = 2N\slepw$ and $\tsamp = \frac{1}{\bnyq}$, $\bPsi$ will be square).

\subsubsection{Approximation quality for sampled multiband signals}

In a probabilistic sense, most multiband signals, when sampled and time-limited, will be well-approximated by a small number of vectors from the multiband modulated DPSS dictionary. In particular, {\em there exists a block-sparse approximation} for such sample vectors using only the modulated DPSS vectors corresponding to the active signal bands.

\begin{thm}
Let $\cI \subseteq \{0,1,\dots,J-1\}$ with $|\cI| = K$. Suppose that for each $i \in \cI$, $x_i(t)$ is a continuous-time, zero-mean, wide sense stationary random process with power spectrum
\begin{equation*}
P_{x_i}(F) = \left\{ \begin{array}{ll} \frac{1}{K \bband} ,&  F \in \Delta_i, \\
0, & \mathrm{otherwise}, \end{array} \right.
\end{equation*}
and furthermore suppose the $\{x_i\}_{i \in \cI}$ are independent and jointly wide sense stationary. Let $x(t) = \sum_{i \in \cI} x_i(t)$, and let $\bx = [x(0) ~ x(\tsamp) ~ \cdots ~ x((N-1)\tsamp)]^T \in \complex^N$ denote a finite vector of samples acquired from $x(t)$ with a sampling interval of $\tsamp \le \frac{1}{\bnyq}$.
Let $\bPsi_{\cI}$ denote the concatenation of the $\bPsi_i$ over all $i \in \cI$, where the $\bPsi_i$ are as defined in~\eqref{eq:qidef}.
Then
\begin{equation}
\expval{\| \bx - \bP_{\bPsi_{\cI}} \bx \|_2^2} \le \frac{K}{2\slepw} \sum_{\ell = k}^{N-1} \lambda_{N,\slepw}^{(\ell)},
\label{eq:mbmse}
\end{equation}
whereas $\expval{\| \bx \|_2^2} = N$.
\label{thm:rprocmb}
\end{thm}
\begin{proof}See Appendix~\ref{pf:rprocmb}.\end{proof}

Theorem~\ref{thm:rprocmb} confirms the existence of a high-quality block-sparse approximation for most sampled multiband signals; in particular, the signal approximation vector specified in~\eqref{eq:mbmse} can be written as $\bP_{\bPsi_{\cI}} \bx = \bPsi \balpha$ for a $K$-block-sparse coefficient vector $\balpha$ given by $\balpha|_{\cI} = \bPsi_{\cI}^\dagger \bx$ and $\balpha|_{\cI^c} = \bzero$.
As in our previous discussions for bandpass signals, we can ensure the MSE in~\eqref{eq:mbmse} is small compared to $\expval{\| \bx \|_2^2}$ by choosing $k \approx 2N\slepw$. Compared to previous analysis, however, the MSE appearing in Theorem~\ref{thm:rprocmb} is larger by a factor of $K$ (though this quantity may still be quite small). Although it may be possible to improve upon this figure, the reader should keep in mind that the multiband modulated DPSS dictionary is not necessarily the {\em optimal} basis for representing sampled multiband signals with a given sparsity pattern $\cI$; it is merely a generic (and easily computable) dictionary that provides highly accurate approximations for most multiband signals having any possible sparsity pattern.

Although we omit the details here, one could also consider generalizing Theorem~\ref{thm:sampbandpass} to multiband signals that have small spectral leakage when windowed in time.

\section{Recovering Sampled Multiband Signals from Random Measurements}
\label{sec:recovery}

In this section, we proceed to develop theoretical guarantees for signal recovery using the multiband modulated DPSS dictionary $\bPsi$ as defined in~\eqref{eq:Qconcat}. Throughout our theoretical discussion and the subsequent experiments, we pay special attention to the role played the number $k$ of DPSS vectors per band. We begin in Section~\ref{sec:dpssrip} with a collection of RIP guarantees, and we extend these to signal recovery guarantees in Section~\ref{sec:recguar}. Throughout this section and the remainder of the paper, we assume that $\tsamp \le \frac{1}{\bnyq}$.

\subsection{Embedding guarantees}
\label{sec:dpssrip}

We can actually immediately establish $\bPsi$-RIP and $\bPsi$-block-RIP guarantees for any value of $k$. The theorem below follows as a direct consequence of Corollaries~\ref{cor:PsiRIP} and~\ref{cor:PsiblockRIP}.
\begin{thm}
Let $k \in \{1,2,\dots,N\}$, set $D = kJ = \frac{k\bnyq}{\bband}$, and let $\bPsi$ be the $N \times D$ multiband modulated DPSS dictionary defined in~\eqref{eq:Qconcat}. The following statements hold:
\begin{enumerate}
\item Fix $\delta, \beta \in (0,1)$, and let $\bA$ be an $M \times N$ subgaussian matrix with $M$ satisfying~\eqref{eq:Psi-RIP-M}. Then with probability exceeding $1 - \beta$, $\bA$ satisfies the $\bPsi$-RIP of order $S$ with constant $\delta$.
\item Fix $\delta, \beta \in (0,1)$, and let $\bA$ be an $M \times N$ subgaussian matrix with $M$ satisfying~\eqref{eq:Psi-blockRIP-M}. Then with probability exceeding $1 - \beta$, $\bA$ satisfies the $\bPsi$-block-RIP of order $K$ with constant $\delta$.
\end{enumerate}
\label{thm:DPSSPsiRIPandPsiblockRIP}
\end{thm}

In order to establish RIP and block-RIP bounds for $\bA \bPsi$, however, we must restrict our attention to values of $k$ that are not too large (this ensures the matrix $\bPsi$ is not overcomplete). To be specific, we note that for any $k$, the columns of $\bPsi$ have unit norm. In addition, when $k$ is suitably small, the columns of $\bPsi$ are approximately orthogonal.

\begin{lemma}
Let $k = 2N\slepw(1-\epsilon)$, set $D = kJ = N(1-\epsilon)\bnyq\tsamp < N$, and let $\bPsi$ be the $N \times D$ multiband modulated DPSS dictionary defined in~\eqref{eq:Qconcat}. Then for any pair of distinct columns $\bq_1$, $\bq_2$ in $\bPsi$,
\begin{equation}
| \langle \bq_1, \bq_2 \rangle | \le 3 C_1^{1/2} e^{-\frac{C_2 N}{2}}
\label{eq:coherencemain}
\end{equation}
if $N \ge N_0$, where $C_1$, $C_2$, and $N_0$ are as specified in Lemma~\ref{lem:smallk}.
\label{lem:coherence}
\end{lemma}
\begin{proof}See Appendix~\ref{pf:coherence}.\end{proof}

Using this fact, we can ensure that whenever $k = 2N\slepw(1-\epsilon)$, $\bPsi$ must act as an approximate isometry between any coefficient vector $\balpha \in \complex^{D}$ and the corresponding signal vector $\bx = \bPsi \balpha \in \complex^N$.

\begin{lemma}
Let $k = 2N\slepw(1-\epsilon)$, set $D = kJ = N(1-\epsilon)\bnyq\tsamp < N$, and let $\bPsi$ be the $N \times D$ multiband modulated DPSS dictionary defined in~\eqref{eq:Qconcat}. Then
\begin{equation}
\sqrt{1-3 N C_1^{1/2} e^{-\frac{C_2 N}{2}}} \le \frac{\norm{\bPsi \balpha}}{\norm{\balpha}} \le \sqrt{1+3 N C_1^{1/2} e^{-\frac{C_2 N}{2}}}
\label{eq:PsiIsom}
\end{equation}
for all $\balpha \in \complex^{D}$.
\label{lem:PsiIsom}
\end{lemma}
\begin{proof}
The sharpest possible lower and upper bounds in~\eqref{eq:PsiIsom} are given by the smallest and largest singular values of $\bPsi$, respectively. Using standard results from linear algebra, $\sigma_{\mathrm{min}}(\bPsi) =\sqrt{\lambda_{\mathrm{min}}(\bPsi^H \bPsi)}$ and $\sigma_{\mathrm{max}}(\bPsi) = \sqrt{\lambda_{\mathrm{max}}(\bPsi^H \bPsi)}$. The Gram matrix $\bPsi^H \bPsi$ has size $D \times D$. All entries on the main diagonal of $\bPsi^H \bPsi$ are equal to $1$, and all entries off of the main diagonal can be bounded using~\eqref{eq:coherencemain}. From the Ger\v{s}gorin circle theorem~\cite{Gerv_Uber}, it follows that all eigenvalues of $\bPsi^H \bPsi$ must fall in the interval $[1-3 D C_1^{1/2} e^{-\frac{C_2 N}{2}},1+3 D  C_1^{1/2} e^{-\frac{C_2 N}{2}}]$, which for simplicity we note is contained within the interval $[1-3 N C_1^{1/2} e^{-\frac{C_2 N}{2}},1+3 N C_1^{1/2} e^{-\frac{C_2 N}{2}}]$ since by assumption $D < N$. 
\end{proof}

Lemma~\ref{lem:PsiIsom} is the key fact we need to establish RIP and block-RIP bounds for $\bA \bPsi$.

\begin{thm} Let $k = 2N\slepw(1-\epsilon)$, set $D = kJ = N(1-\epsilon)\bnyq\tsamp < N$, and let $\bPsi$ be the $N \times D$ multiband modulated DPSS dictionary defined in~\eqref{eq:Qconcat}. The following statements hold:
\begin{enumerate}
\item If $\bA$ satisfies the $\bPsi$-RIP of order $S$ with constant $\delta$, then $\bA \bPsi$ satisfies the RIP of order $S$ with constant $\delta + 6 N C_1^{1/2} e^{-\frac{C_2 N}{2}}$.
\item If $\bA$ satisfies the $\bPsi$-block-RIP of order $K$ with constant $\delta$, then $\bA \bPsi$ satisfies the block-RIP of order $K$ with constant $\delta + 6 N C_1^{1/2} e^{-\frac{C_2 N}{2}}$.
\end{enumerate}
\label{thm:DPSSRIPandDPSSblockRIP}
\end{thm}
\begin{proof}
For any $\balpha \in \complex^D$ such that~\eqref{Psi-RIP} holds, we can use~\eqref{eq:PsiIsom} to conclude that
$$
\sqrt{1-\delta} \cdot \sqrt{1-3 N C_1^{1/2} e^{-\frac{C_2 N}{2}}} \le \frac{\norm{\bA \bPsi \balpha}}{\norm{\balpha}} \le \sqrt{1+\delta} \cdot \sqrt{1+3 N C_1^{1/2} e^{-\frac{C_2 N}{2}}}.
$$
For the upper bound, note that
\begin{eqnarray*}
\left( 1+\delta \right) \left( 1+3 N C_1^{1/2} e^{-\frac{C_2 N}{2}} \right) & = & 1 + \delta + 3 N C_1^{1/2} e^{-\frac{C_2 N}{2}} + \delta 3 N C_1^{1/2} e^{-\frac{C_2 N}{2}}\\
&\le& 1 + \delta + 6 N C_1^{1/2} e^{-\frac{C_2 N}{2}},
\end{eqnarray*}
where the second line follows from the assumption that $\delta < 1$. The lower bound follows similarly.
\end{proof}

\subsection{Recovery guarantees}
\label{sec:recguar}

In this section we prove that with a sufficient number of measurements and an appropriately constructed multiband modulated DPSS dictionary, most sample vectors of multiband signals can be accurately reconstructed. Our proof of this fact relies on two principles.

\subsubsection{Recovery of exactly block-sparse signals}

The first of these principles is that, as a consequence of our RIP results in Section~\ref{sec:dpssrip}, signal vectors having representations that are exactly $K$-block-sparse in the dictionary $\bPsi$ can be accurately reconstructed from compressive samples. The following two results follow from combining Theorems~\ref{thm:blockiht}, \ref{thm:blockcosamp}, \ref{thm:DPSSPsiRIPandPsiblockRIP}, and~\ref{thm:DPSSRIPandDPSSblockRIP}.

\begin{thm} \label{thm:blockihtdpss}
Let $k \in \{1,2,\dots,N\}$, set $D = kJ = \frac{k\bnyq}{\bband}$, and let $\bPsi$ be the $N \times D$ multiband modulated DPSS dictionary defined in~\eqref{eq:Qconcat}. Fix $\delta \in (0,0.2)$ and $\beta \in (0,1)$, and let $\bA$ be an $M \times N$ subgaussian matrix with
\begin{equation}
M \ge  \frac{ 4 K \left(\frac{D}{J} \log (42/ \delta) + \log (e J/2K)\right) + \log(4/\beta) }{c_1(\delta/\sqrt{2})}.
\label{eq:m4k}
\end{equation}
Then with probability exceeding $1 - \beta$, the following statement holds: For any $\bx' \in \complex^N$ that has a $K$-block-sparse representation in the dictionary $\bPsi$ (i.e., that can be written as $\bx' = \bPsi \balpha'$ for some $K$-block-sparse vector $\balpha' \in \complex^D$), if we use block-based IHT~\eqref{eq:IHTblock} with $\frac{1}{\mu} \in [1+\delta,1.5(1-\delta))$ to recover an estimate $\widehat{\bx}$ of $\bx'$ from the observations $\by = \bA \bx' + \be$, the resulting $\widehat{\bx}$ will satisfy
\begin{equation} \label{eq:bbihtthmdpss}
\norm{\bx' - \widehat{\bx}} \le \kappa_1 \norm{\be},
\end{equation}
where $\kappa_1>1$ is as specified in Theorem~\ref{thm:blockiht}.
\end{thm}

\begin{thm} \label{thm:blockcosampdpss}
Let $k = 2N\slepw(1-\epsilon)$, set $D = kJ = N(1-\epsilon)\bnyq\tsamp < N$, and let $\bPsi$ be the $N \times D$ multiband modulated DPSS dictionary defined in~\eqref{eq:Qconcat}. Fix $\delta \in (0,0.1 - 6 N C_1^{1/2} e^{-\frac{C_2 N}{2}})$ and $\beta \in (0,1)$. Let $\bA$ be an $M \times N$ subgaussian matrix with 
\begin{equation}
M \ge  \frac{ 8 K \left(\frac{D}{J} \log (42/ \delta) + \log (e J/4K)\right) + \log(4/\beta) }{c_1(\delta/\sqrt{2})}
\label{eq:m8k}
\end{equation}
Then with probability exceeding $1 - \beta$, the following statement holds: For any $\bx' \in \complex^N$ that has a $K$-block-sparse representation in the dictionary $\bPsi$ (i.e., that can be written as $\bx' = \bPsi \balpha'$ for some $K$-block-sparse vector $\balpha' \in \complex^D$), if we use block-based CoSaMP (Algorithm~\ref{alg:bbcosamp}) with $\widehat{\balpha} = \mathrm{BBCoSaMP}(\bA \bPsi, \bI, \by, K)$ to recover an estimate $\widehat{\balpha}$ of $\balpha'$ from the observations $\by = \bA \bx' + \be$, the resulting $\widehat{\balpha}$ will satisfy
\begin{equation} \label{eq:bbcosampthmdpss}
\norm{\balpha' - \widehat{\balpha}} \le \kappa_2 \norm{\be}
\end{equation}
where $\kappa_2>1$ is as specified in Theorem~\ref{thm:blockcosamp}.
\end{thm}

\subsubsection{Approximating sampled multiband signals with exactly block-sparse signals}

The second of our principles is that most multiband signals with $K$ occupied bands, when sampled and time-limited, have a high-quality $K$-block-sparse representation in the dictionary $\bPsi$. Let $x(t)$ denote a continuous-time multiband signal with nonzero support on blocks indexed by $\cI \subseteq \{0,1,\dots,J-1\}$, where $|\cI| = K$. Let $\bx = [x(0) ~ x(\tsamp) ~ \cdots ~ x((N-1)\tsamp)]^T \in \complex^N$ denote a finite vector of samples acquired from $x(t)$ with a sampling interval of $\tsamp \le \frac{1}{\bnyq}$.
Let $\balpha'$ be a $K$-block-sparse coefficient vector given by $\balpha'|_{\cI} = \bPsi_{\cI}^\dagger \bx$ and $\balpha'|_{\cI^c} = \bzero$.
Defining $\bx' := \bPsi \balpha'$ and $\be_{\bx} := \bx - \bPsi \balpha' = \bx - \bx'$, we can then write
\begin{equation}
\bx = \bx' + \be_{\bx},
\label{eq:xxprime}
\end{equation}
where $\bx'$ has an exactly $K$-block-sparse representation in the dictionary $\bPsi$ and we expect $\be_{\bx}$ to be small.

We can more formally bound the size of $\be_{\bx}$. For example, under the multiband random process model for $x(t)$ described in Theorem~\ref{thm:rprocmb}, we will have $\expval{\| \be_{\bx} \|_2^2} \le \frac{K}{2 \slepw} \sum_{\ell = k}^{N-1} \lambda_{N,\slepw}^{(\ell)}$. By setting the number of columns per band $k$ to be on the order of $2N\slepw$, we can make this error small relative to $\expval{\| \bx \|_2^2} = N$. For example, if we take $k = 2N\slepw(1-\epsilon)$ for some $\epsilon \in (0,1)$, Corollary~\ref{cor:smallk} allows us to conclude that
\begin{equation}
\expval{\| \be_{\bx} \|_2^2} \le \frac{K}{2 \slepw} \sum_{\ell = k}^{N-1} \lambda_{N,\slepw}^{(\ell)} \le K N (\epsilon + C_1 e^{-C_2 N}).
\label{eq:etail1}
\end{equation}
for $N \ge N_0$. The rightmost upper bound in~\eqref{eq:etail1} can be made as small as desired (relative to $\expval{\| \bx \|_2^2} = N$) by choosing $\epsilon$ sufficiently small and $N$ sufficiently large.

For any value of $k$, we can also establish a tail bound on $\| \be_{\bx} \|_2^2$, guaranteeing that it is unlikely to significantly exceed the quantity $\frac{K}{2 \slepw} \sum_{\ell = k}^{N-1} \lambda_{N,\slepw}^{(\ell)}$. Using standard concentration of measure arguments for subgaussian and subexponential random variables (see~\cite{vershynin2010introduction} for a thorough discussion), one can make the following guarantee on the squared norm of a Gaussian random vector.

\begin{lemma}
Let $\bz \in \complex^N$ be a complex-valued Gaussian random vector with mean zero. Then
$$
\prob{ \left| \| \bz \|_2^2 - \expval{ \| \bz \|_2^2} \right| \ge \gamma \expval{ \| \bz \|_2^2} }  \le 2 \; \mathrm{exp} \left\{ - \mathrm{min} \left( \frac{\gamma^2 (\sum_n \lambda_n)^2}{c_2^2 \sum_n \lambda_n^2}, \frac{\gamma \sum_n \lambda_n}{c_2 \max_n \lambda_n} \right) \right\},
$$
where $c_2$ is a universal constant and $\{\lambda_n\}$ denote the eigenvalues of the autocorrelation matrix of the length-$2N$ random vector $[\mathrm{Re}(\bz)^T \; \mathrm{Im}(\bz)^T]^T$.
\label{lem:tail}
\end{lemma}

If we assume in our multiband random process model that $x(t)$ is a Gaussian random process, this will imply that $\bx$ is a Gaussian random vector, and since $\be_{\bx}$ is a linear transformation of $\bx$, $\be_{\bx}$ will be Gaussian as well. This allows us to apply Lemma~\ref{lem:tail} to the quantity $\| \be_{\bx} \|_2^2$. Using the fact that $\norm[1]{\blambda} \ge \norm[2]{\blambda} \ge \norm[\infty]{\blambda}$ for any vector $\blambda$, we can pessimistically simplify this bound to read:
\begin{equation}
\prob{ \| \be_{\bx} \|_2^2 \ge (1 + \gamma) \frac{K}{2 \slepw} \sum_{\ell = k}^{N-1} \lambda_{N,\slepw}^{(\ell)} } \le 2 \; \mathrm{exp} \left\{ - \mathrm{min} \left( \frac{\gamma^2}{c_2^2}, \frac{\gamma}{c_2} \right) \right\}.
\label{eq:etail}
\end{equation}

\subsubsection{Combined guarantees}
\label{sec:cguar}

To put these two principles together, we note that when taking noise-free compressive measurements $\bA \bx$ of a signal vector $\bx$ obeying~\eqref{eq:xxprime}, we will have $\bA \bx = \bA \bx' + \bA \be_{\bx}$.
This allows us to invoke Theorems~\ref{thm:blockihtdpss} and~\ref{thm:blockcosampdpss} with $\be := \bA \be_{\bx}$,\footnote{One could easily incorporate actual measurement noise into $\be$ as well (and thus extend Theorems~\ref{thm:blockihtdpssfinal} and~\ref{thm:blockcosampdpssfinal} to the case of noisy measurements). However, for the sake of clarity we simply set $\be := \bA \be_{\bx}$ in order to highlight the impact of modeling error in our main results.} and when the number of columns per band $k$ is chosen so that we expect $\be_{\bx}$ to be small, the concentration of measure phenomenon tells us that $\be$ should be small as well.
In particular, note that for fixed $\be_{\bx}$ and random subgaussian $\bA$, \eqref{star} guarantees that
\begin{equation}
\prob{ \left| \norm[2]{\be}^2- \norm[2]{\be_{\bx}}^2\right| \ge \eta \norm[2]{\be_{\bx}}^2} \le  4 e^{-c_1(\eta) M}.
\label{eq:eex}
\end{equation}
Then, for example, if we let let $\widehat{\bx}$ denote the estimated signal vector recovered via block-based IHT, ~\eqref{eq:eex} allows us to write
$$
\| \bx  - \widehat{\bx} \|_2 \le \| \bx - \bx' \|_2 + \| \bx' - \widehat{\bx} \|_2 \le \| \be_{\bx} \|_2 + \kappa_1 \norm{\be} \le (1 + \kappa_1 (1+\eta)^{1/2}) \| \be_{\bx} \|_2,
$$
where the second inequality follows from Theorem~\ref{thm:blockihtdpss} and the third inequality holds with probability at least $1 - 4 e^{-c_1(\eta) M}$. Combining this fact with the tail bound~\eqref{eq:etail}, we can establish the following guarantee.

\begin{thm} \label{thm:blockihtdpssfinal}
Let $k \in \{1,2,\dots,N\}$, set $D = kJ = \frac{k\bnyq}{\bband}$, and let $\bPsi$ be the $N \times D$ multiband modulated DPSS dictionary defined in~\eqref{eq:Qconcat}. Fix $\delta \in (0,0.2)$ and $\beta \in (0,1)$, and let $\bA$ be an $M \times N$ subgaussian matrix with $M$ satisfying~\eqref{eq:m4k}.
If $x(t)$ is a Gaussian random process obeying the $K$-band model described in Theorem~\ref{thm:rprocmb}, $\bx \in \complex^N$ is generated by sampling $x(t)$ as in Theorem~\ref{thm:rprocmb}, and we use block-based IHT~\eqref{eq:IHTblock} with $\frac{1}{\mu} \in [1+\delta,1.5(1-\delta))$ to recover an estimate of $\bx$ from the observations $\by = \bA \bx$, then the resulting estimate $\widehat{\bx}$ will satisfy
\begin{equation} \label{eq:bbihtthmdpssfinal}
\norm{\bx - \widehat{\bx}}^2 \le (1 + \kappa_1 (1+\eta)^{1/2})^2 (1 + \gamma) \frac{K}{2 \slepw} \sum_{\ell = k}^{N-1} \lambda_{N,\slepw}^{(\ell)}
\end{equation}
with probability at least $1 - \beta - 4 e^{-c_1(\eta) M} - 2 e^{-\mathrm{min}(\gamma^2/c_2^2, \gamma/c_2)}$, where $\slepw= \frac{\bband\tsamp}{2}$ as specified in the dictionary construction~\eqref{eq:Qconcat}, $\kappa_1>1$ is a constant as specified in Theorem~\ref{thm:blockiht}, $c_1$ is a constant as specified in~\eqref{star}, and $c_2$ is a universal constant.
\end{thm}

Although the above result holds for any $k \in \{1,2,\dots,N\}$, it is perhaps most interesting when the number of columns per band $k$ is chosen to be on the order of $2N\slepw$.
For example, if we take $k = 2N\slepw(1-\epsilon)$,~\eqref{eq:bbihtthmdpssfinal} will read
\begin{equation} \label{eq:bbihtthmdpssfinaloneminus}
\norm{\bx - \widehat{\bx}}^2 \le (1 + \kappa_1 (1+\eta)^{1/2})^2 (1 + \gamma) (\epsilon + C_1 e^{-C_2 N}) K N.
\end{equation}
Supposing that $K$, $\kappa_1$, $\eta$, and $\gamma$ are fixed, one can make the upper bound appearing in~\eqref{eq:bbihtthmdpssfinaloneminus} as small as desired (relative to $\expval{\| \bx \|_2^2} = N$) by choosing $\epsilon$ sufficiently small and $N$ sufficiently large.\footnote{One could also use  Lemma~\ref{lem:tail} to guarantee that, with high probability, $\|\bx\|_2^2$ will not be too small compared to $N$. We omit these details in the interest of conciseness.}
Specifically, the term $\epsilon + C_1 e^{-C_2 N}$ can be made arbitrarily close to zero by choosing $\epsilon$ sufficiently small and $N$ sufficiently large.
Of course, as one increases the length $N$ of the signal vector, the requisite number $M$ of measurements will increase proportionally (this is captured in~\eqref{eq:m4k} via the dependence on $D$).
What is important about the measurement bound is that $\frac{M}{N}$, the permitted undersampling ratio relative to the Nyquist sampling rate (supposing $\tsamp = \frac{1}{\bnyq}$), will scale like
$$
\frac{M}{N} \sim \frac{K \bband}{\bnyq}.
$$
In other words, we need only collect compressive samples at a rate proportional to $K\bband$, the total amount of occupied bandwidth (i.e., the so-called Landau rate~\cite{Landa_Necessary}).

For sufficiently small $k$, we can establish a similar guarantee for block-based CoSaMP. Let $\widehat{\balpha}$ be the recovered coefficient vector, and define $\widehat{\bx} := \bPsi \widehat{\balpha}$. Then we can write
\begin{eqnarray*}
\| \bx  - \widehat{\bx} \|_2 &\le& 
\| \bx - \bx' \|_2 + \| \bPsi \balpha' - \bPsi \widehat{\balpha} \|_2 \\
&\le& \| \be_{\bx} \|_2 + \sqrt{1+3 N C_1^{1/2} e^{-\frac{C_2 N}{2}}} \| \balpha' - \widehat{\balpha} \|_2 \\
&\le& \| \be_{\bx} \|_2 + \kappa_2 \sqrt{1+3 N C_1^{1/2} e^{-\frac{C_2 N}{2}}} \| \be \|_2 \\
&\le&  (1 + \kappa_2 (1+3 N C_1^{1/2} e^{-\frac{C_2 N}{2}})^{1/2} (1+\eta)^{1/2}) \| \be_{\bx} \|_2,
\end{eqnarray*}
where the second line follows from Lemma~\ref{lem:PsiIsom}, the third line follows from Theorem~\ref{thm:blockcosampdpss}, and the last line holds with probability at least $1 - 4 e^{-c_1(\eta) M}$. Combining this fact with Corollary~\ref{cor:smallk} and the tail bound~\eqref{eq:etail}, we can establish the following guarantee.

\begin{thm} \label{thm:blockcosampdpssfinal}
Let $k = 2N\slepw(1-\epsilon)$, set $D = kJ = N(1-\epsilon)\bnyq\tsamp < N$, and let $\bPsi$ be the $N \times D$ multiband modulated DPSS dictionary defined in~\eqref{eq:Qconcat}. Fix $\delta \in (0,0.1 - 6 N C_1^{1/2} e^{-\frac{C_2 N}{2}})$ and $\beta \in (0,1)$, and let $\bA$ be an $M \times N$ subgaussian matrix with $M$ satisfying~\eqref{eq:m8k}.
If $x(t)$ is a Gaussian random process obeying the $K$-band model described in Theorem~\ref{thm:rprocmb}, $\bx \in \complex^N$ is generated by sampling $x(t)$ as in Theorem~\ref{thm:rprocmb}, we use block-based CoSaMP (Algorithm~\ref{alg:bbcosamp}) with $\widehat{\balpha} = \mathrm{BBCoSaMP}(\bA \bPsi, \bI, \by, K)$ to recover an estimate of $\balpha$ from the observations $\by = \bA \bx$, and we set $\widehat{\bx} := \bPsi \widehat{\balpha}$, then the resulting estimate will satisfy
\begin{equation} \label{eq:bbcosampthmdpssfinal}
\norm{\bx - \widehat{\bx}}^2 \le (1 + \kappa_2 (1+3 N C_1^{1/2} e^{-\frac{C_2 N}{2}})^{1/2} (1+\eta)^{1/2})^2 (1 + \gamma) (\epsilon + C_1 e^{-C_2 N}) K N
\end{equation}
with probability at least $1 - \beta- 4 e^{-c_1(\eta) M} - 2 e^{-\mathrm{min}(\gamma^2/c_2^2, \gamma/c_2)}$, where $\kappa_2>1$ is a constant as specified in Theorem~\ref{thm:blockcosamp}, $C_1$ and $C_2$ are constants as specified in Lemma~\ref{lem:smallk}, $c_1$ is a constant as specified in~\eqref{star}, and $c_2$ is a universal constant.
\end{thm}

Once again, supposing that $K$, $\kappa_2$, $\eta$, and $\gamma$ are fixed, one can make the upper bound appearing in~\eqref{eq:bbcosampthmdpssfinal} as small as desired (relative to $\expval{\| \bx \|_2^2} = N$) by choosing $\epsilon$ sufficiently small and $N$ sufficiently large.
Specifically, the terms $3 N C_1^{1/2} e^{-\frac{C_2 N}{2}}$ and $\epsilon + C_1 e^{-C_2 N}$ can be made arbitrarily close to zero by choosing $\epsilon$ sufficiently small and $N$ sufficiently large.
Thus, we have again guaranteed that most finite-length sample vectors arising from multiband analog signals can---to a very high degree of approximation---be recovered from a number of compressive measurements that is proportional to the underlying information level.

\section{Simulations}
\label{sec:sims}

In this section we present the results of a suite of simulations that demonstrate the effectiveness of our proposed approaches to multiband signal recovery from compressive measurements.
In doing so, we address various practical considerations that arise within our framework.
All of our simulations were performed via a MATLAB software package that we have made available for download at \url{http://www.mines.edu/~mwakin/software/}.
This software package contains all of the code and results necessary to reproduce the experiments and figures described below, but should additionally serve as a platform upon which other researchers can test and develop their own extensions of these ideas.

\subsection{Implementation and experimental setup}

We begin with a brief discussion regarding our implementation and experimental setup.

\subsubsection{Computing $\cP(\bx,K)$}

We first recall that a key step in solving either block-based IHT (specifically, the variation in~\eqref{eq:IHTblock}) or block-based CoSaMP (specifically, the variation $\widehat{\bx} = \mathrm{BBCoSaMP}(\bA,\bPsi,\by,K)$) is the computation of $\cP(\bx,K)$, i.e., the projection of $\bx$ onto the set of $K$-block sparse signals in the dictionary $\bPsi$ (as defined in~\eqref{eq:blockproj}).
If $\bPsi$ were an orthonormal basis, this projection could be computed simply by taking $\bPsi^H \bx$ and setting to zero all but the $K$ blocks of this vector with the highest energy.
Unfortunately, the columns of the multiband modulated DPSS dictionary $\bPsi$ (as defined in~\eqref{eq:Qconcat}) are not orthogonal, and so this thresholding approach is not guaranteed to even approximate the correct solution.
In fact, when the number of columns per band $k$ exceeds $2N\slepw(1-\epsilon)$, $\bPsi$ will fail to act as an isometry (recall Lemma~\ref{lem:PsiIsom}), and when $k$ exceeds $2N\slepw$, $\bPsi$ will actually be overcomplete.
As we will see in our experiments, it can often be desirable to choose $k$ to be larger than $2N\slepw$, but this lack of isometry and this overcompleteness prevent us from applying any of the standard arguments from sparse approximation to solving the projection problem.

It is important to note that in~\eqref{eq:blockproj} we are seeking a vector $\bz$ that minimizes $\norm{\bx - \bz}$ and has a block-sparse representation in $\bPsi$.
Although $\bz$ will be unique, the corresponding representation in $\bPsi$ need not be unique if $\bPsi$ fails to satisfy the block-RIP.
Unfortunately, there seems to be relatively little known about solving this type of sparse approximation problem.
Nevertheless, since we are ultimately interested in $\bz$ (not its representation with respect to $\bPsi$) this should not necessarily stop us from proceeding under the hope that standard sparse approximation algorithms can still succeed even when $\bPsi$ is moderately overcomplete.
Thus, in our simulations, we use block-OMP to obtain an approximation to $\cP(\bx,K)$.
Block-OMP is a straightforward generalization of the classical OMP algorithm.
It proceeds by: {\em (i)}~initializing the residual vector $\br = \bx$ and the set $\cI = \emptyset$, {\em (ii)}~computing the proxy vector $\bh = \bPsi^H \br$, {\em (iii)}~identifying the block in $\bh$ that has the largest energy and adding this block to the set $\cI$, and {\em (iv)}~orthogonalizing $\br$ against the columns in $\bPsi_\cI$.
This last step is equivalent to the ``update'' step in CoSaMP in Algorithms~\ref{alg:cosamp} and~\ref{alg:bbcosamp}.
Steps {\em (ii)}~through {\em (iv)}~are repeated until termination.
The final output can be computed from $\bx - \br^{K}$, where $\br^{K}$ is the residual after $K$ iterations.

We close by noting that the lack of a provable technique for computing $\cP(\bx,K)$ represents an important gap between some of our theory and practice.
Specifically, one of our main results (Theorem~\ref{thm:blockihtdpssfinal}) pertains to block-based IHT~\eqref{eq:IHTblock}, but this is an algorithm that we can only implement approximately.
As noted at the end of Section~\ref{sec:mbguarantees}, we conjecture that one could also establish theoretical guarantees for the $\mathrm{BBCoSaMP}(\bA, \bPsi, \by, K)$ version of block-based CoSaMP, but this too relies on being able to compute $\cP(\bx,K)$ exactly.
In light of this, we point out the following.
First, we are able to implement the $\mathrm{BBCoSaMP}(\bA \bPsi, \bI, \by, K)$ version of block-based CoSaMP exactly because this problem can be solved with simple block thresholding of the vector $\balpha$, and one of our main results (Theorem~\ref{thm:blockcosampdpssfinal}) does pertain specifically to $\mathrm{BBCoSaMP}(\bA \bPsi, \bI, \by, K)$.
Second, our simulations in Sections~\ref{sec:bbcosamp12} and beyond, which rely on block-OMP for approximating $\cP(\bx,K)$, indicate that we are clearly finding high quality solutions to this problem despite the lack of provable guarantees.
What is interesting is that our experimental results seem most favorable when $k$ exceeds $2N\slepw$ by a nontrivial amount, and that is a regime where the dictionary $\bPsi$ is overcomplete.
We hope to further address the implications of overcomplete $\bPsi$ in future work.

\subsubsection{Regularized least-squares}
\label{ssec:regu}

In our simulations below, we focus exclusively on testing the two versions of block-based CoSaMP.
One could conceivably expect similar experimental results using a properly-tuned version of block-based IHT, but since we lack theory concerning $\cP(\bx,K)$ we find it more worthwhile to devote our space to exploring the various practical issues surrounding block-based CoSaMP.
Specifically, let us note that when solving block-based CoSaMP, the key steps in this algorithm consist of solving the least-squares problem in the ``update'' step and of computing $\cP(\bx,K)$ (twice), which also involves solving (potentially multiple) least-squares problems.
Thus, it is crucial that we solve these least-squares problems efficiently and accurately.

While there has been much work in the CS community on solving these problems efficiently (for example, see~\cite{NeedeT_CoSaMP}), we leave aside the issue of speed for the moment and focus instead on the issue of stability.  In our context, we must take some care in solving these problems, since when the number of columns $k$ per band becomes becomes much larger than $2N\slepw$, the matrices involved can become potentially rank deficient and thus standard approaches can be numerically unstable.  Fortunately, this instability can be easily addressed using relatively simple techniques.

We first consider the least-squares problem that must be solved in the computation of $\cP(\bx,K)$.
In this case we will not focus specifically on the block-OMP algorithm that we implement, but simply assume that we have some technique for obtaining an estimate $\cI$ of $\cS(\bx,K)$.
At that point, solving~\eqref{eq:blockproj} is relatively straightforward.
One could solve this via $\cP(\bx,K) = \bPsi_{\cI} \bPsi_{\cI}^{\dag} \bx$, but when $\cI$ contains indices corresponding to adjacent bands, $\bPsi_{\cI}$ can be nearly rank deficient.
In this case, a better approach is to construct a reduced basis $\bU$ for $\cR(\bPsi_{\cI})$.  One can then compute the projection simply via $\cP(\bx,K) = \bU \bU^H \bx$.
In our context, we use this reduced basis approach to perform the orthogonalization in each step of the block-OMP approach to computing $\cP(\bx,K)$.

We now turn to the main least-squares problem in the ``update'' step of CoSaMP.  In this step we are given a set of blocks $\cI$ and wish to solve a problem of the form
\begin{equation} \label{eq:ls_not}
\widetilde{\bx} = \argmin_{\bz}  \norm{\by - \bA \bz} \quad \mathrm{s.t.} \quad \bz \in \cR(\bPsi_{\cI}).
\end{equation}
Note that in this case we require more than simply the projection of $\by$ onto $\cR(\bA \bPsi_{\cI})$---we wish to actually calculate the vector $\bz$ corresponding to this projection.  This requires a different approach than the simple reduced basis approach described above (although constructing a reduced basis for $\cR(\bPsi_{\cI})$ can still be useful in this context).  In our simulations we regularize~\eqref{eq:ls_not} via {\em Tikhonov regularization}~\cite{Phill_technique,Tikho_Solution,TikhoA_Solutions}.  The key idea is to replace~\eqref{eq:ls_not} with
\begin{equation} \label{eq:ls_tikh}
\widetilde{\bx} = \argmin_{\bz}  \norm{\by - \bA \bz} \quad \mathrm{s.t.} \quad \bz \in \cR(\bPsi_{\cI}), \ \norm{\bz} \le \gamma.
\end{equation}
The constraint $\norm{\bz} \le \gamma$ significantly improves the conditioning of this problem when $\bA \bPsi_{\cI}$ is ill-conditioned.  In our simulations below, we use the toolbox of~\cite{Hanse_Regularization} to efficiently solve~\eqref{eq:ls_tikh}.  Note also that in our setting, the parameter $\gamma$ can be easily set using the norm of the original signal $\bx$ that we are acquiring.  In all simulations below (using either variation of block-based CoSaMP), we assume that an upper bound on $\norm{\bx}$ is known. In practical settings such an upper bound will usually be available, but if necessary one could also estimate this upper bound from $\norm{\by}$.

\subsubsection{Experimental setup}

In all of the experiments below, we assume that $\tsamp = \frac{1}{\bnyq}$ and that there are $J = \frac{\bnyq}{\bband} = 256$ possible bands.
For each value of $k$ that we consider, we set $D = kJ$ and let $\bPsi$ be the $N \times D$ multiband modulated DPSS dictionary defined in~\eqref{eq:Qconcat}.
The digital half-bandwidth parameter $\slepw$ is set to be $\slepw = \frac{\bband\tsamp}{2} = \frac{1}{512}$, and we consider sample vectors with length $N = 4096$, so that $2N\slepw = 16$.

In our experiments we generate our sampled multiband signal vectors $\bx$ by selecting the positions of the $K$ occupied bands uniformly at random from the $J$ possibilities, and then within each band adding together 50 complex exponentials with frequencies selected uniformly at random from within the frequency band (not aligned with the ``Nyquist grid'').
Each complex exponential in this summation is given a random amplitude and phase via multiplication by a complex Gaussian random variable.
There are a variety of other possibilities for generating test signals which we considered and for which we have observed essentially the same results as presented below.

We will typically report our results in terms of the SNR of the recovery, which in this context is defined as
$$
\mathrm{SNR} = 20\log_{10} \left( \frac{\norm{\bx}}{\norm{\bx  - \widehat{\bx}}} \right) \mathrm{dB},
$$
where $\bx$ is the original input to the measurement matrix and $\widehat{\bx}$ is the recovered estimate of $\bx$ provided by CoSaMP.
In all cases where we plot the SNR, what we actually show is the result of 50 independent trials (with a new signal for each trial).
Rather than plotting the mean SNR, we plot the contour for which 95\% of trials result in an SNR {\em at least} as large as the level indicated (so that only 5\% of trials result in a worse performance than indicated).

\subsection{BBCoSaMP$(\bA,\bPsi,\by,K)$ versus BBCoSaMP$(\bA\bPsi,\bI,\by,K)$}
\label{sec:bbcosamp12}

We begin with a comparison between the two versions of block-based CoSaMP discussed in Section~\ref{subsec:rec}---specifically, $\mathrm{BBCoSaMP}(\bA,\bPsi,\by,K)$ and $\mathrm{BBCoSaMP}(\bA\bPsi,\bI,\by,K)$.  Recall that the key difference between these approaches is that the former essentially tries to recover $\bx$ directly, while the latter instead attempts to first recover $\balpha$ and then estimates $\widehat{\bx}$ as $\bPsi \widehat{\balpha}$.  Since less is known about the theoretical properties of $\mathrm{BBCoSaMP}(\bA,\bPsi,\by,K)$ (both in and of itself and with respect to the computation of $\cP(\bx,K)$), we have focused more on $\mathrm{BBCoSaMP}(\bA\bPsi,\bI,\by,K)$ in our theoretical development.  However, all of our theorems regarding $\mathrm{BBCoSaMP}(\bA\bPsi,\bI,\by,K)$ (e.g., Theorems~\ref{thm:blockcosampdpss} and~\ref{thm:blockcosampdpssfinal}) are limited to the case where the number $k$ of DPSS vectors per band used to construct the multiband modulated DPSS dictionary $\bPsi$ satisfies $k < 2N\slepw$, and as we will see shortly, in practice we obtain significant improvements in performance by considering $k > 2N\slepw$.

In our experiments we evaluate both versions of block-based CoSaMP as a function of $k$.  For the purposes of this experiment, we assume that there are $K = 5$ active bands, we set the number of measurements $M = 512$, and we construct the measurement matrix $\bA$ to be $M \times N$ with i.i.d.\ Gaussian entries with mean zero and variance $\frac{1}{M}$. In this experiment, we take $\by = \bA \bx$ and do not add noise to the measurements. In Figure~\ref{fig:exp1}(a) we see that when $k$ is small the two approaches perform similarly, but when $k$ exceeds $2N\slepw$ by more than a small amount, the performance of $\mathrm{BBCoSaMP}(\bA,\bPsi,\by,K)$ is far superior to that of $\mathrm{BBCoSaMP}(\bA\bPsi,\bI,\by,K)$.  In fact, when $k$ becomes large we observe that for at least 5\% of trials $\mathrm{BBCoSaMP}(\bA\bPsi,\bI,\by,K)$ results in an SNR of almost 0~dB (which can be achieved simply via the trivial estimate $\widehat{\bx} = \bzero$.)  This gap is further illustrated in Figure~\ref{fig:exp1}(b), which shows the probability that the performance of $\mathrm{BBCoSaMP}(\bA\bPsi,\bI,\by,K)$ is within 3~dB of $\mathrm{BBCoSaMP}(\bA,\bPsi,\by,K)$ as a function of $k$.  This also begins to rapidly decay when $k$ exceeds $2N\slepw$.  For a point of reference, in Figure~\ref{fig:exp1}(c) we plot the eigenvalue $\lambda_{N,\slepw}^{(k)}$ corresponding to the first DPSS basis vector which is not used in constructing our dictionary $\bPsi$.  This value is the dominant term in the approximation error that we can expect when using $\bPsi$ to represent multiband signals.  We see that $\mathrm{BBCoSaMP}(\bA,\bPsi,\by,K)$ continues to improve with increasing $k$, roughly until $\lambda_{N,\slepw}^{(k)}$ approaches the level of machine precision.

\begin{figure*}
   \centering
   \begin{tabular*}{\linewidth}{@{\extracolsep{\fill}} ccc}
   \hspace*{-.3in} \includegraphics[width=2.35in]{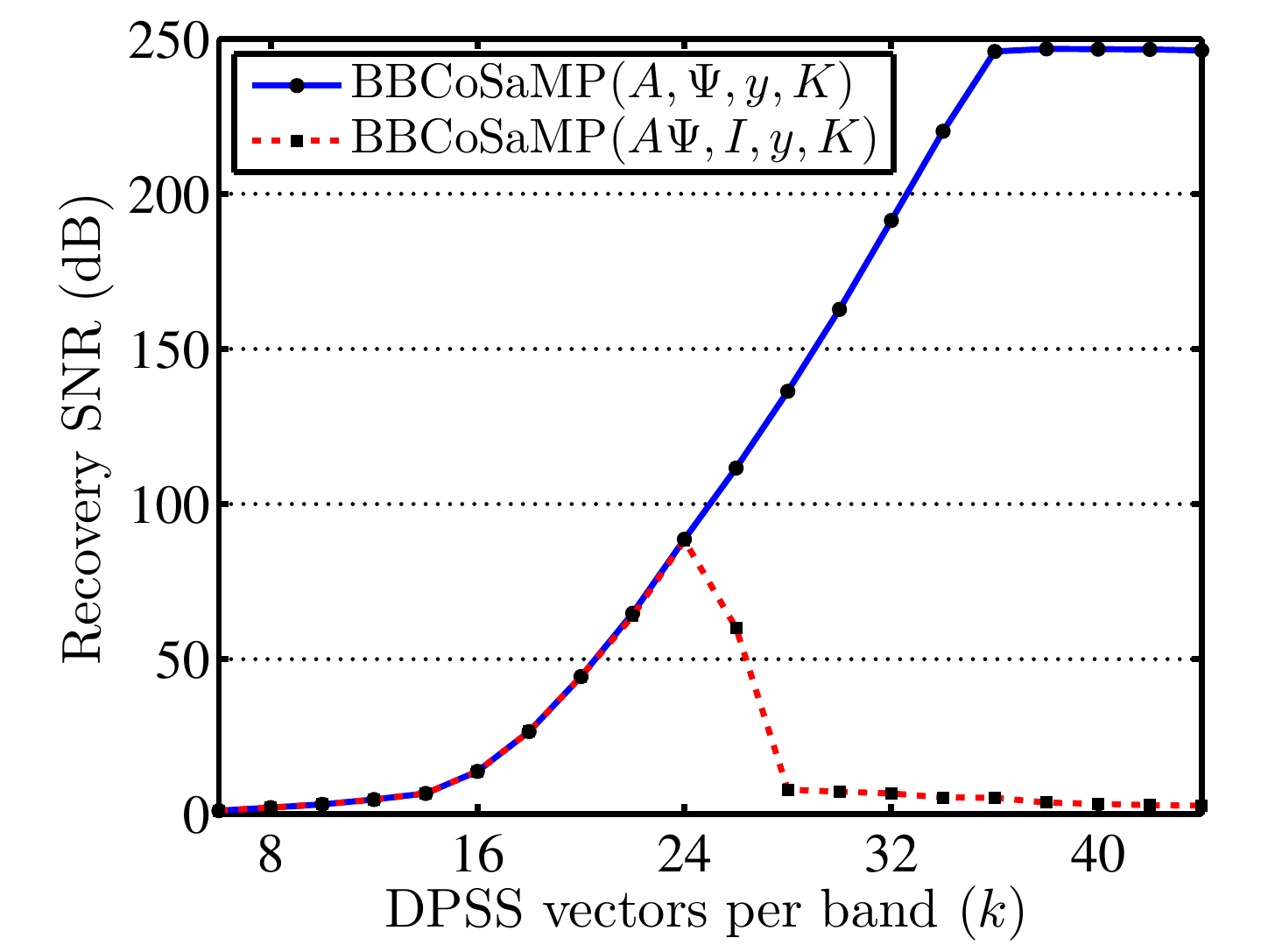} & \hspace*{-.35in} \includegraphics[width=2.35in]{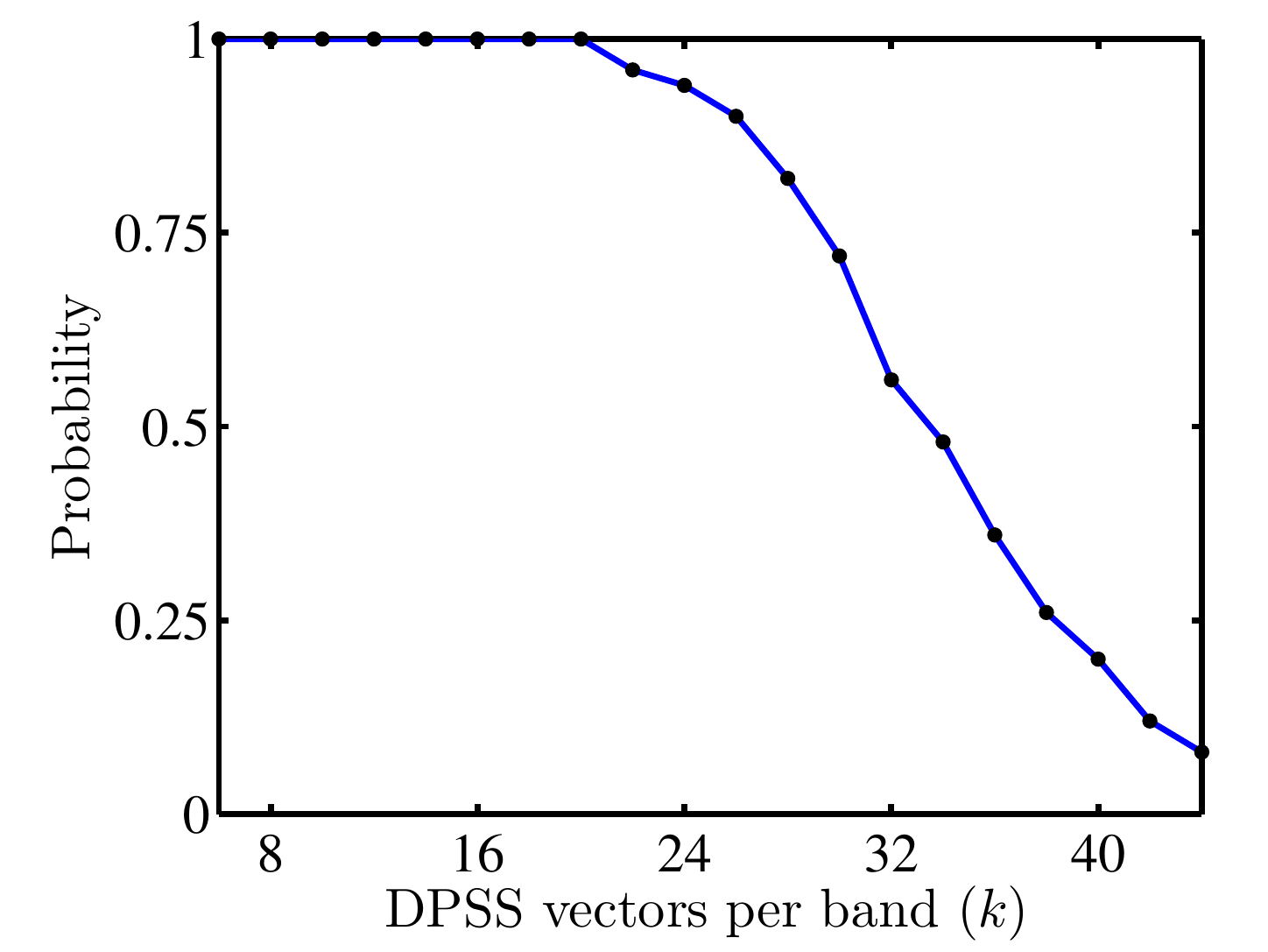} & \hspace*{-.3in} \includegraphics[width=2.35in]{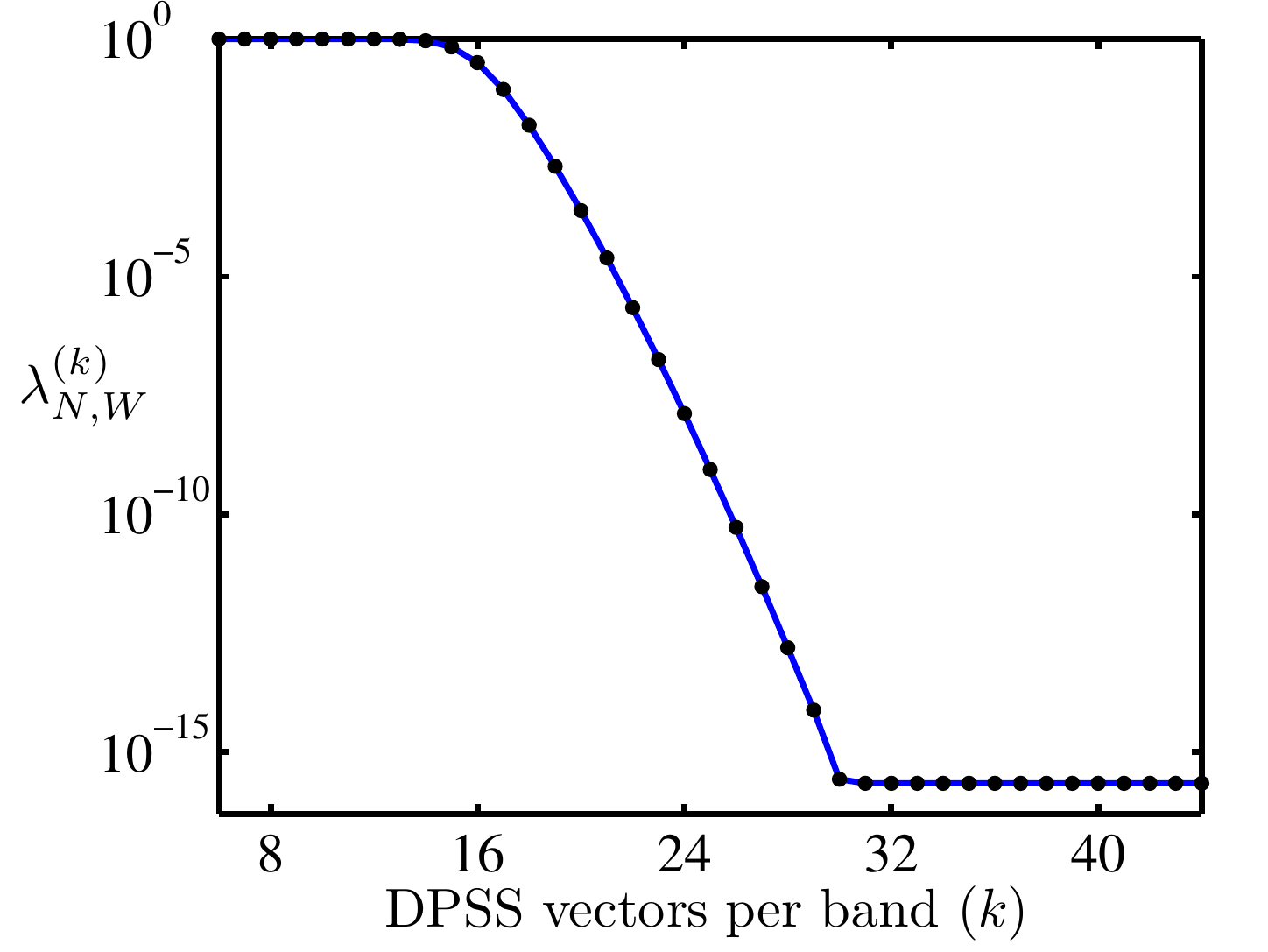} \\
   \hspace*{.01in} {\small (a)} & {\small (b)} & \hspace*{.01in} {\small (c)}
   \end{tabular*}
   \caption{\small \sl Comparison of $\mathrm{BBCoSaMP}(\bA,\bPsi,\by,K)$ to $\mathrm{BBCoSaMP}(\bA\bPsi,\bI,\by,K)$. (a)~The recovery SNR of $\mathrm{BBCoSaMP}(\bA,\bPsi,\by,K)$ and $\mathrm{BBCoSaMP}(\bA\bPsi,\bI,\by,K)$ as a function of the number $k$ of DPSS vectors used per band in constructing the multiband modulated DPSS dictionary $\bPsi$.  (b)~The probability that the performance of $\mathrm{BBCoSaMP}(\bA\bPsi,\bI,\by,K)$ is within 3~dB of $\mathrm{BBCoSaMP}(\bA,\bPsi,\by,K)$ as a function of $k$. (c)~The eigenvalue corresponding to the first DPSS basis vector that is not used.
   \label{fig:exp1}}
\end{figure*}

Our experimental results suggest that in our analysis of $\mathrm{BBCoSaMP}(\bA\bPsi,\bI,\by,K)$ we are correct in requiring that $k$ be relatively small, as the algorithm does indeed break down when $k$ becomes too large. However, before this breakdown the performance can be quite favorable, with recovery SNR exceeding 88~dB, and Theorem~\ref{thm:blockcosampdpssfinal} does guarantee even better performance were $N$ to increase.

We speculate that the breakdown in the performance of $\mathrm{BBCoSaMP}(\bA\bPsi,\bI,\by,K)$ as $k$ grows is likely due to the fact that, for large enough $k$, the dictionary $\bPsi$ begins to contain highly coherent columns, so that any method that attempts to recover $\balpha$ itself is likely to encounter significant problems.  However, the strong performance of $\mathrm{BBCoSaMP}(\bA,\bPsi,\by,K)$ (with no clear limitation on the size of $k$) seems to suggest that this latter approach likely satisfies the kinds of guarantees provided for block-based IHT in Theorems~\ref{thm:blockihtdpss} and~\ref{thm:blockihtdpssfinal}, which provide for arbitrarily large $k$ and hence arbitrarily accurate recovery.  In light of these results, for the remainder of our experiments we focus exclusively on $\mathrm{BBCoSaMP}(\bA,\bPsi,\by,K)$.

\subsection{Impact of measurement noise}

In most realistic scenarios, the compressive measurements $\by$ will be subject to various sources of noise (including noise in the signal itself, noise within the sensing hardware, quantization, etc.)  As noted in Section~\ref{sec:cguar}, our approach to signal recovery inherits the same robustness to noise that is exhibited by traditional CoSaMP.  To illustrate this, we consider the case where the measurements $\by$ are corrupted by additive white Gaussian noise,\footnote{Note that the case where the signal is corrupted with white noise as opposed to the measurements can be reduced to the case where the signal is noise-free and the measurements are noisy, and thus we restrict our attention to noisy measurements.  See~\cite{CandeD_how,DavenLTB_pros} for further discussion of this equivalence and the challenges posed by signal noise.} i.e., $\by = \bA \bx + \be$ where each $\be \sim \cN(\bzero,\sigma^2 \bI)$. In our experiments we consider three different noise levels, quantified by the ``measurement SNR'' (MSNR), which is defined as
$$
\mathrm{MSNR} = 20\log_{10} \left( \frac{\norm{\bA \bx}}{\norm{\be}} \right) \mathrm{dB},
$$
In general, the theoretical guarantees for this scenario suggest that the SNR of the recovery should be roughly comparable to the MSNR.

\begin{figure}
   \centering
   \hspace*{-.45in} \includegraphics[width=\imgwidth]{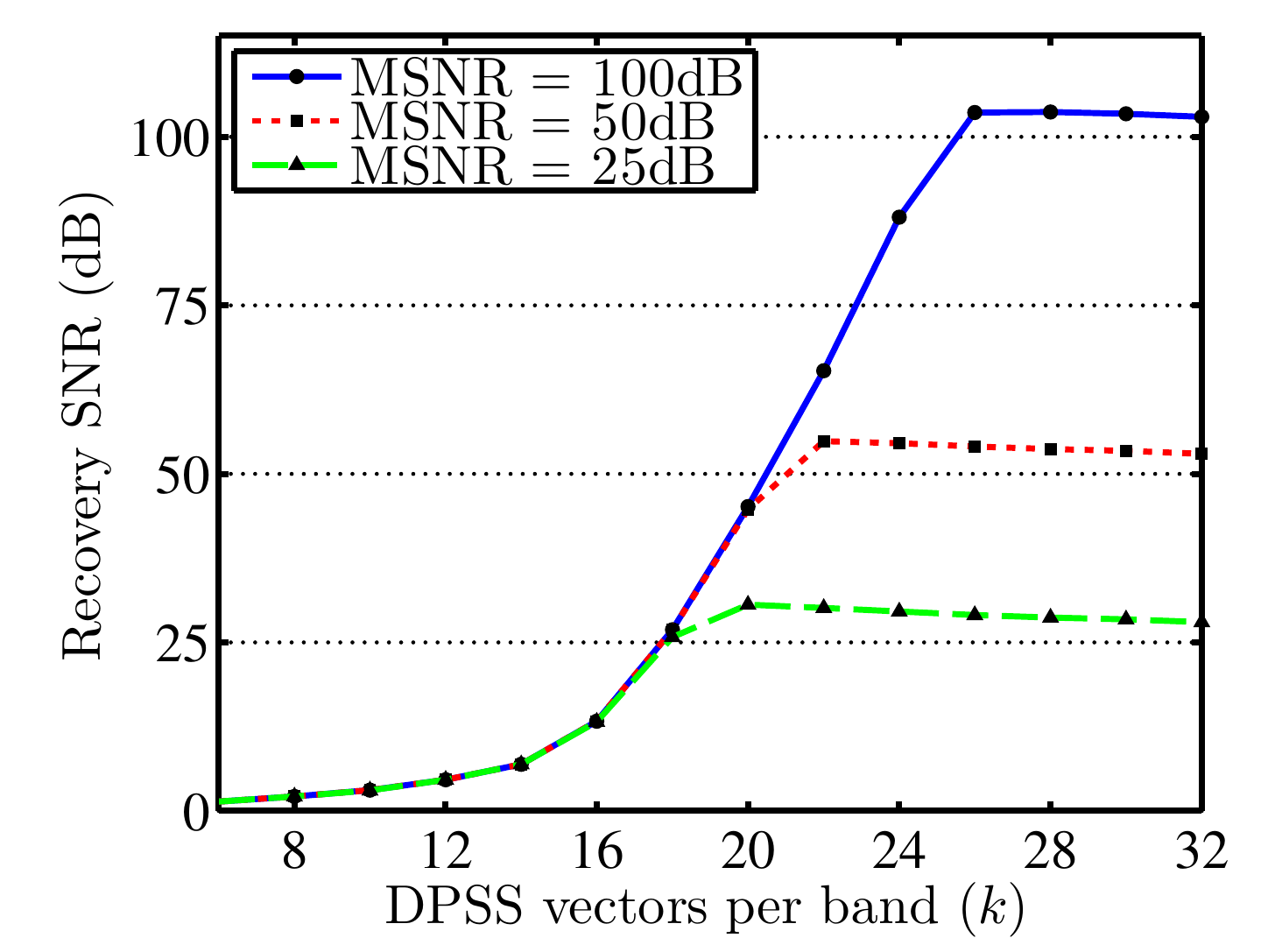}
   \caption{\small \sl Impact of measurement noise on the signal recovery SNR as a function of $k$.
   \label{fig:exp2}}
\end{figure}

The results of our experiment are illustrated in Figure~\ref{fig:exp2}.  In this case we are still assuming that there are $K = 5$ active bands, we use an i.i.d.\ Gaussian matrix $\bA$ with $M = 512$ rows, and we plot the performance as a function of $k$.  We observe that the results are essentially the same as in Figure~\ref{fig:exp1}(a), but that as we increase $k$ the recovery SNR will hit a plateau dictated by the best possible SNR achievable for a given MSNR.  Roughly speaking, when $k$ is small, performance is being limited by ``modeling error'', but as $k$ increases we eventually reach a regime where measurement noise surpasses modeling error as the limiting factor, and no further gains are possible by increasing $k$.  Thus, in practice it will typically be the noise level that dictates the optimal choice in $k$.  For the remainder of our experiments we wish to avoid any assumptions about the noise level, and so we restrict our attention to noise-free measurements.  However, the results all translate to the noisy setting roughly as one would expect based on these results.

\subsection{Required measurement rate}

We now study the performance of our approach as a function of the number of measurements $M$.  Specifically, we consider the cases where $K=5$, $10$, and $15$ bands are active, and for each value of $K$ we let $M$ vary from $2N\slepw K$ (the Landau rate) up to $14N\slepw K$ (oversampling the Landau rate by a factor of 7).  The results are shown in Figures~\ref{fig:exp4}(a) and (b), which plot the results in terms of $\frac{M}{2N\slepw K}$ and $M$, respectively.  We observe that when the measurement rate is only 3 or 4 times that of the Landau rate, we are already doing extremely well, and we obtain near-perfect recovery at 6 times the Landau rate.  We observe very similar behavior for all values of $K$.  Thus, for the sake of simplicity, we restrict the remainder of our experiments to the case where there are only $K=5$ active bands.

\begin{figure*}
   \centering
   \begin{tabular}{cc}
   \hspace*{-.45in} \includegraphics[width=\imgwidth]{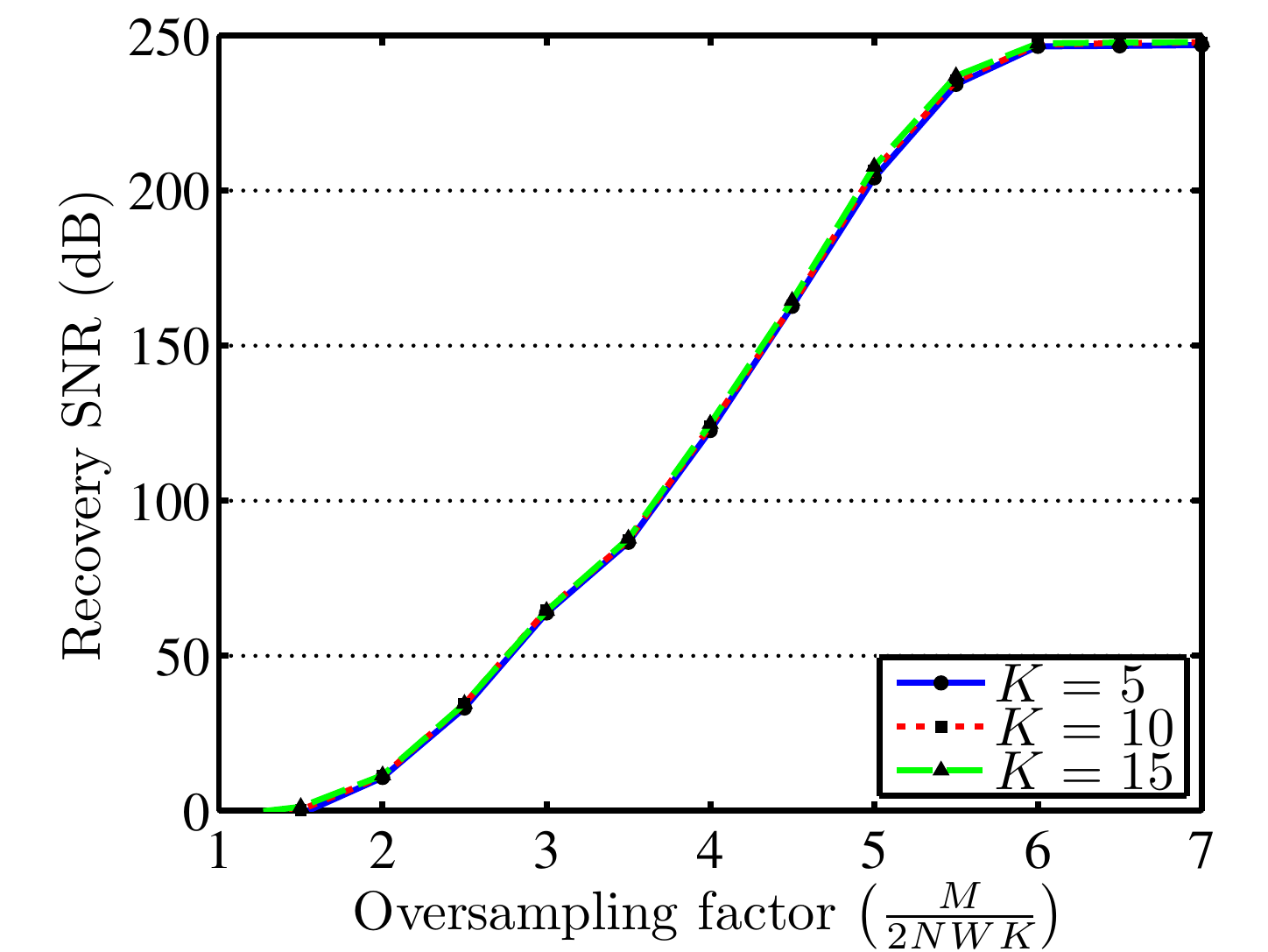} & \includegraphics[width=\imgwidth]{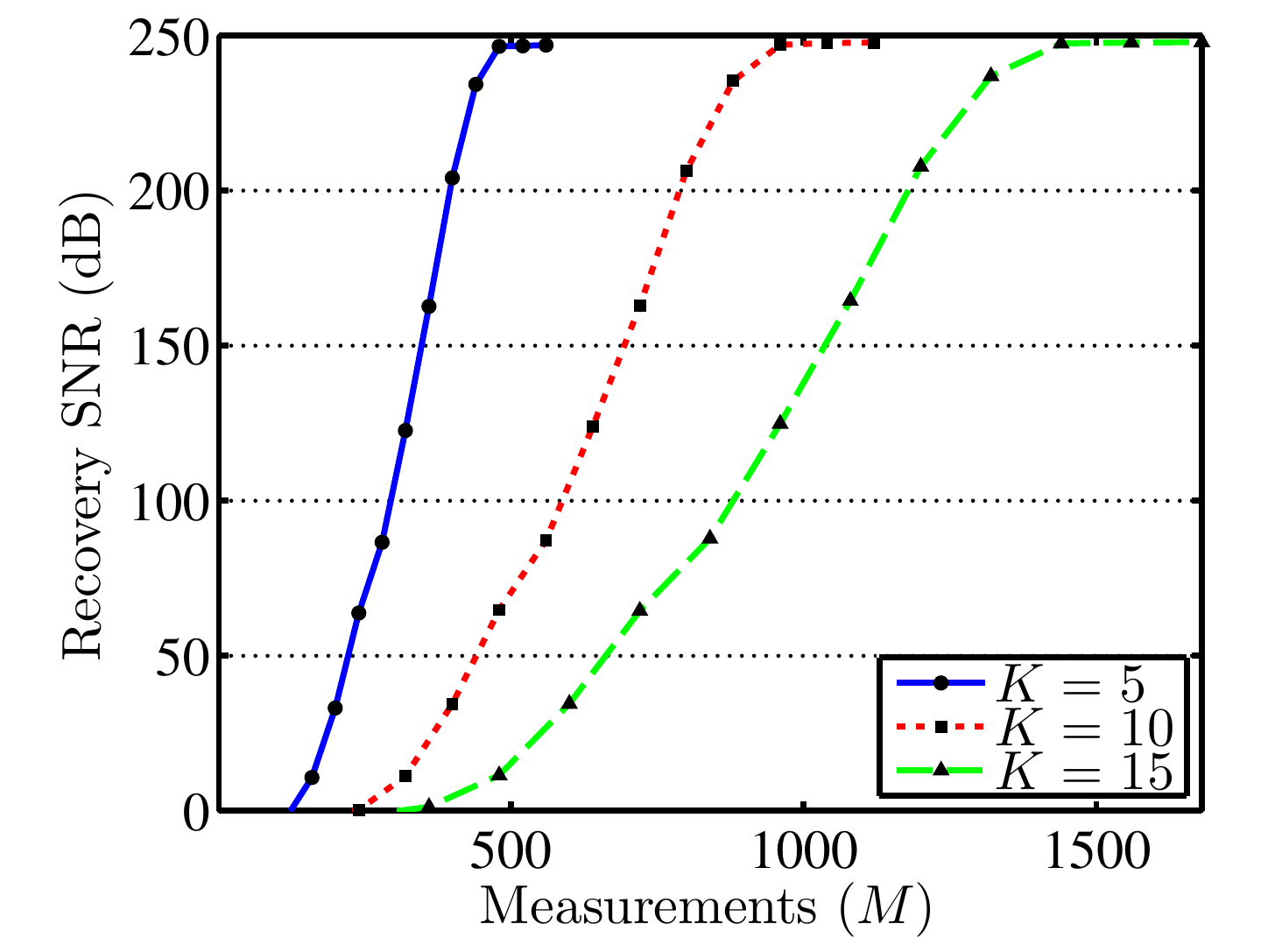}  \\
   {\small (a)} & \hspace*{.35in} {\small (b)}
   \end{tabular}
   \caption{\small \sl Recovery performance as a function of the number of measurements $M$ for $K=5$, $10$, and $15$ active bands.  (a)~Performance in terms of $\frac{M}{2N\slepw K}$, which measures oversampling with respect to the Landau rate.  (b)~Performance in terms of the actual number of measurements $M$.
   \label{fig:exp4}}
\end{figure*}

Before moving on, we note that an important caveat to these results is that as we vary $M$ and $K$, it is sometimes necessary to adjust the value of $k$ used in the construction of $\bPsi$.  This can easily be seen by considering the regime where the measurement rate is very close to the Landau rate.  Suppose $M = 2N\slepw K$ and that we knew {\em a priori} which bands were occupied.  In this case we could perform recovery using the appropriate submatrix of $\bPsi$, which would have $k K$ columns.  Thus, if $k > 2N\slepw$, then we would have only $M = 2N\slepw K$ measurements and would need to estimate $kK > M$ values---a situation we would clearly like to avoid.  Moreover, if we must also estimate the support from the data as well, then it is clearly best to avoid setting $k$ to be too large when $M$ is small. Thus, in the experiments for Figure~\ref{fig:exp4} (and for those below) we set $k$ based on a rule of thumb that we determined based on empirical performance.  Specifically, we set $k = 2N\slepw$ when $\frac{M}{2N\slepw K} \le 2$.  For $\frac{M}{2N\slepw K} \in [2,6]$, we linearly increase $k$ from $2N\slepw = 16$ to, in our case $k=38$ (which represents the rough point at which the first omitted eigenvalue $\lambda_{N,\slepw}^{(k)}$ reaches the level of machine precision).  For larger values of $M$ there is no performance gain by considering $k$ larger than this level.

\subsection{Alternative measurement architectures}

As promised in Section~\ref{ssec:matdesign}, we now evaluate the performance of our approach using more practical measurement schemes.
We compare the performance achieved using an i.i.d.\ Gaussian matrix to that achieved using an $M \times N$ random demodulator matrix~\cite{TroppLDRB_Beyond} and to a random sampling approach where $M$ samples are taken uniformly at random on the length-$N$ Nyquist grid.
The results are shown in Figure~\ref{fig:exp5}.  We observe very similar performance among the three approaches, with the Gaussian matrix performing best, then the random demodulator, followed by random sampling.  Note that while there is certainly a gap between these three approaches, it is also most pronounced in a regime which is likely irrelevant in practice, since it is rare to find applications where a recovery SNR of 200~dB or more is feasible.

\begin{figure}
   \centering
   \hspace*{-.4in} \includegraphics[width=\imgwidth]{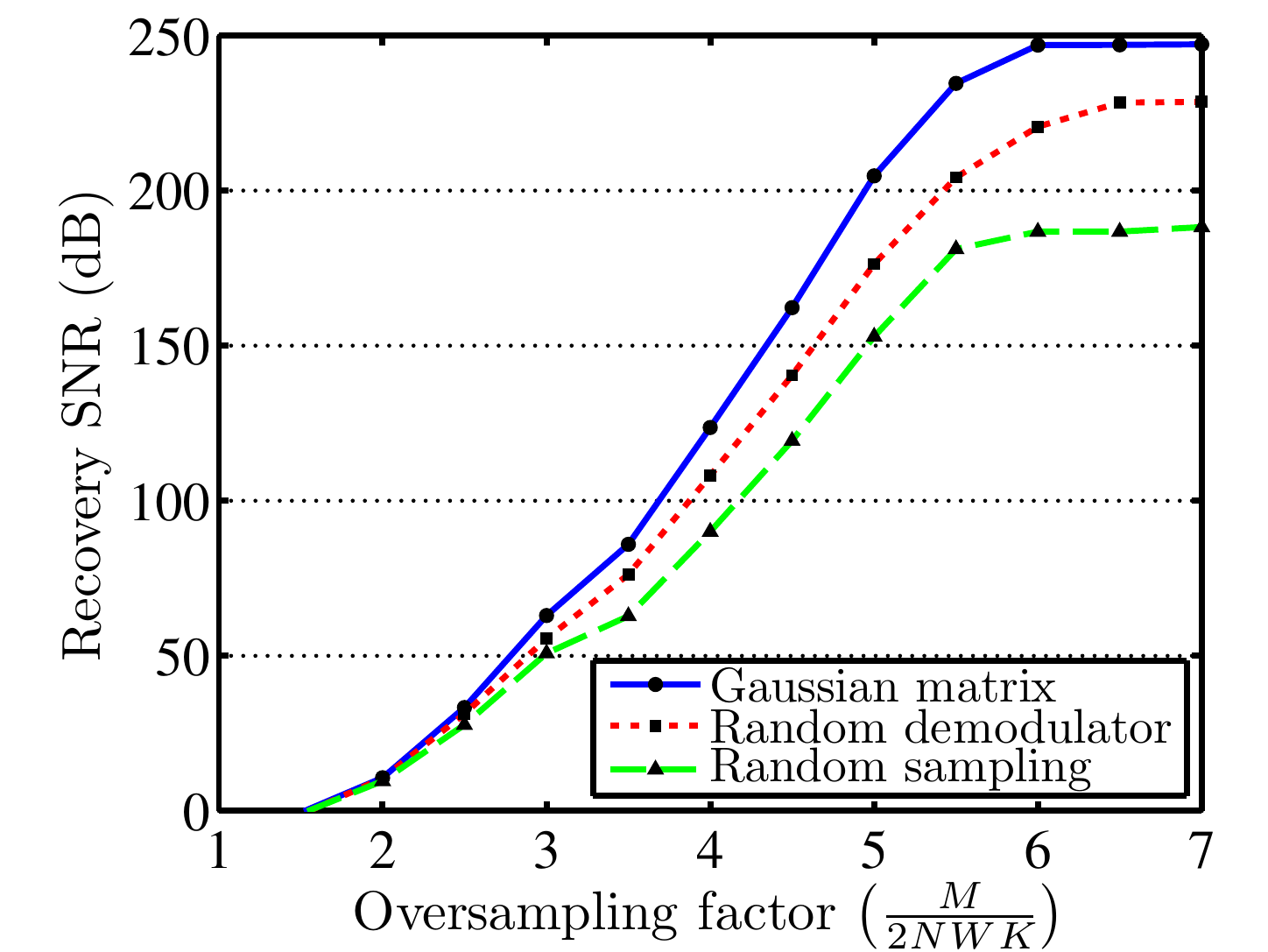}
   \caption{\small \sl Comparison of recovery SNR using i.i.d.\ Gaussian measurements, random demodulator measurements, and random samples.
   \label{fig:exp5}}
\end{figure}

\subsection{Comparison with DFT-based approach}

Finally, we close with a comparison between the performance achievable using our proposed multiband modulated DPSS dictionary $\bPsi$ and the performance achievable using the $N \times N$ DFT as a sparsifying basis.
We test both dictionaries using an $M \times N$ random demodulator measurement matrix~\cite{TroppLDRB_Beyond}, which was originally designed with frequency-sparse signals in mind (but under a multitone signal model, not a multiband model).
Due to DFT leakage effects, we believe that it would be inappropriate to break the DFT dictionary into bands and use a block-based recovery algorithm; thus when $\bPsi$ is the DFT basis, we use $\widehat{\bx} = \mathrm{CoSaMP}(\bA,\bPsi,\by,S)$ as our recovered signal estimate.
In order to give the DFT approach the best possible chance for success, for each value of $M$ we consider a range of possible values for the sparsity parameter $S$, selecting the value of $S$ that achieves the best possible performance according to an oracle.

The performance gap between these two approaches---illustrated in Figure~\ref{fig:exp6}(a)---is monumental.
Using the DFT dictionary, we never achieve a recovery SNR better than 20~dB over this range of measurements, whereas using the multiband modulated DPSS dictionary with block-based CoSaMP, the recovery SNR can exceed 200~dB.

To illustrate the typical difference between these two approaches, in Figure~\ref{fig:exp6}(b) we show in blue (solid line) the DTFT of a representative signal from one trial in this experiment.
Plotted in red (dashed line) is the DTFT of the signal vector recovered using the DPSS dictionary with $\frac{M}{2N\slepw K} = 4$; the recovery SNR is 109~dB and the recovered signal is visually indistinguishable from the original.
Plotted in green (dots) are the DFT coefficients of the signal vector recovered using the DFT dictionary with the same measurements; the recovery SNR is now only 13.4~dB.
While the DFT-based estimate does successfully capture the main peaks of each band, it naturally misses all of the sidelobes of each band (and despite the multiband nature of $x(t)$, these sidelobes are important for accurately representing the window $\bx$).
Moreover, due to modeling error, the DFT-based approach also results in a number of spurious artifacts in regions where there is no significant frequency content in the original signal.

We note that for the DFT-based approach in this experiment, the best estimate produced using $\mathrm{CoSaMP}(\bA,\bPsi,\by,S)$ was observed when setting $S = 85$.
On the other hand, the DPSS approach (according to our rule of thumb) selects $k = 27$, so that in this approach there are a total of $kK = 135$ free parameters.
It is not surprising that the DPSS approach results in superior performance since the model has so many extra dimensions.
What {\em is} potentially surprising is that if we set $S = 135$, the DFT approach exhibits a complete breakdown and is unable to make any practical use of these extra degrees of freedom.

\begin{figure}
   \centering
   \begin{tabular}{cc}
    \hspace*{-.4in} \includegraphics[width=\imgwidth]{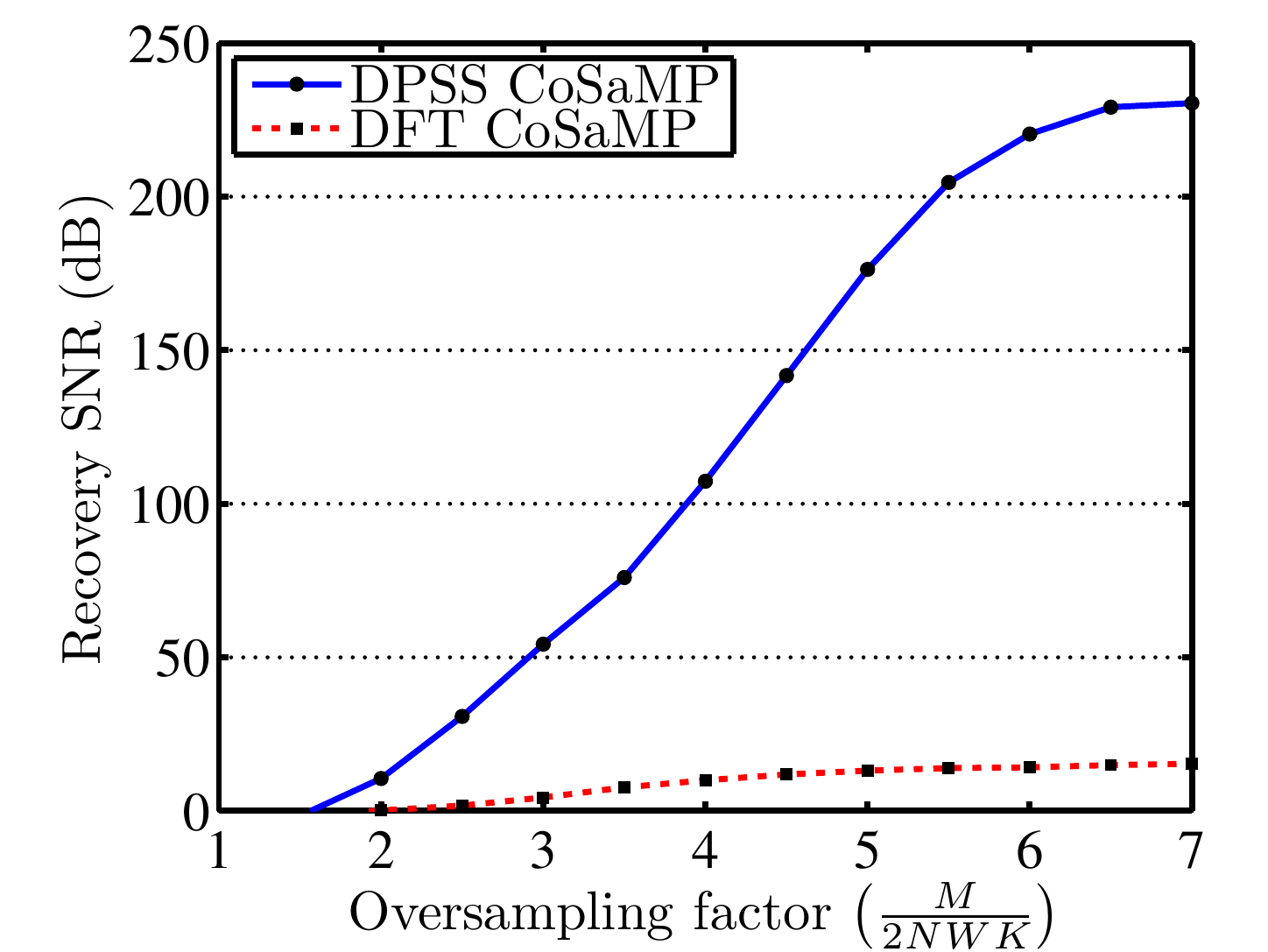} & \includegraphics[width=\imgwidth]{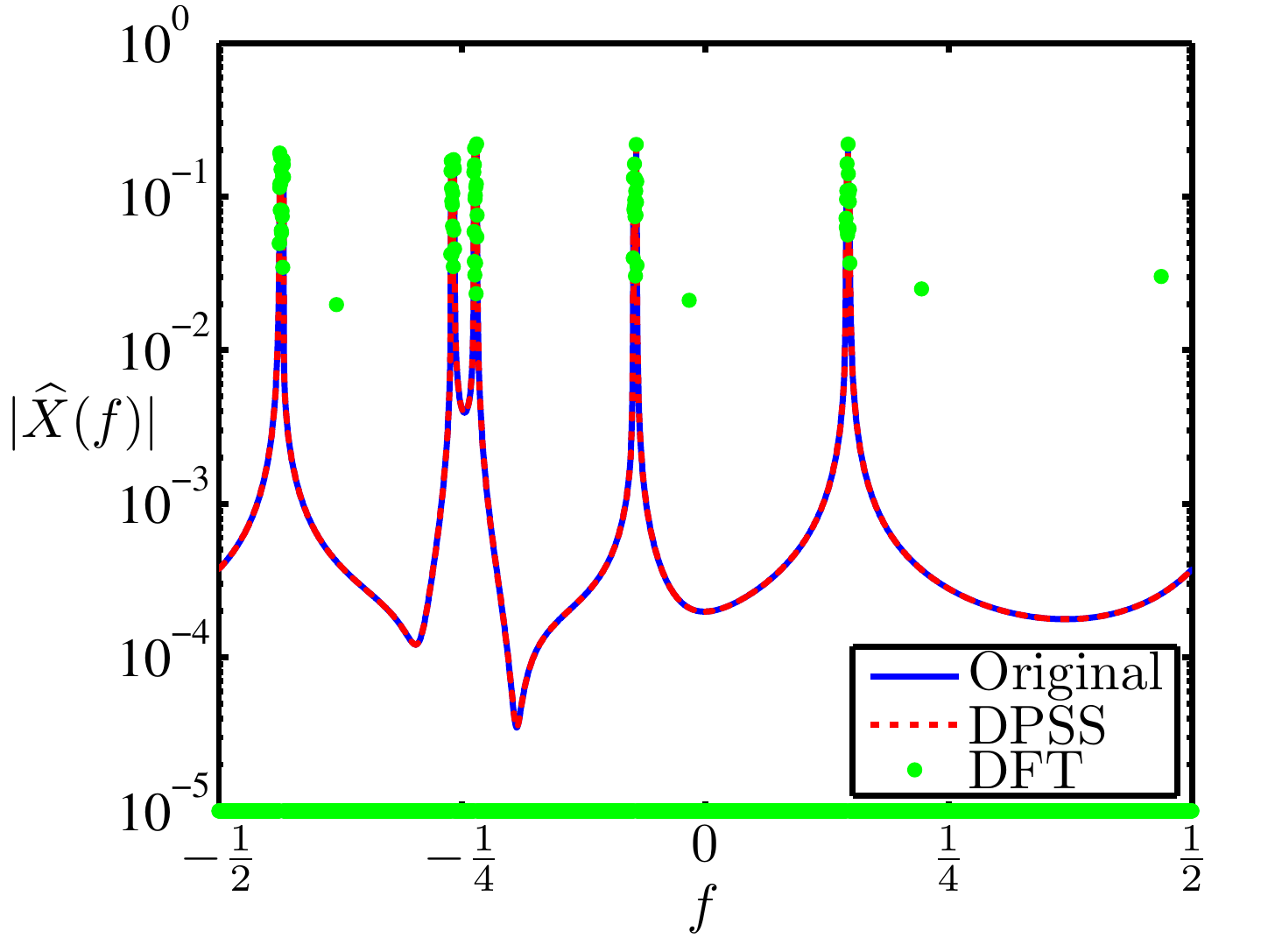}  \\
   \hspace*{-.05in} {\small (a)} & \hspace*{.3in} {\small (b)}
   \end{tabular}
   \caption{\small \sl  Comparison of recovery performance between block-based CoSaMP with the multiband modulated DPSS dictionary and standard CoSaMP using the DFT basis. (a) Performance comparison as a function of $\frac{M}{2N\slepw K}$.  (b) The DTFT of a representative signal, and the recovered estimate using each approach.
   \label{fig:exp6}}
\end{figure} 

\section{Conclusions}
\label{sec:conc}

There are likely to be many ways that one could bridge the gap between the discrete, finite framework of CS and the problem of acquiring a continuous-time signal.
In this paper, using a dictionary constructed from a sequence of modulated DPSS basis elements, we have argued that when dealing with finite-length windows of Nyquist-rate samples, it is possible to map the multiband analog signal model---in a very natural way---into a finite-dimensional sparse model.
This allows us to apply many of the standard theoretical and algorithmic tools from CS to the problem of recovering such a sample vector from compressive measurements.
Moreover, the sparse signals that we encounter in this model actually have a structured sparsity pattern (namely, block-sparsity), which we can exploit through the use of specialized ``model-based'' CS recovery algorithms.
Although our recovery bounds are qualified---we have showed that {\em most} finite-length sample vectors arising from multiband analog signals can be {\em highly accurately} recovered from a number of measurements that is proportional to the Landau rate---these qualifiers are ultimately necessary.
Moreover, we have demonstrated through a series of careful experiments that the multiband modulated DPSS dictionary can provide a far superior alternative to the DFT for sparse recovery.
Our experiments also confirm that certain practical measurement architectures such as the random demodulator, which was previously viewed only as a mechanism for capturing multitone signals~\cite{lexa2011reconciling}, can indeed be used to efficiently capture multiband signals.

As the reader will note from our discussions regarding the selection of the number of columns per band $k$, the computation of $\cP(\bx,K)$, etc., there are certain subtle challenges in choosing the proper dictionary design and implementing an effective recovery algorithm.
However, our experimental results do confirm that it is possible to navigate these waters.
We have also aimed to give some practical insight into the proper choice of design parameters, and we have made all of the software for our simulations available for download at \url{http://www.mines.edu/~mwakin/software/}.
Ultimately, in addition to the open algorithmic questions discussed in Section~\ref{sec:sims}, there are many open questions concerning the most effective way to construct a multiband DPSS dictionary; one could imagine possible advantages to considering multiscale dictionaries (see also~\cite{sejdic2008channel}), adaptive band sizes, overlapping bands, etc.
We leave the consideration of these questions to future work.

It is worth emphasizing that the remarkable efficiency of our DPSS dictionary for sparsely approximating finite-length multiband sample vectors is owed entirely to the eigenvalue concentration behavior described in Section~\ref{sec:eigconc}.
It would be interesting to explore other possible problem domains (beyond multiband signal models) where similar eigenvalue concentration behavior holds.
Again, we leave the consideration of such questions to future work.

We conclude by noting that applications of sparse representations abound, and so there are many possible settings outside of CS where the multiband modulated DPSS dictionary could be useful for processing finite-length sample vectors arising from multiband signals.
For example, there are several possible applications in {\em compressive signal processing}~\cite{DavenBWB_Signal}, where one can attempt to answer certain questions about a signal vector directly from compressive measurements of that vector without having to actually perform signal recovery.
One specific problem in this area could involve cancellation of an interfering signal that has corrupted the measurements.
Suppose $s(t)$ is an analog signal of interest (not necessarily obeying a multiband model), and $i(t)$ is a multiband (or even just single-band) interferer supported on bands indexed by the set $\cI$.
Let $\bs$ and $\bi$ denote the sample vectors arising from these two signals, and suppose we collect compressive measurements $\by = \bA (\bs + \bi)$. %
Then because $\bi$ is concentrated in $\cR(\bPsi_{\cI})$, it follows that $\bA \bi$ is concentrated in $\cR(\widetilde{\bA}_{\cI})$, where $\widetilde{\bA} = \bA \bPsi$.
One can therefore define an orthogonal projection matrix $\bP := \bI - \widetilde{\bA}_{\cI} \widetilde{\bA}_{\cI}^\dag$.
By applying $\bP$ to the measurement vector $\by$, one can remove virtually all of the influence of the interferer, since $\bP \bA \bi \approx \bzero$ even for very strong interferers, and therefore $\bP \by = \bP \bA \bs  + \bP \bA \bi \approx \bP \bA \bs$.
From these processed measurements $\bP \by$ one can attempt to recover $\bs$ or answer various compressive-domain questions that do not require signal recovery.
It can even be possible to derive RIP bounds for $\bP \bA$ and thus provide guarantees on the performance of these techniques.
We refer the reader to~\cite{DavenBB_Compressive,DavenSSBWB_wideband} for additional discussion of the interference cancellation problem and an example using a preliminary version of the modulated DPSS dictionary.

\section*{Acknowledgements}

We are very grateful to Richard Baraniuk, Emmanuel Cand\`{e}s, and Justin Romberg for helpful discussions during the formative stages of this work and to Marco Duarte, Armin Eftekhari, Mark Embree, Chinmay Hegde, Deanna Needell, Yaniv Plan, and Borhan Sanandaji for their input on certain technical points. A preliminary version of this work was first presented at the 2011 Workshop on Signal Processing with Adaptive Sparse Structured Representations (SPARS'11)~\cite{DavenW_Reconstruction}.

\appendix

\section{Review of the Karhunen-Loeve (KL) Transform}
\label{sec:klreview}

We briefly review the basic ideas behind the KL transform~\cite{stark1986probability}.
Suppose that $\br \in \complex^N$ is a random vector with mean zero and a known autocorrelation matrix denoted by $\bR$ with entries defined as
$$
\bR[m,n] = \expval{\overline{\br[n]} \br[m]},
$$
where the overline denotes complex conjugation. Next suppose that we would like to find, for some fixed $k \in \{1,2,\dots,N\}$, a $k$-dimensional subspace $\bQ$ of $\complex^N$ that best captures the energy of $\br$ in expectation. That is, we wish to find the $\bQ$ that minimizes
$$
\expval{\| \br - \bP_{\bQ} \br\|_2^2}
$$
over all $k$-dimensional subspaces of $\complex^N$. The optimal solution to this minimization problem can be found by computing an eigendecomposition of $\bR$. In particular, let
$$
\bR = \bU \bLambda_R \bU^H
$$
denote the eigendecomposition of $\bR$, with the eigenvalues $\lambda_0(\bR) \ge \lambda_1(\bR) \ge \cdots \ge \lambda_{N-1}(\bR) \ge 0$ along the main diagonal of $\bLambda_R$. Let $\bU_k$ denote the $N \times k$ matrix that results from taking the first $k$ columns of $\bU$. The columns of $\bU_k$ provide an orthonormal basis for the optimal $\bQ$, and thus we will have $\bP_{\bQ} = \bU_k \bU^H_k$. For this optimal choice of $\bQ$, we will also have
$$
\expval{\| \br - \bP_{\bQ} \br\|_2^2} = \sum_{\ell = k}^{N-1} \lambda_\ell(\bR).
$$

\section{Proof of Theorem~\ref{thm:sampsinusoid}}
\label{pf:sampsinusoid}
\begin{proof}
Let $f'$ denote a random variable with uniform distribution on $[f_c-\slepw,f_c+\slepw]$, let $\theta'$ denote a random variable with uniform distribution on $[0,2\pi)$, and let $\br = \br(f',\theta') := e^{j \theta'} \sampe{f'}$ denote a random vector in $\complex^N$ defined in terms of $f'$ and $\theta'$. Then we can write
\begin{align*}
\frac{1}{2\slepw} \cdot \int_{f_c-\slepw}^{f_c+\slepw} \| \sampe{f} - \bP_\bQ \sampe{f} \|_2^2 \; df
&=  \int_{f_c-\slepw}^{f_c+\slepw} \| \sampe{f'} - \bP_\bQ \sampe{f'} \|_2^2 p_{f'}(f') df' \\
&=  \int_0^{2\pi} \int_{f_c-\slepw}^{f_c+\slepw} \| e^{j \theta'}( \sampe{f'} - \bP_\bQ \sampe{f'} ) \|_2^2 p_{f'}(f') \frac{1}{2\pi} df' \; d\theta' \\
&=  \int_0^{2\pi} \int_{f_c-\slepw}^{f_c+\slepw} \| e^{j \theta'} \sampe{f'} - \bP_\bQ (e^{j \theta'} \sampe{f'} ) \|_2^2 p_{f'}(f') p_{\theta'}(\theta') df' \; d\theta' \\
&=  \expval{\| \br - \bP_\bQ \br \|_2^2}.
\end{align*}
We next verify that the random vector $\br$ has mean zero, i.e., for each $n \in \{ 0,1,\dots,N-1\}$ we have
\begin{align*}
\expval{\br[n]} &= \int_0^{2\pi} \int_{f_c-\slepw}^{f_c+\slepw} e^{j \theta'} \sampe{f'}[n] p_{f'}(f') p_{\theta'}(\theta') \; df'\; d\theta'  \\
&= \int_0^{2\pi} \int_{f_c-\slepw}^{f_c+\slepw} e^{j \theta'} e^{j 2\pi f' n} \frac{1}{4\slepw\pi} \; df'\; d\theta'\\
&= \frac{1}{4\slepw\pi} \cdot \int_{f_c-\slepw}^{f_c+\slepw} e^{j 2\pi f' n} \int_0^{2\pi} e^{j \theta'}  \; d\theta' \; df' \\
&= 0.
\end{align*}
Next, observe that the random vector $\br$ has autocorrelation matrix $\bR$ with entries:
\begin{align*}
\bR[m,n] &= \expval{ \overline{\br[n]} \br[m]} \\
&= \expval{ \overline{\sampe{f'}[n]} e^{-j \theta'} e^{j \theta'} \sampe{f'}[m]} \\
&= \frac{1}{2\slepw} \cdot \int_{f_c-\slepw}^{f_c+\slepw}  e^{-j 2\pi f' n} e^{j 2\pi f' m} df' \\
&= \frac{1}{2\slepw} \cdot e^{j 2\pi f_c m} \cdot 2\slepw \sinc{ 2\slepw(m-n) } \cdot e^{-j 2\pi f_c n}.
\end{align*}
Therefore, we can write
\begin{equation}
\bR = \frac{1}{2\slepw} \cdot \bE_{f_c} \bB_{N,\slepw} \bE^H_{f_c},
\label{eq:R}
\end{equation}
where $\bE_{f_c}$ is as defined in~\eqref{eq:modulatordef} and $\bB_{N,\slepw}$ is as defined in (\ref{eq:sdef}). Plugging the eigendecomposition for $\bB_{N,\slepw}$ from~\eqref{eq:Seigdecomp} into~\eqref{eq:R}, we obtain
$$
\bR = \frac{1}{2\slepw} \cdot \bE_{f_c} \bS_{N,\slepw} \bLambda_{N,\slepw} \bS^H_{N,\slepw} \bE^H_{f_c}.
$$
From this eigendecomposition and standard results concerning the KL transform (see Appendix~\ref{sec:klreview}), it follows that the optimal $k$-dimensional subspace $\bQ$ will be spanned by the first $k$ columns of $\bE_{f_c} \bS_{N,\slepw}$. Furthermore, this choice of $\bQ$ yields
$$
\frac{1}{2\slepw} \cdot \int_{f_c-\slepw}^{f_c+\slepw} \| \sampe{f} - \bP_\bQ \sampe{f} \|_2^2 \; df = \expval{\| \br - \bP_\bQ \br \|_2^2} = \sum_{\ell = k}^{N-1} \lambda_\ell(\bR) = \frac{1}{2\slepw} \sum_{\ell = k}^{N-1} \lambda_{N,\slepw}^{(\ell)},
$$
as desired.
\end{proof}

\section{Proof of Theorem~\ref{thm:rproc}}
\label{pf:rproc}
\begin{proof}

The sequence $x[n] := x(n \tsamp), n \in \integers$ is a discrete-time, zero-mean, wide sense stationary random process with power spectrum equal to
\begin{equation*}
P_x(f) = \left\{ \begin{array}{ll} \frac{1}{2\slepw},&  f_c - \slepw \le f \le f_c + \slepw, \\ 0, & \mathrm{otherwise}, \end{array} \right.
\end{equation*}
and the vector $\bx \in \complex^N$ equals the restriction of $x[n]$ to the indices $n = 0,1,\dots,N-1$. The inverse DTFT of the power spectrum $P_x(f)$ gives the autocorrelation function for $x[n]$:
$$
r_x[n] = e^{j 2 \pi n f_c} \cdot \sinc{2\slepw n}.
$$
It follows that the random vector $\bx$ has mean zero and an autocorrelation matrix $\bR$ given by~(\ref{eq:R}). Thus, $\bx$ has exactly the same autocorrelation structure as the random vector $\br$ we considered in Appendix~\ref{pf:sampsinusoid}, and so the same KL transform analysis applies. Finally, we can also compute
$$
\expval{\| \bx \|_2^2} = \sum_{n = 0}^{N-1} \expval{|\bx[n]|^2} = N r_x[0] = N.
$$
\end{proof}

\section{Proof of Theorem~\ref{thm:sampbandpass}}
\label{pf:sampbandpass}
\begin{proof}

First, let us define $\cE_{f_c}$ to be an operator that modulates a discrete-time signal by $f_c$; more specifically,
$$
\cE_{f_c}(s)[n] = e^{j 2\pi f_c n} s[n]
$$
for all $n \in \integers$. Since the modulated DPSS vectors form an orthonormal basis for $\complex^N$, and since the discrete-time signal $x[n]$ is time-limited, we can write $x[n]$ as
$$
x = \sum_{\ell=0}^{N-1} \balpha[\ell] \cE_{f_c} \cT_N(s_{N,\slepw}^{(\ell)}),
$$
where the coefficients $\balpha$ are given by $\balpha[\ell] = \langle x, \cE_{f_c} \cT_N(s_{N,\slepw}^{(\ell)}) \rangle = \langle \bx, \bE_{f_c} \bs_{N,\slepw}^{(\ell)} \rangle$ for $\ell = 0,1,\dots, N-1$ and satisfy
$\|\balpha\|_2^2 = \|x\|_2^2 = \|\bx\|_2^2$. By linearity, we can also write
$$
\cB_{f_c,\slepw} x = \sum_{\ell=0}^{N-1} \balpha[\ell] \cB_{f_c,\slepw} \cE_{f_c} \cT_N(s_{N,\slepw}^{(\ell)}),
$$
where the functions $\cB_{f_c,\slepw} \cE_{f_c} \cT_N(s_{N,\slepw}^{(\ell)})$ are modulated DPSS's and therefore remain orthogonal. Therefore,
$$
(1-\delta) \|x\|_2^2 \le  \| \cB_{f_c,\slepw} x \|^2_2 = \sum_{\ell=0}^{N-1} |\balpha[\ell]|^2 \|\cB_{f_c,\slepw} \cE_{f_c}  \cT_N(s_{N,\slepw}^{(\ell)})\|_2^2 = \sum_{\ell=0}^{N-1} |\balpha[\ell]|^2 \lambda_{N,\slepw}^{(\ell)}.
$$
Now note that if $k = 2N\slepw(1+\epsilon)$, then
\begin{align*}
(1-\delta) \|\bx\|_2^2 &\le \sum_{\ell=0}^{k-1} |\balpha[\ell]|^2 \lambda_{N,\slepw}^{(\ell)} + \sum_{\ell=k}^{N-1} |\balpha[\ell]|^2 \lambda_{N,\slepw}^{(\ell)} \\
 &\le \sum_{\ell=0}^{k-1} |\balpha[\ell]|^2 + \sum_{\ell=k}^{N-1} \|\bx\|_2^2 \lambda_{N,\slepw}^{(\ell)} \\
 &\le \| P_{\bQ} \bx \|_2^2 +  \|\bx\|_2^2 N C_3 e^{- C_4 N} \\
 &= (\|\bx\|_2^2 - \| \bx - P_{\bQ} \bx \|_2^2) + \|\bx\|_2^2 N C_3 e^{- C_4 N}
\end{align*}
for $N \ge N_1$. Rearranging terms in the above inequality yields~(\ref{eq:sampbandpassdet}).
\end{proof}

\section{Proof of Theorem~\ref{thm:rprocmb}}
\label{pf:rprocmb}
\begin{proof}

For each $i \in \cI$, the sequence $x_i[n] := x_i(n \tsamp), n \in \integers$ is a discrete-time, zero-mean, wide sense stationary random process with power spectrum equal to
\begin{equation*}
P_{x_i}(f) = \left\{ \begin{array}{ll} \frac{1}{2\slepw K},&  f \in \left[-\frac{\bnyq\tsamp}{2} + i 2\slepw, -\frac{\bnyq\tsamp}{2}+ (i+1)2\slepw \right], \\ 0, & \mathrm{otherwise}. \end{array} \right.
\end{equation*}
These discrete-time sequences will also be independent and jointly wide sense stationary.
%
%
Defining $\bx_i \in \complex^N$ to be the restriction of $x_i[n]$ to the indices $n = 0,1,\dots,N-1$, we have $\bx = \sum_{i \in \cI} \bx_i$ and $P_{x}(f) = \sum_{i \in \cI} P_{x_i}(f)$. We can compute
$$
\expval{\| \bx \|_2^2} = \sum_{n = 0}^{N-1} \expval{|\bx[n]|^2} = \sum_{n = 0}^{N-1} \sum_{i,\ell \in \cI} \expval{ \bx_i[n] \overline{\bx_\ell[n]}} =
\sum_{n = 0}^{N-1} \sum_{i \in \cI} \expval{ |\bx_i[n]|^2} = N \cdot K \cdot \frac{1}{K} = N.
$$
Also, we have
\begin{align*}
\expval{\| \bx - \bP_{\bPsi_{\cI}} \bx \|_2^2}
&= \expval{\left\| \left(\sum_{i \in \cI} \bx_i\right) - \bP_{\bPsi_{\cI}} \left(\sum_{i \in \cI} \bx_i\right) \right\|_2^2} \\
&= \expval{\left\| \sum_{i \in \cI} (\bx_i - \bP_{\bPsi_{\cI}} \bx_i) \right\|_2^2} \\
&\le \expval{\left(\sum_{i \in \cI}\| \bx_i - \bP_{\bPsi_{\cI}} \bx_i \|_2\right)^2} \\
&\le \expval{\left(\sum_{i \in \cI}\| \bx_i -  \bP_{\bPsi_i} \bx_i \|_2\right)^2} \\
&\le \expval{K \sum_{i \in \cI}\| \bx_i -  \bP_{\bPsi_i} \bx_i \|_2^2} \\
&= K \sum_{i \in \cI} \expval{\| \bx_i -  \bP_{\bPsi_i} \bx_i \|_2^2} \\
&= K \cdot K \cdot \frac{1}{2\slepw K} \sum_{\ell = k}^{N-1} \lambda_{N,\slepw}^{(\ell)},
\end{align*}
where the third line follows from the triangle inequality, the fourth line follows because the columns of each $\bPsi_i$ belong to $\bPsi_{\cI}$ as well, the fifth line follows because $\|\bh\|^2_1 \le K \|\bh\|^2_2$ for any $K$-dimensional vector $\bh$, and the final line follows from Theorem~\ref{thm:rproc}.
\end{proof}

\section{Proof of Lemma~\ref{lem:coherence}}
\label{pf:coherence}

\begin{proof}
Let $\bPsi_{i_1}$ and $\bPsi_{i_2}$ denote the blocks from which $\bq_1$ and $\bq_2$, respectively, are drawn. If $i_1 = i_2$, we have already established at the beginning of Section~\ref{sec:mod} that $\langle \bq_1, \bq_2 \rangle = 0$. Thus, we assume henceforth that $i_1 \neq i_2$.

Let $\cE_{f_i}$ denote an operator that modulates a discrete-time signal by $f_i$; more specifically, $\cE_{f_i}(s)[n] = e^{j 2\pi f_i n} s[n]$ for all $n \in \integers$.
Let $\cB_{f_i,\slepw}$ denote an orthogonal projection operator that takes a discrete-time signal, bandlimits its DTFT to the frequency range $f \in [f_i-\slepw, f_i+\slepw]$, and returns the corresponding signal in the time domain.
Finally, define $\cB^c_{f_{i_1},f_{i_2},\slepw} := I - \cB_{f_{i_1},\slepw} - \cB_{f_{i_2},\slepw}$, which is also an orthogonal projection for finite-energy signals because $|f_{i_1} - f_{i_2}| \ge 2\slepw$.
%
%
For some $\ell ,\ell' \in \{0,1,\dots,2N\slepw(1-\epsilon)-1\}$, we can write
\begin{align*}
\left\langle \bq_1, \bq_2 \right\rangle &= \left\langle \bE_{f_{i_1}} s_{N,\slepw}^{(\ell)}, \bE_{f_{i_2}} s_{N,\slepw}^{(\ell')} \right\rangle \\ &= \left\langle \cE_{f_{i_1}}\cT_N s_{N,\slepw}^{(\ell)}, \cE_{f_{i_2}}\cT_N s_{N,\slepw}^{(\ell')} \right\rangle  \\
&=
 \left\langle \cB_{f_{i_1},\slepw} \cE_{f_{i_1}}\cT_N s_{N,\slepw}^{(\ell)} + \cB_{f_{i_2},\slepw}  \cE_{f_{i_1}}\cT_N s_{N,\slepw}^{(\ell)} +  \cB^c_{f_{i_1},f_{i_2},\slepw} \cE_{f_{i_1}}\cT_N s_{N,\slepw}^{(\ell)}, \right. \\
& ~~~ \left. \cB_{f_{i_1},\slepw}\cE_{f_{i_2}}\cT_N s_{N,\slepw}^{(\ell')}  + \cB_{f_{i_2},\slepw} \cE_{f_{i_2}}\cT_N s_{N,\slepw}^{(\ell')}  + \cB^c_{f_{i_1},f_{i_2},\slepw}\cE_{f_{i_2}}\cT_N s_{N,\slepw}^{(\ell')} \right\rangle \\
&= \left\langle \cB_{f_{i_1},\slepw} \cE_{f_{i_1}}\cT_N s_{N,\slepw}^{(\ell)}, \cB_{f_{i_1},\slepw}\cE_{f_{i_2}}\cT_N s_{N,\slepw}^{(\ell')} \right\rangle + \left\langle \cB_{f_{i_2},\slepw}  \cE_{f_{i_1}}\cT_N s_{N,\slepw}^{(\ell)}, \right. \\
& ~~~ \left. \cB_{f_{i_2},\slepw} \cE_{f_{i_2}}\cT_N s_{N,\slepw}^{(\ell')}  \right\rangle + \left\langle  \cB^c_{f_{i_1},f_{i_2},\slepw} \cE_{f_{i_1}}\cT_N s_{N,\slepw}^{(\ell)}, \cB^c_{f_{i_1},f_{i_2},\slepw}\cE_{f_{i_2}}\cT_N s_{N,\slepw}^{(\ell')} \right\rangle.
\end{align*}
From (\ref{eq:dpssnorm}), (\ref{eq:apbl}), and the fact that $|f_{i_2} - f_{i_1}| \ge 2\slepw$, it follows that
$$
| \left\langle \bq_1, \bq_2 \right\rangle | \le \sqrt{\lambda_{N,\slepw}^{(\ell)}} \cdot \sqrt{1-\lambda_{N,\slepw}^{(\ell')}} + \sqrt{1-\lambda_{N,\slepw}^{(\ell)}} \cdot \sqrt{\lambda_{N,\slepw}^{(\ell')}} + \sqrt{1-\lambda_{N,\slepw}^{(\ell)}} \cdot \sqrt{1-\lambda_{N,\slepw}^{(\ell')}}.
$$
Finally,~\eqref{eq:coherencemain} follows by bounding the $\sqrt{\lambda}$ terms in this expression by $1$ and by bounding the $\sqrt{1-\lambda}$ terms in this expression using (\ref{eq:eig2nbminus}).
\end{proof}

\bibliographystyle{plain}
\footnotesize
\bibliography{preamble,mainbib,supplementalbib}

\begin{thebibliography}{10}

\bibitem{BajwaHRWN_Toeplitz}
W.~Bajwa, J.~Haupt, G.~Raz, S.~Wright, and R.~Nowak.
\newblock Toeplitz-structured compressed sensing matrices.
\newblock In {\em Proc. IEEE Work. Stat. Signal Processing}, Madison, WI, Aug.
  2007.

\bibitem{Baran_Compressive}
R.~Baraniuk.
\newblock Compressive sensing.
\newblock {\em IEEE Signal Processing Mag.}, 24(4):118--120, 124, 2007.

\bibitem{BaranCDH_Model}
R.~Baraniuk, V.~Cevher, M.~Duarte, and C.~Hegde.
\newblock Model-based compressive sensing.
\newblock {\em IEEE Trans. Inform. Theory}, 56(4):1982--2001, 2010.

\bibitem{BaranDDW_Simple}
R.~Baraniuk, M.~Davenport, R.~DeVore, and M.~Wakin.
\newblock A simple proof of the restricted isometry property for random
  matrices.
\newblock {\em Const. Approx.}, 28(3):253--263, 2008.

\bibitem{Blume_Sampling}
T.~Blumensath.
\newblock Sampling and reconstructing signals from a union of linear subspaces.
\newblock {\em IEEE Trans. Inform. Theory}, 57(7):4660--4671, 2011.

\bibitem{BlumeD_Iterative}
T.~Blumensath and M.~Davies.
\newblock Iterative hard thresholding for compressive sensing.
\newblock {\em Appl. Comput. Harmon. Anal.}, 27(3):265--274, 2009.

\bibitem{BlumeD_Sampling}
T.~Blumensath and M.~Davies.
\newblock Sampling theorems for signals from the union of finite-dimensional
  linear subspaces.
\newblock {\em IEEE Trans. Inform. Theory}, 55(4):1872--1882, 2009.

\bibitem{BreslF_Spectrum}
Y.~Bresler and P.~Feng.
\newblock Spectrum-blind minimum-rate sampling and reconstruction of {2D}
  multiband signals.
\newblock In {\em Proc. IEEE Int. Conf. Image Processing (ICIP)}, Zurich,
  Switzerland, Sept. 1996.

\bibitem{BuldyK_Metric}
V.~Buldygin and Y.~Kozachenko.
\newblock {\em Metric Characterization of Random Variables and Random
  Processes}.
\newblock American Mathematical Society, Providence, RI, 2000.

\bibitem{Cande_Compressive}
E.~Cand\`{e}s.
\newblock Compressive sampling.
\newblock In {\em Proc. Int. Congress of Math.}, Madrid, Spain, Aug. 2006.

\bibitem{CandeD_how}
E.~Cand\`{e}s and M.~Davenport.
\newblock How well can we estimate a sparse vector?
\newblock {\em Arxiv preprint arXiv:1104.5246}, 2011.

\bibitem{CandeENR_Compressed}
E.~Cand\`{e}s, Y.~Eldar, D.~Needell, and P.~Randall.
\newblock Compressed sensing with coherent and redundant dictionaries.
\newblock {\em Appl. Comput. Harmon. Anal.}, 31(1):59--73, 2011.

\bibitem{CandeRT_Robust}
E.~Cand\`{e}s, J.~Romberg, and T.~Tao.
\newblock Robust uncertainty principles: {E}xact signal reconstruction from
  highly incomplete frequency information.
\newblock {\em IEEE Trans. Inform. Theory}, 52(2):489--509, 2006.

\bibitem{CandeT_Decoding}
E.~Cand\`{e}s and T.~Tao.
\newblock Decoding by linear programming.
\newblock {\em IEEE Trans. Inform. Theory}, 51(12):4203--4215, 2005.

\bibitem{cevher2011alps}
V.~Cevher.
\newblock An {ALPS} view of sparse recovery.
\newblock In {\em Proc. IEEE Int. Conf. Acoust., Speech, and Signal Processing
  (ICASSP)}, Prague, Czech Republic, May 2011.

\bibitem{CohenDD_Instance}
A.~Cohen, W.~Dahmen, and R.~DeVore.
\newblock Instance optimal decoding by thresholding in compressed sensing.
\newblock In {\em Proc. Int. Conf. Harmon. Anal. and Partial Diff. Eq.}, El
  Escorial Madrid, Spain, Jun. 2008.

\bibitem{DaiM_Subspace}
W.~Dai and O.~Milenkovic.
\newblock Subspace pursuit for compressive sensing signal reconstruction.
\newblock {\em IEEE Trans. Inform. Theory}, 55(5):2230--2249, 2009.

\bibitem{DaubeD_Approximating}
I.~Daubechies and R.~DeVore.
\newblock Approximating a bandlimited function using very coarsely quantized
  data: {A} family of stable sigma-delta modulators of arbitrary order.
\newblock {\em Ann. Math.}, 158(2):679--710, 2003.

\bibitem{Daven_Concentration}
M.~Davenport.
\newblock Concentration of measure and sub-gaussian distributions, 2009.
\newblock available online at \url{http://cnx.org/content/m32583/latest/}.

\bibitem{DavenBB_Compressive}
M.~Davenport, P.~Boufounos, and R.~Baraniuk.
\newblock Compressive domain interference cancellation.
\newblock In {\em {Proc. Work. Struc. Parc. Rep. Adap. Signaux (SPARS)}},
  Saint-Malo, France, Apr. 2009.

\bibitem{DavenBWB_Signal}
M.~Davenport, P.~Boufounos, M.~Wakin, and R.~Baraniuk.
\newblock Signal processing with compressive measurements.
\newblock {\em IEEE J. Select. Top. Signal Processing}, 4(2):445--460, 2010.

\bibitem{DavenLTB_pros}
M.~Davenport, J.~Laska, J.~Treichler, and R.~Baraniuk.
\newblock The pros and cons of compressive sensing for wideband signal
  acquisition: {N}oise folding vs. dynamic range.
\newblock {\em to appear in {IEEE Trans. Signal Processing}}, 2012.

\bibitem{DavenSSBWB_wideband}
M.~Davenport, S.~Schnelle, J.P. Slavinsky, R.~Baraniuk, M.~Wakin, and
  P.~Boufounos.
\newblock A wideband compressive radio receiver.
\newblock In {\em Proc. Military Communications Conference (MILCOM)}, San Jose,
  CA, Oct. 2010.

\bibitem{DavenW_Analysis}
M.~Davenport and M.~Wakin.
\newblock Analysis of {O}rthogonal {M}atching {P}ursuit using the restricted
  isometry property.
\newblock {\em IEEE Trans. Inform. Theory}, 56(9):4395--4401, 2010.

\bibitem{DavenW_Reconstruction}
M.~Davenport and M.~Wakin.
\newblock Reconstruction and cancellation of sampled multiband signals using
  {Discrete Prolate Spheroidal Sequences}.
\newblock In {\em {Proc. Work. Struc. Parc. Rep. Adap. Signaux (SPARS)}},
  Edinburgh, Scotland, Jun. 2011.

\bibitem{DavisMZ_Adaptive}
G.~Davis, S.~Mallat, and Z.~Zhang.
\newblock Adaptive time-frequency decompositions.
\newblock {\em SPIE J. Opt. Engin.}, 33(7):2183--2191, 1994.

\bibitem{DeVorPW_Instance}
R.~DeVore, G.~Petrova, and P.~Wojtaszczyk.
\newblock Instance-optimality in probability with an $\ell_1$-minimization
  decoder.
\newblock {\em Appl. Comput. Harmon. Anal.}, 27(3):275--288, 2009.

\bibitem{Donoh_Compressed}
D.~Donoho.
\newblock Compressed sensing.
\newblock {\em IEEE Trans. Inform. Theory}, 52(4):1289--1306, 2006.

\bibitem{DonohDTS_Sparse}
D.~Donoho, I.~Drori, Y.~Tsaig, and J.-L. Starck.
\newblock Sparse solution of underdetermined linear equations by stagewise
  orthogonal matching pursuit.
\newblock Stanford Technical Report, 2006.

\bibitem{DuartB_Spectral}
M.~Duarte and R.~Baraniuk.
\newblock Spectral compressive sensing.
\newblock Preprint, 2011.

\bibitem{EldarKB_Block-Sparse}
Y.~Eldar, P.~Kuppinger, and H.~B{\"o}lcskei.
\newblock Block-sparse signals: {U}ncertainty relations and efficient recovery.
\newblock {\em IEEE Trans. Signal Processing}, 58(6):3042--3054, 2010.

\bibitem{EldarM_Robust}
Y.~Eldar and M.~Mishali.
\newblock Robust recovery of signals from a structured union of subspaces.
\newblock {\em IEEE Trans. Inform. Theory}, 55(11):5302--5316, 2009.

\bibitem{fancourt2000relationship}
C.~Fancourt and J.~Principe.
\newblock On the relationship between the {Karhunen-Loeve} transform and the
  prolate spheroidal wave functions.
\newblock In {\em Proc. IEEE Int. Conf. Acoust., Speech, and Signal Processing
  (ICASSP)}, Istanbul, Turkey, Jun. 2000.

\bibitem{Feng_Universal}
P.~Feng.
\newblock {\em Universal spectrum blind minimum rate sampling and
  reconstruction of multiband signals}.
\newblock PhD thesis, University of Illinois at Urbana-Champaign, Mar. 1997.

\bibitem{FengB_Spectrum}
P.~Feng and Y.~Bresler.
\newblock Spectrum-blind minimum-rate sampling and reconstruction of multiband
  signals.
\newblock In {\em Proc. IEEE Int. Conf. Acoust., Speech, and Signal Processing
  (ICASSP)}, Atlanta, GA, May 1996.

\bibitem{Gerv_Uber}
S.~Ger\v{s}gorin.
\newblock {\"U}ber die {A}bgrenzung der {E}igenwerte einer {M}atrix.
\newblock {\em Izv. Akad. Nauk SSSR Ser. Fiz.-Mat.}, 6:749--754, 1931.

\bibitem{gosse2010compressed}
L.~Gosse.
\newblock Compressed sensing with preconditioning for sparse recovery with
  subsampled matrices of {S}lepian prolate functions.
\newblock Preprint, 2010.

\bibitem{Hanse_Regularization}
P.~C. Hansen.
\newblock {Regularization Tools Version 4.0 for Matlab 7.3}.
\newblock {\em Numerical Algorithms}, 46:189--194, 2007.

\bibitem{hein1994theoretical}
S.~Hein and A.~Zakhor.
\newblock Theoretical and numerical aspects of an {SVD}-based method for
  band-limiting finite-extent sequences.
\newblock {\em IEEE Trans. Signal Processing}, 42(5):1227--1230, 1994.

\bibitem{izu2009time}
S.~Izu and J.~Lakey.
\newblock Time-frequency localization and sampling of multiband signals.
\newblock {\em Acta Applicandae Mathematicae}, 107(1):399--435, 2009.

\bibitem{Landa_Necessary}
H.~Landau.
\newblock Necessary density conditions for sampling and interpolation of
  certain entire functions.
\newblock {\em Acta Math.}, 117:37--52, 1967.

\bibitem{LandaP_ProlateII}
H.~Landau and H.~Pollak.
\newblock Prolate spheroidal wave functions, {F}ourier analysis, and
  uncertainty. {II}.
\newblock {\em Bell Systems Tech. J.}, 40(1):65--84, 1961.

\bibitem{LandaP_ProlateIII}
H.~Landau and H.~Pollak.
\newblock Prolate spheroidal wave functions, {F}ourier analysis, and
  uncertainty. {III}.
\newblock {\em Bell Systems Tech. J.}, 41(4):1295--1336, 1962.

\bibitem{lexa2011reconciling}
M.~Lexa, M.~Davies, and J.~Thompson.
\newblock Reconciling compressive sampling systems for spectrally-sparse
  continuous-time signals.
\newblock {\em IEEE Trans. Signal Processing}, 60(1):155--171, 2012.

\bibitem{LuD_Theory}
Y.~Lu and M.~Do.
\newblock A theory for sampling signals from a union of subspaces.
\newblock {\em IEEE Trans. Signal Processing}, 56(6):2334--2345, 2008.

\bibitem{Lyons_Understanding}
R.~Lyons.
\newblock {\em {Understanding Digital Signal Processing}}.
\newblock Pearson Education, Inc., Upper Saddle River, NJ, 2004.

\bibitem{MishaE_Blind}
M.~Mishali and Y.~Eldar.
\newblock Blind multi-band signal reconstruction: {C}ompressed sensing for
  analog signals.
\newblock {\em IEEE Trans. Signal Processing}, 57(3):993--1009, 2009.

\bibitem{MishaE_From}
M.~Mishali and Y.~Eldar.
\newblock From theory to practice: {S}ub-{N}yquist sampling of sparse wideband
  analog signals.
\newblock {\em IEEE J. Select. Top. Signal Processing}, 4(2):375--391, 2010.

\bibitem{mugler99linear}
D.~Mugler, Y.~Wu, and S.~Clary.
\newblock Linear prediction of bandpass signals based on past samples.
\newblock In {\em Proc. Sampling Theory and Applications (SampTA)}, Loen,
  Norway, Aug. 1999.

\bibitem{NeedeT_CoSaMP}
D.~Needell and J.~Tropp.
\newblock {CoSaMP}: {I}terative signal recovery from incomplete and inaccurate
  samples.
\newblock {\em Appl. Comput. Harmon. Anal.}, 26(3):301--321, 2009.

\bibitem{NeedeV_Uniform}
D.~Needell and R.~Vershynin.
\newblock Uniform uncertainty principle and signal recovery via regularized
  orthogonal matching pursuit.
\newblock {\em Found. Comput. Math.}, 9(3):317--334, 2009.

\bibitem{NeedeV_Signal}
D.~Needell and R.~Vershynin.
\newblock Signal recovery from incomplete and inaccurate measurements via
  regularized orthogonal matching pursuit.
\newblock {\em IEEE J. Select. Top. Signal Processing}, 4(2):310--316, 2010.

\bibitem{oh2010signal}
J.~Oh, S.~Senay, and L.~Chaparro.
\newblock Signal reconstruction from nonuniformly spaced samples using
  evolutionary {S}lepian transform-based {POCS}.
\newblock {\em EURASIP J. Adv. Signal Processing}, 2010, February 2010.

\bibitem{PatiRK_Orthogonal}
Y.~Pati, R.~Rezaifar, and P.~Krishnaprasad.
\newblock Orthogonal matching pursuit: {R}ecursive function approximation with
  applications to wavelet decomposition.
\newblock In {\em Proc. Asilomar Conf. Signals, Systems, and Computers},
  Pacific Grove, CA, Nov. 1993.

\bibitem{Phill_technique}
D.~Phillips.
\newblock A technique for the numerical solution of certain integral equations
  of the first kind.
\newblock {\em J. ACM}, 9:84--97, 1962.

\bibitem{Prony_Essai}
R.~Prony.
\newblock Essai exp\'{e}rimental et analytique sur les lois de la
  {D}ilatabilit\'{e} des fluides \'{e}lastiques et sur celles de la {F}orce
  expansive de la vapeur de l'eau et de la vapeur de l'alkool, \`{a}
  diff\'{e}rentes temp\'{e}ratures.
\newblock {\em J. de l' \hspace{-0.75em} \'{E}cole Polytechnique, {\em
  Flor\'{e}al et {P}rairial {III}}}, 1(2):24--76, 1795.
\newblock {R. Prony is Gaspard Riche, baron de Prony}.

\bibitem{Rombe_Compressive}
J.~Romberg.
\newblock Compressive sensing by random convolution.
\newblock {\em SIAM J. Imag. Sci.}, 2(4):1098--1128, 2009.

\bibitem{scharf1991svd}
L.~Scharf.
\newblock The {SVD} and reduced rank signal processing.
\newblock {\em Signal Processing}, 25(2):113--133, 1991.

\bibitem{sejdic2008channel}
E.~Sejdi\'{c}, M.~Luccini, S.~Primak, K.~Baddour, and T.~Willink.
\newblock Channel estimation using {DPSS} based frames.
\newblock In {\em Proc. IEEE Int. Conf. Acoust., Speech, and Signal Processing
  (ICASSP)}, Las Vegas, Nevada, Mar. 2008.

\bibitem{senay2009reconstruction}
S.~Senay, L.~Chaparro, and L.~Durak.
\newblock Reconstruction of nonuniformly sampled time-limited signals using
  prolate spheroidal wave functions.
\newblock {\em Signal Processing}, 89(12):2585--2595, 2009.

\bibitem{senay2008compressive}
S.~Senay, L.~Chaparro, M.~Sun, and R.~Sclabassi.
\newblock Compressive sensing and random filtering of {EEG} signals using
  {S}lepian basis.
\newblock In {\em Proc. European Signal Processing Conf. (EUSIPCO)}, Lausanne,
  Switzerland, Aug. 2008.

\bibitem{SlaviLDB_Compressive}
J.~P. Slavinsky, J.~Laska, M.~Davenport, and R.~Baraniuk.
\newblock The compressive multiplexer for multi-channel compressive sensing.
\newblock In {\em Proc. IEEE Int. Conf. Acoust., Speech, and Signal Processing
  (ICASSP)}, Prague, Czech Republic, May 2011.

\bibitem{Slepi_ProlateIV}
D.~Slepian.
\newblock Prolate spheroidal wave functions, {F}ourier analysis, and
  uncertainty. {IV}.
\newblock {\em Bell Systems Tech. J.}, 43(6):3009--3058, 1964.

\bibitem{Slepi_On}
D.~Slepian.
\newblock On {B}andwidth.
\newblock {\em Proc. IEEE}, 64(3):292--300, 1976.

\bibitem{Slepi_ProlateV}
D.~Slepian.
\newblock Prolate spheroidal wave functions, {F}ourier analysis, and
  uncertainty. {V} -- {T}he discrete case.
\newblock {\em Bell Systems Tech. J.}, 57(5):1371--1430, 1978.

\bibitem{Slepi_Some}
D.~Slepian.
\newblock Some comments on {F}ourier analysis, uncertainty, and modeling.
\newblock {\em SIAM Rev.}, 25(3):379--393, 1983.

\bibitem{SlepiP_ProlateI}
D.~Slepian and H.~Pollak.
\newblock Prolate spheroidal wave functions, {F}ourier analysis, and
  uncertainty. {I}.
\newblock {\em Bell Systems Tech. J.}, 40(1):43--64, 1961.

\bibitem{stark1986probability}
H.~Stark and J.~Woods.
\newblock {\em Probability, random processes, and estimation theory for
  engineers}.
\newblock Prentice-Hall, 1986.

\bibitem{Thoms_Spectrum}
D.~Thomson.
\newblock Spectrum estimation and harmonic analysis.
\newblock {\em Proc. IEEE}, 70(9):1055--1096, 1982.

\bibitem{Tikho_Solution}
A.~Tikhonov.
\newblock Solution of incorrectly formulated problems and the regularization
  method.
\newblock {\em Dokl. Akad. Nauk. SSSR}, 151:1035--1038, 1963.

\bibitem{TikhoA_Solutions}
A.~Tikhonov and V.~Arsenin.
\newblock {\em Solutions of Ill-Posed Problems}.
\newblock Winston \& Sons, Washington, D.C., 1977.

\bibitem{Tropp_Greed}
J.~Tropp.
\newblock Greed is good: {A}lgorithmic results for sparse approximation.
\newblock {\em IEEE Trans. Inform. Theory}, 50(10):2231--2242, 2004.

\bibitem{TroppLDRB_Beyond}
J.~Tropp, J.~Laska, M.~Duarte, J.~Romberg, and R.~Baraniuk.
\newblock Beyond {N}yquist: {E}fficient sampling of sparse, bandlimited
  signals.
\newblock {\em IEEE Trans. Inform. Theory}, 56(1):520--544, 2010.

\bibitem{TroppWDBB_Random}
J.~Tropp, M.~Wakin, M.~Duarte, D.~Baron, and R.~Baraniuk.
\newblock Random filters for compressive sampling and reconstruction.
\newblock In {\em Proc. IEEE Int. Conf. Acoust., Speech, and Signal Processing
  (ICASSP)}, Toulouse, France, May 2006.

\bibitem{VenkaB_Further}
R.~Venkataramani and Y.~Bresler.
\newblock Further results on spectrum blind sampling of {2D} signals.
\newblock In {\em Proc. IEEE Int. Conf. Image Processing (ICIP)}, Chicago, IL,
  Oct. 1998.

\bibitem{vershynin2010introduction}
R.~Vershynin.
\newblock Introduction to the non-asymptotic analysis of random matrices.
\newblock {\em Arxiv preprint arXiv:1011.3027}, 2010.

\bibitem{Walde_Analog}
R.~Walden.
\newblock Analog-to-digital converter survey and analysis.
\newblock {\em IEEE J. Selected Areas Comm.}, 17(4):539--550, 1999.

\end{thebibliography}

\end{document}